\ifdefined\pdfoutput\pdfoutput=1\fi
\documentclass[prx,aps,twocolumn,notitlepage,superscriptaddress,showpacs,nofootinbib]{revtex4-2}

\usepackage{amsmath, amsthm,commath}
\usepackage{graphicx}
\usepackage[dvipsnames,svgnames,table,cmyk]{xcolor}
\usepackage[colorlinks]{hyperref}
\hypersetup{
	colorlinks = true,
	urlcolor = {blue},
	citecolor = {blue},
	linkcolor= {blue}
}

\usepackage[charter,cal=cmcal,sfscaled=false]{mathdesign}
\usepackage{booktabs}
\usepackage{multirow}
\usepackage{caption}
\usepackage{mathrsfs}

\DeclareMathOperator{\Tr}{Tr}
\newcommand{\id}{\mathbb I}
\newcommand{\dd}{\,\mathrm d}
\newtheorem{theorem}{Theorem}

\newtheorem{proposition}[theorem]{Proposition}
\newtheorem{corollary}[theorem]{Corollary}

\begin{document}

\title{Resource Compatibility in Optimal Quantum Symmetry Testing}

\begin{abstract}
Can a symmetry be tested optimally without using the quantum resource naturally associated with it?
For parallel subgroup-versus-Haar testing, we derive a quantitative representation-theoretic converse for every protocol with zero type-I error: the excess type-II error is bounded below by the sum of two nonnegative penalties, one for weight on suboptimal subgroup types and one for deviations of within-type Schur profiles from their typewise extremal profiles.
At optimality, both penalties vanish, yielding support and weight locking for every optimal system marginal.
Disjointness of the resulting locked set and the free marginal set rules out resource-free optimality, whereas overlap alone does not establish attainability.
For the diagonal torus in dimension \(d\ge2\), optimal testing without coherence is possible only with a single query.
We determine both the exact optimal type-II error over all protocols with incoherent system marginals and the minimum system-marginal coherence required to attain the unrestricted optimum.
For the orthogonal subgroup in dimension \(d\ge2\), a real protocol attains the optimum at every query number.
For the qutrit Clifford group, Wigner-positive optimal protocols exist at three and four queries but not from five through nine; whether they exist beyond nine queries remains open.
Together, these results relate the subgroup- and query-dependent resource requirements of optimal symmetry testing to structural constraints on optimal system marginals.
\end{abstract}

\hypersetup{pdftitle={Resource Compatibility in Optimal Quantum Symmetry Testing},pdfauthor={Yimeng Cao, Ranyiliu Chen, Chengkai Zhu, Xin Wang}}

\author{Yimeng Cao}
\email{yimengcao728@gmail.com}
\affiliation{Thrust of Artificial Intelligence, Information Hub,\\
The Hong Kong University of Science and Technology (Guangzhou), Guangzhou 511453, China}

\author{Ranyiliu Chen}
\email{chenranyiliu@quantumsc.cn}
\affiliation{Quantum Science Center of Guangdong-Hong Kong-Macao Greater Bay Area, Shenzhen 518045, China}

\author{Chengkai Zhu}
\affiliation{QudeLeap Research, Shanghai 200030, China}

\author{Xin Wang}
\email{felixxinwang@hkust-gz.edu.cn}
\affiliation{Thrust of Artificial Intelligence, Information Hub,\\
The Hong Kong University of Science and Technology (Guangzhou), Guangzhou 511453, China}
\maketitle

\par\noindent\textbf{Introduction.}---Quantum process tomography reconstructs an unknown quantum process from
suitably chosen preparations and measurements~\cite{ChuangNielsen,PCP}, whereas
many tasks ask only whether the process has a specified structural property.
Quantum property testing formalizes this reduced objective~\cite{Montanaro2013a}.
For unitary dynamics, existing query models test structural properties of a
fixed black-box unitary, including general unitary
properties~\cite{Wang2011a,She2023}, Clifford or stabilizer
structure~\cite{Low2009Clifford,GrossNezamiWalter2021,Hinsche2026CliffordTesting},
and close-versus-far subgroup membership~\cite{unitarysubgroup1}. These settings
impose a promise on a single unknown unitary. By contrast, we distinguish two
averaged hypotheses: the same unitary is sampled once either from Haar measure
on a subgroup or from Haar measure on \(\mathrm U(d)\), and that realization is reused
in all queries. Related symmetry-testing tasks have also been studied for
Hamiltonians, states, and channels
\cite{LaBordeWilde2022,LaBordeRethinasamyWilde2023}.

The subgroup structure also connects the testing problem to quantum resource
theories~\cite{ChitambarGour2019}, which distinguish free states and operations
from resourceful ones. Symmetry restrictions in particular give rise to
resource theories of asymmetry and quantum reference frames
\cite{BartlettRudolphSpekkens2007,GourSpekkens2008}. Diagonal and real unitaries are naturally associated
with coherence~\cite{baumgratz2014quantifying,streltsov2017colloquium} and
imaginarity~\cite{RTOI2,RTOI1}, respectively. Clifford and stabilizer
structures are central to magic~\cite{Veitch2014StabilizerResource,HowardCampbell2017}; for
qutrits, a pure state has a nonnegative discrete Wigner function if and only if
it is a stabilizer state~\cite{gross2006hudson}, while Wigner negativity is a
resource for quantum computation and controls the overhead of quasiprobability
simulation~\cite{veitch2012negative,MariEisert2012,PashayanWallmanBartlett2015}.
Together with the operational role of resources in discrimination
tasks~\cite{TakagiRegula2019,NapoliEtAl2016,WuEtAl2021Operational,PianiWatrous2009,TakagiEtAl2019Subchannel}, these
connections motivate a resource-sensitive version of symmetry testing.
Attainment of the unrestricted optimum by a free protocol shows that the
associated resource is unnecessary, whereas ruling out such attainment
establishes a genuine resource requirement.

Recent work introduced finite-query hypothesis tests for prescribed symmetries
of unknown unitary dynamics~\cite{yuao}. For qubits, that work also compared a repetition strategy with protocols using correlated probes.
For parallel subgroup-versus-Haar
testing at zero type-I error, Hayashi et al.~\cite{hayashi2025predicting}
determined the optimal type-II error and proved attainability. These results
determine how well the hypotheses can be distinguished. Yet the optimal value,
even together with one protocol that attains it, neither characterizes the
system marginals of all optimizers nor quantifies how closely near-optimal
marginals satisfy the structural constraints obeyed at optimality.
This motivates the central question: can subgroup-versus-Haar
testing be performed optimally using only free probes and free
measurements?

In this work, we derive a quantitative representation-theoretic converse for
every zero-type-I-error parallel protocol. It lower-bounds the excess type-II
error by penalties for weight on suboptimal subgroup types and deviations of
within-type Schur profiles from their typewise extremal profiles, yielding exact
support and weight locking at optimality and quantitative stability near it.

Applied to the diagonal torus and orthogonal subgroup for \(d\ge2\), and to the qutrit Clifford group, the framework reveals three qualitatively different behaviors.
For the diagonal torus, optimal testing without coherence is possible only with a single query.
We also determine the exact optimal type-II error and its query scaling for protocols with incoherent system marginals, as well as the minimum system-marginal coherence required to attain the unrestricted optimum.
For the orthogonal subgroup, a real protocol \mbox{attains} the optimum at every query number.
For the qutrit Clifford group, Wigner-positive protocols attain the optimum at \(m=3,4\), whereas no optimal system marginal is Wigner-positive for \(m=5,\ldots,9\).
Because Wigner positivity is preserved under partial trace, this marginal obstruction extends to complete protocols; whether Wigner-positive optimality is possible for \(m\ge10\) remains open.
Figure~\ref{fig:resource-constrained-testing} summarizes these regimes.

\par\noindent\textbf{Unitary subgroup testing.}---Let
\(\mathcal H=\mathbb C^d\) be the single-system Hilbert space, and fix a
closed subgroup \(G\subseteq\mathrm U(d)\). Let \(\mu_G\) and
\(\mu_{\mathrm{Haar}}\) denote the normalized Haar measures on \(G\) and
\(\mathrm U(d)\), respectively. We consider the hypotheses
\begingroup
\abovedisplayskip=10pt minus 6pt
\belowdisplayskip=10pt minus 6pt
\begin{equation}
H_0:\ U\sim\mu_G,
\qquad
H_1:\ U\sim\mu_{\mathrm{Haar}}.
\label{eq:subgroup-haar-hypotheses}
\end{equation}
\endgroup

\begin{figure}[!t]
  \centering
  \includegraphics[width=\columnwidth]{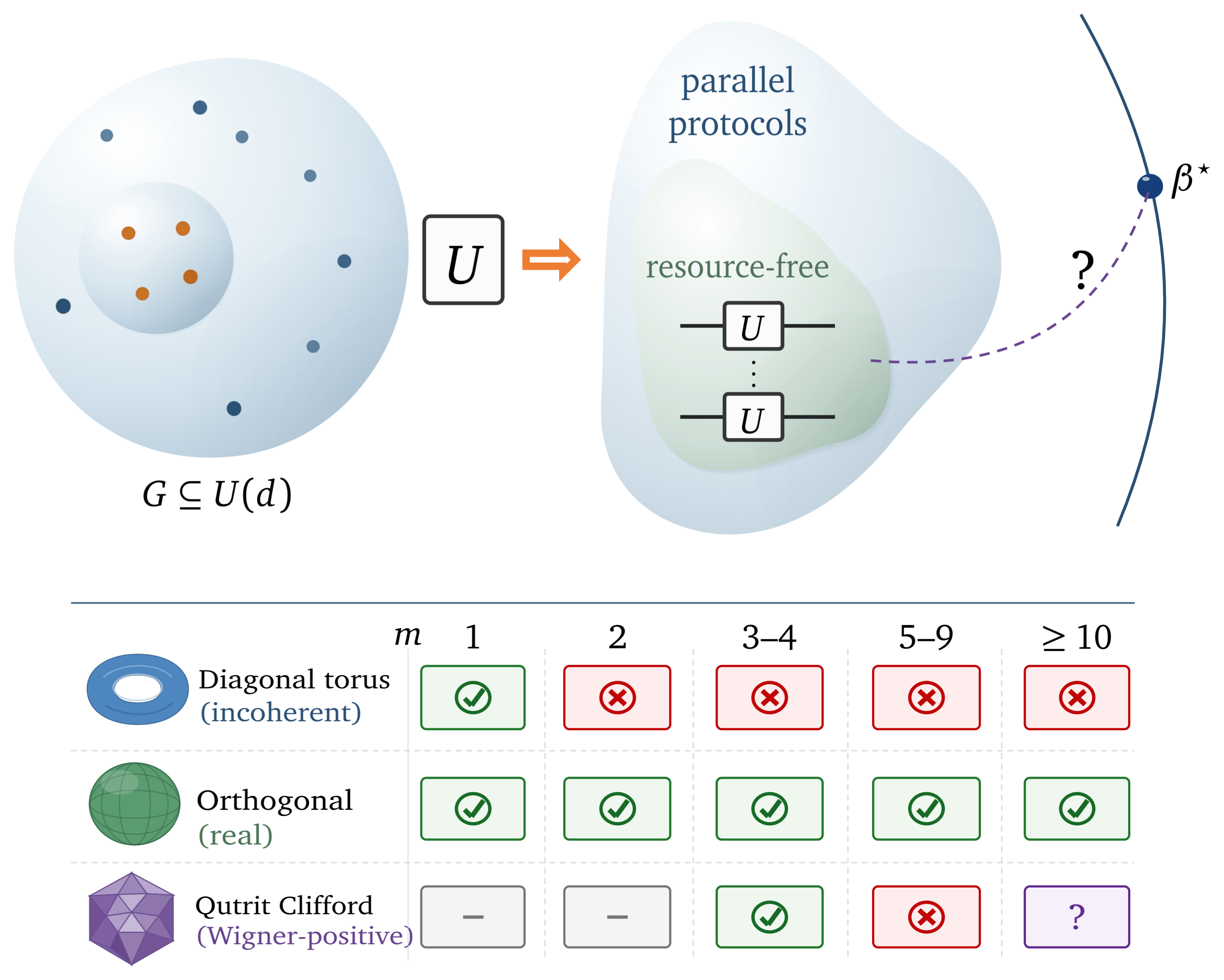}
  \caption{Resource-free unitary subgroup testing.
The upper schematic asks whether free parallel protocols attain the
unrestricted optimum \(\beta^\star\); the lower panel summarizes attainable,
ruled-out, trivial, and open query regimes for the three subgroup tests.}
  \label{fig:resource-constrained-testing}
\end{figure}

For a positive integer \(m\), the unknown unitary is sampled once under
whichever hypothesis holds, and the same realization is used in all \(m\)
parallel queries.

An \(m\)-query parallel protocol may use an arbitrary ancillary system
\(R\) with Hilbert space \(\mathcal H_R\). A protocol consists of a probe
state \(\sigma\) on
\(\mathcal H^{\otimes m}\otimes\mathcal H_R\) and a binary POVM
\(M=\{M_0,M_1\}\) on the same joint space, where outcome \(0\) accepts
\(H_0\). The probe may be mixed and entangled across the query systems and
the ancilla. The unknown unitary acts as \(U^{\otimes m}\) on the query
systems and trivially on \(R\), with \(\id_R\) denoting the identity
operator on \(\mathcal H_R\).

\begingroup
\clubpenalty=10000
Define the joint output state \(\sigma_V\) (\(V\in\mathrm U(d)\)) and the
average type-I and type-II errors by
\begin{equation}
\begin{aligned}
\sigma_V
&:=
(V^{\otimes m}\otimes\id_R)
\sigma
(V^{\otimes m}\otimes\id_R)^\dagger,
\\
\alpha(\sigma,M)
&:=
1-
\int_G
\Tr(M_0\sigma_g)
\dd\mu_G(g),
\\
\beta(\sigma,M)
&:=
\int_{\mathrm U(d)}
\Tr(M_0\sigma_U)
\dd\mu_{\mathrm{Haar}}(U).
\end{aligned}
\label{eq:testing-errors}
\end{equation}

Since \(\mu_G\) has full support on \(G\) and the acceptance probability is
continuous in the unitary, \(\alpha=0\) holds if and only if every subgroup
unitary is accepted with certainty. This is the zero-type-I-error endpoint of
binary quantum hypothesis testing~\cite{Helstrom1976,WangRenner2012}. We denote by
\(\beta_{d,m}^{\star}(0)\) the infimum of \(\beta(\sigma,M)\) over all
\(m\)-query parallel protocols satisfying this constraint. We suppress the dependence of
\(\beta_{d,m}^{\star}(0)\) and related quantities on the fixed subgroup \(G\)
throughout. The reduced probe state
\(\rho:=\Tr_R\sigma\) on the \(m\) query systems is the system marginal
studied below.
We call \(\beta^\star_{d,m}(0)\) the unrestricted optimum within the present
parallel architecture: no additional resource restriction, such as
incoherence, reality, or Wigner positivity, is imposed. A broader foundational
literature studies the optimal and perfect discrimination of unitary and
general quantum operations
\cite{Acin2001,Sacchi2005,ChiribellaDArianoPerinotti2008,DuanFengYing2009}; adaptive and indefinite-causal-order
strategies~\cite{Bavaresco_2021,Bavaresco2022} lie outside the present model.
\par
\endgroup

The query-system space \(\mathcal H^{\otimes m}\) carries commuting
tensor-power and permutation actions of \(\mathrm U(d)\) and the symmetric
group \(\mathfrak S_m\), respectively. Schur--Weyl duality~\cite{Group}, followed by
restriction of each \(\mathrm U(d)\)-irreducible component to \(G\), gives
\begin{equation}
\begin{aligned}
\mathcal H^{\otimes m}
&\cong
\bigoplus_{\substack{\lambda\vdash m\\ \ell(\lambda)\le d}}
U_\lambda\otimes S_\lambda,
\\
U_\lambda\!\downarrow_G
&\cong
\bigoplus_{\eta}
V_\eta\otimes\mathbb C^{\,n_{\eta,\lambda}}.
\end{aligned}
\label{eq:schur-weyl-branching}
\end{equation}
Here \(\lambda\) ranges over partitions of \(m\) with at most \(d\) parts,
with \(\ell(\lambda)\) denoting the number of nonzero parts. The spaces
\(U_\lambda\) and \(S_\lambda\) carry the corresponding irreducible
representations of \(\mathrm U(d)\) and \(\mathfrak S_m\), respectively.
In the second line, \(V_\eta\) is an irreducible \(G\)-module of type
\(\eta\), and \(n_{\eta,\lambda}\) is its multiplicity in
\(U_\lambda\!\downarrow_G\). We write \(d_\lambda:=\dim U_\lambda\) and
\(d_\eta:=\dim V_\eta\). A subgroup type \(\eta\) occurs at \(m\) queries
if and only if \(n_{\eta,\lambda}>0\) for at least one Schur label
\(\lambda\) appearing in Eq.~\eqref{eq:schur-weyl-branching}.

Let \(\mathcal I_m\) denote the set of subgroup types occurring at
\(m\) queries. For \(\eta\in\mathcal I_m\), define the
representation-theoretic score and its maximum by
\begin{equation}
h_m(\eta)
:=
\frac{1}{d_\eta}
\sum_\lambda d_\lambda n_{\eta,\lambda},
\qquad
h_{\max}
:=
\max_{\nu\in\mathcal I_m}h_m(\nu).
\label{eq:subgroup-type-score}
\end{equation}
The sum over \(\lambda\) runs over the Schur labels in
Eq.~\eqref{eq:schur-weyl-branching}, with \(n_{\eta,\lambda}=0\) whenever
\(V_\eta\) does not occur in \(U_\lambda\!\downarrow_G\). The maximizing
types form the subset
\(\mathcal I_m^\star
:=\bigl\{\eta\in\mathcal I_m:h_m(\eta)=h_{\max}\bigr\}\).

Hayashi et al.~\cite{hayashi2025predicting} showed that the optimal
type-II error under the zero-type-I-error constraint is
\begin{equation}
\beta_{d,m}^{\star}(0)
=
h_{\max}^{-1}.
\label{eq:optimal-type-II-error}
\end{equation}
They also established that this value is attained by an \(m\)-query
parallel protocol.

\par\noindent\textbf{Quantitative locking principle.}---To determine whether a resource
restriction is compatible with optimality, the scalar value in
Eq.~\eqref{eq:optimal-type-II-error} is not enough: one must constrain the
system marginals of optimal and near-optimal protocols. The converse below bounds the excess type-II error in terms of weight on
suboptimal types and deviations of the within-type Schur profiles from
their typewise extremal profiles. At optimality, the
resulting support and weight conditions define a convex locked marginal set.
If this set is disjoint from the free marginal set associated with a resource
restriction, no protocol with a free system marginal can be optimal; overlap
removes only this marginal obstruction, and attainability still requires an
explicit construction.

To state the bound, only the type and joint-sector weights of the system
marginal are needed. For each \(\eta\in\mathcal I_m\), let
\(\mathcal H_\eta\) be the corresponding \(G\)-isotypic subspace.
For each Schur label with \(n_{\eta,\lambda}>0\), let
\(\Pi_{\eta,\lambda}\) project onto the corresponding joint
\(G\)-type/Schur sector in Eq.~\eqref{eq:schur-weyl-branching}. Set
\[
\Lambda_\eta:=\{\lambda:n_{\eta,\lambda}>0\},
\qquad
\Pi_\eta:=\sum_{\lambda\in\Lambda_\eta}\Pi_{\eta,\lambda}.
\]
Then \(\Pi_\eta\) projects onto \(\mathcal H_\eta\).
The type weights and joint-sector weights of \(\rho\) are
\[
p_\eta:=\Tr(\Pi_\eta\rho),
\qquad
q_{\eta,\lambda}:=\Tr(\Pi_{\eta,\lambda}\rho),
\]
so that \(p_\eta=\sum_{\lambda\in\Lambda_\eta}q_{\eta,\lambda}\).
The conditional Schur profile \(\boldsymbol{\pi}_\eta\) is defined when
\(p_\eta>0\), whereas the typewise extremal profile
\(\boldsymbol{\pi}_\eta^\star\) is defined for every
\(\eta\in\mathcal I_m\):
\[
\boldsymbol{\pi}_\eta
:=
\left(
\frac{q_{\eta,\lambda}}{p_\eta}
\right)_{\lambda\in\Lambda_\eta},
\qquad
\boldsymbol{\pi}_\eta^\star
:=
\left(
\frac{d_\lambda n_{\eta,\lambda}}
     {d_\eta h_m(\eta)}
\right)_{\lambda\in\Lambda_\eta}.
\]
By Eq.~\eqref{eq:subgroup-type-score}, \(\boldsymbol{\pi}_\eta^\star\)
is a probability vector with full support on \(\Lambda_\eta\).
For probability vectors \(\mathbf a\) and \(\mathbf b\) on the same
finite set, with \(b_\lambda>0\), write
\(\chi^2(\mathbf a\Vert\mathbf b)
:=\sum_\lambda(a_\lambda-b_\lambda)^2/b_\lambda\) for the Pearson
divergence.

\begin{theorem}[Quantitative support and weight locking]
\label{thm:support-weight-locking}
Let \((\sigma,M)\) be an \(m\)-query parallel protocol satisfying
\(\alpha(\sigma,M)=0\), and let \(\rho=\Tr_R\sigma\). Then
\begin{align}
\beta(\sigma,M)-\beta^\star_{d,m}(0)
\geq{}&
\sum_{\eta:p_\eta>0}
p_\eta
\left(
\frac{1}{h_m(\eta)}
-
\frac{1}{h_{\max}}
\right)
\notag\\
&+
\sum_{\eta:p_\eta>0}
\frac{p_\eta}{h_m(\eta)}
\chi^2\!\left(
\boldsymbol{\pi}_\eta
\middle\Vert
\boldsymbol{\pi}_\eta^\star
\right).
\label{eq:quantitative-locking}
\end{align}
The two terms on the right-hand side are nonnegative. If moreover
\(\beta(\sigma,M)=\beta_{d,m}^{\star}(0)\), then
\begin{align}
\operatorname{supp}\rho
&\subseteq
\bigoplus_{\eta\in\mathcal I_m^\star}\mathcal H_\eta ,
\label{eq:support-locking}
\\
\boldsymbol{\pi}_\eta
&=
\boldsymbol{\pi}_\eta^\star ,
\label{eq:weight-locking}
\end{align}
where the second relation holds for every
\(\eta\in\mathcal I_m^\star\) with \(p_\eta>0\).
\end{theorem}

At zero excess error, Eq.~\eqref{eq:support-locking} excludes every suboptimal
type, while Eq.~\eqref{eq:weight-locking} fixes the Schur profile within each
maximizing type with \(p_\eta>0\). The weights \(p_\eta\) of distinct maximizing
types remain free.

Let \(\mathcal E_{\mathrm{lock}}\) denote the convex set of system marginals
satisfying Eqs.~\eqref{eq:support-locking} and~\eqref{eq:weight-locking}.
Geometrically, Eq.~\eqref{eq:support-locking} selects a face of the state
space, while Eq.~\eqref{eq:weight-locking} cuts out an affine slice within that
face. The free-set criterion stated above is therefore a direct consequence of
the locking conditions.

Beyond exact optimality, let
\(\varepsilon:=\beta(\sigma,M)-\beta^\star_{d,m}(0)\).
For zero-type-I-error protocols at fixed \(G,d,m\),
Eq.~\eqref{eq:quantitative-locking} bounds the total weight on suboptimal
types and the \(p_\eta\)-weighted Pearson deviation of the Schur profiles
from their typewise extremal profiles by \(O(\varepsilon)\).
Corollary~\ref{cor:supp-trace-stability} in Supplemental
Sec.~\ref{sec:supp-locking} places the system marginal within trace
distance \(O(\sqrt{\varepsilon})\) of \(\mathcal E_{\mathrm{lock}}\).
Consequently, a nonempty closed free marginal set disjoint from
\(\mathcal E_{\mathrm{lock}}\) entails a strictly positive excess-error
gap for protocols with marginals in that set.
This trace-distance estimate concerns \(\mathcal E_{\mathrm{lock}}\),
not the set of system marginals realized by optimal protocols.

The penalties depend only on \(p_\eta\) and \(q_{\eta,\lambda}\);
they do not determine the states within occupied joint sectors or require
coherences between those sectors to vanish. They do not characterize the
full protocol. A marginal no-go statement extends to protocols with free joint probes
when the relevant free property is preserved under partial trace, without
imposing a resource restriction on the POVM effects. Proofs and explicit
constants are given in Supplemental Sec.~\ref{sec:supp-locking}.

\par\noindent\textbf{Diagonal torus.}---Let
\(\mathbb T_d\subseteq\mathrm U(d)\) be the diagonal unitary subgroup,
with \(d\ge2\). Throughout this section, \(\alpha=0\).
The torus types at \(m\ge1\) queries are indexed by occupations
\(\mathbf n\in\mathbb Z_{\ge0}^d\), \(|\mathbf n|:=\sum_jn_j=m\).
The isotypic space \(W_{\mathbf n}\) is spanned by computational-basis
strings containing label \(j\) exactly \(n_j\) times and has dimension
\({N_{\mathbf n}=m!/\prod_jn_j!}\).
The corresponding irreducible torus representation is one-dimensional;
its multiplicity in \(U_\lambda\) is the weight multiplicity
\(K_{\lambda,\mathbf n}\). Write \(f^\lambda=\dim S_\lambda\).
As shown in Supplemental Sec.~\ref{sec:supp-diagonal-torus}, balanced
occupations \(\mathbf b\), whose entries differ by at most one, are exactly
the maximizing types, so \(h_m(\mathbf b)=h_{\max}\).

Let \(\Delta\) dephase all query systems in the computational product
basis; an incoherent marginal satisfies \(\Delta\rho=\rho\).
For \(\Lambda_{\mathbf n}=\{\lambda\vdash m:\ell(\lambda)\le d,
K_{\lambda,\mathbf n}>0\}\), the extremal profile of
Theorem~\ref{thm:support-weight-locking} and the incoherent profile are
\begin{equation}
\pi^\star_{\mathbf n,\lambda}
=\frac{d_\lambda K_{\lambda,\mathbf n}}{h_m(\mathbf n)},
\qquad
\pi^{\mathrm{inc}}_{\mathbf n,\lambda}
=\frac{f^\lambda K_{\lambda,\mathbf n}}{N_{\mathbf n}}.
\label{eq:torus-two-profiles}
\end{equation}
The joint projectors \(\Pi_{\mathbf n,\lambda}\) commute with query
permutations, which act transitively on strings of occupation
\(\mathbf n\); their diagonals on \(W_{\mathbf n}\) are therefore constant.
Thus every incoherent system marginal has profile
\(\boldsymbol\pi^{\mathrm{inc}}_{\mathbf n}\) in each occupied type
\((p_{\mathbf n}>0)\). Optimality instead locks the marginal to balanced
types with profiles \(\boldsymbol\pi^\star_{\mathbf b}\).
For \(m\ge2\), the two profiles differ on every balanced type, so no protocol with an
incoherent system marginal attains the unrestricted optimum.
At \(m=1\), a basis probe and
its return projection, together with the complementary effect, are all
diagonal and attain \(\beta^\star_{d,1}(0)=1/d\).

\begingroup
\displaywidowpenalty=10000
Let \(\beta^{\mathrm{inc}}_{d,m}\) be the infimum of \(\beta\) over all
\(m\)-query parallel protocols satisfying \(\alpha=0\) and
\(\Delta(\Tr_R\sigma)=\Tr_R\sigma\), with arbitrary ancilla \(R\),
joint probe otherwise unrestricted, and arbitrary measurement.
Preparing a basis string of occupation \(\mathbf n\) and projecting back
onto it gives the Haar acceptance probability
\begin{equation}
B(\mathbf n):=\int_{\mathrm U(d)}\prod_{j=1}^d|U_{jj}|^{2n_j}
\,\dd\mu_{\mathrm{Haar}}(U).
\label{eq:diagonal-occupation-error}
\end{equation}
\endgroup
\begin{theorem}[Incoherent-marginal optimum]
\label{thm:diagonal-error}
For \(d\ge2\) and \(m\ge1\), let \(\mathbf b\) be a balanced occupation. Then
\begin{equation}
\beta^{\mathrm{inc}}_{d,m}
=\min_{|\mathbf n|=m}B(\mathbf n)=B(\mathbf b).
\label{eq:diagonal-exact-optimum}
\end{equation}
The same optimum holds when the joint probe and both POVM effects are
required to be diagonal in the computational product basis tensored with a fixed ancilla basis.
It is attained without an ancilla by any computational-basis string
\(|x\rangle\in W_{\mathbf b}\), with
\(\sigma=M_0=|x\rangle\langle x|\) and \(M_1=\id-M_0\).
\end{theorem}

Evaluating the bound of Theorem~\ref{thm:support-weight-locking}
at the fixed incoherent profiles gives
\begin{equation}
\beta\ge\sum_{\mathbf n:p_{\mathbf n}>0}p_{\mathbf n}B(\mathbf n)
\ge B(\mathbf b).
\label{eq:diagonal-short-converse}
\end{equation}
See Supplemental Sec.~\ref{sec:supp-diagonal-torus} for the Schur
evaluation and the Haar balancing comparison giving the last inequality.
The torus acts on a basis string by a phase, so the return test has
\(\alpha=0\) and Haar error \(B(\mathbf b)\), proving attainment.
The qubit repetition strategy was studied in Ref.~\cite{yuao};
Theorem~\ref{thm:diagonal-error} supplies the converse for the full class
with incoherent system marginals.

The complementary question is how much system-marginal coherence is
needed to attain the unrestricted optimum.
Let \(C_{\mathrm{rel}}(\rho)=D(\rho\Vert\Delta\rho)\) denote the relative
entropy of coherence, using base-two logarithms~\cite{baumgratz2014quantifying}.
Define \(C_{\min}(d,m)\) as the infimum of
\(C_{\mathrm{rel}}(\Tr_R\sigma)\) over all \(m\)-query parallel protocols,
including arbitrary ancillas, with \(\alpha=0\) and
\(\beta=\beta^\star_{d,m}(0)\).

\begin{theorem}[Minimum marginal coherence]
\label{thm:minimum-marginal-coherence}
For every \(d\ge2\), \(m\ge1\), and balanced occupation \(\mathbf b\),
the infimum is attained~and
\begin{equation}
C_{\min}(d,m)
=D(\boldsymbol\pi^\star_{\mathbf b}\Vert
\boldsymbol\pi^{\mathrm{inc}}_{\mathbf b}).
\label{eq:minimum-marginal-coherence}
\end{equation}
This value is independent of the balanced occupation chosen, zero at
\(m=1\), and strictly positive for every \(m\ge2\).
\end{theorem}

For an optimal marginal, data processing under the joint occupation/Schur
measurement bounds the coherence below by the type-weighted average of
the profile divergences. Coordinate permutations of balanced occupations
make these divergences identical.
The explicit attaining marginal constructed in Supplemental
Sec.~\ref{sec:supp-diagonal-torus} is supported on a single balanced
occupation space. Its purification and the joint return projection give
\(\alpha=0\) and \(\beta=h_{\max}^{-1}\).
For \(m\ge2\), this marginal is invariant under the collective torus
action but has computational-product-basis coherence.

Proposition~\ref{prop:supp-diagonal-scaling} in Supplemental
Sec.~\ref{sec:supp-diagonal-torus} gives the following rates
for every fixed \(d\ge2\) as \(m\to\infty\):
\begin{equation}
\begin{gathered}
\beta^{\mathrm{inc}}_{d,m}=\Theta_d\!\left(m^{-d(d-1)/2}\right),\\
\beta^\star_{d,m}(0)=\Theta_d\!\left(m^{-d(d-1)}\right).
\end{gathered}
\label{eq:diagonal-scaling}
\end{equation}
At \(\alpha=0\) and the same target type-II error \(\delta\downarrow0\),
the minimum query number for incoherent system marginals is of the order
of the square of the unrestricted minimum, with constants depending on
\(d\).
For qubits, Proposition~\ref{prop:supp-qubit-coherence-rate} gives the
minimum-coherence rate as \(m\to\infty\):
\begin{equation}
C_{\min}(2,m)=\left(\frac23-\frac1{6\ln2}\right)m
+O(\log m)\ \mathrm{bits}.
\label{eq:qubit-coherence-rate}
\end{equation}
This minimum concerns the full query-system marginal of an exactly
optimal protocol; it does not include the total resource cost of the
joint probe and measurement.

\par\noindent\textbf{Orthogonal subgroup.}---Unlike coherence in diagonal-torus testing, imaginarity can be avoided
at every query number for the orthogonal subgroup. The parallel-testing optimization is invariant under entrywise
complex conjugation, so an optimal tester, which jointly represents the
probe and measurement, can be chosen real and then realized by a real
protocol.

Let \(\mathrm O(d)\subseteq\mathrm U(d)\) be the subgroup of real
orthogonal matrices. With the computational product basis on the query systems
and a reference basis on the ancillary system both fixed, we call a protocol
real if its probe state and both POVM effects have real
matrix entries in the resulting product basis.

\begin{theorem}[Real attainability for orthogonal-subgroup testing]
\interlinepenalty=10000\relax
For every \(d\ge2\) and \(m\ge1\), the optimum
\(\beta_{d,m}^{\star}(0)\) for
\(\mathrm O(d)\)-versus-\(\mathrm U(d)\) testing is attained by a real
\(m\)-query parallel protocol.
\label{thm:orthogonal-real-attainability}
\end{theorem}

Supplemental Sec.~\ref{sec:supp-orthogonal-real} constructs an ancilla-assisted real protocol attaining the optimum.

\par\noindent\textbf{Qutrit Clifford group.}---We take
\(\mathcal C_3\subseteq\mathrm U(3)\) to be the finite linear qutrit Clifford
group, equipped with the uniform measure. The analysis has two steps: first
identify the score-maximizing Clifford type or types, which determine the
locked marginal set, and then decide whether that set intersects the
Wigner-positive states. The relevant type family is determined by
\(m\bmod 3\); for \(m\ge5\), the unique maximizer depends only on
\(m\bmod 6\).

Since
\(\mathcal C_3\) is a unitary \(2\)-design~\cite{gross2007evenly}, the
Clifford and Haar \(m\)-fold twirls coincide for \(m=1,2\).
Consequently, \(\beta_{3,m}^{\star}(0)=1\) for \(m=1,2\), and the first
informative query number is \(m=3\). For an \(r\)-qutrit register and
\(\mathbf u=(u_1,\ldots,u_r)\in(\mathbb F_3^2)^r\), let
\(A_{\mathbf u}:=A_{u_1}\otimes\cdots\otimes A_{u_r}\) denote the
corresponding tensor-product phase-point operator~\cite{gross2006hudson}. For a
state \(\tau\) and an effect \(E\), define
\begin{equation}
W_\tau(\mathbf u)
:=
3^{-r}\Tr(A_{\mathbf u}\tau),
\qquad
W(E\mid\mathbf u)
:=
\Tr(EA_{\mathbf u}).
\label{eq:qutrit-wigner-functions}
\end{equation}
A state or effect is Wigner-positive if the corresponding function in
Eq.~\eqref{eq:qutrit-wigner-functions} is nonnegative at every
phase-space point. A binary POVM is Wigner-positive
if both effects are. A protocol with a qutrit ancillary register, possibly
trivial, is Wigner-positive if its joint probe and both POVM effects are
Wigner-positive.

Since \(\Tr A_{\mathbf u}=1\) for every \(\mathbf u\), the
always-accept measurement \(M_0=\id\) and \(M_1=0\) is
Wigner-positive; with any stabilizer probe it attains
\(\beta^\star_{3,m}(0)=1\) at \(m=1,2\).

\begin{theorem}
[Query regimes for Wigner-positive qutrit-Clifford testing]
\label{thm:clifford-wigner-threshold}
For \(\mathcal C_3\)-versus-\(\mathrm U(3)\) testing, the following hold.
\par\noindent\textup{(i)} For \(m\in\{3,4\}\), the optimum
\(\beta_{3,m}^{\star}(0)\) is attained by an ancilla-free Wigner-positive
\(m\)-query parallel protocol.
\par\noindent\textup{(ii)} For \(m=5,\ldots,9\), the system marginal
\(\rho\) of every optimal \(m\)-query parallel protocol is not
Wigner-positive. Consequently, no Wigner-positive \(m\)-query parallel
protocol attains \(\beta^\star_{3,m}(0)\).
\end{theorem}

For \(m=5,\ldots,9\), the maximizing type is unique, so locking cuts out a
single explicit affine slice, which exact phase-space witnesses separate from
the Wigner-positive state set. The protocol-level conclusion
follows because Wigner positivity is preserved under partial trace, so every
Wigner-positive probe has a Wigner-positive system marginal. Uniqueness of the
maximizing type is not by itself sufficient: the maximizer is also unique at
\(m=3\), where a Wigner-positive optimal protocol exists.

Theorem~\ref{thm:supp-clifford-finite-window-wigner} in Supplemental
Sec.~\ref{sec:supp-clifford-wigner} gives a strictly positive lower bound
on the excess type-II error of protocols with \(\alpha=0\) and
Wigner-positive system marginals at \(m=5,\ldots,9\).
Although type selection is
settled for all \(m\ge5\), the present witnesses are query-specific; beyond
nine queries, the missing ingredient is a uniform phase-space separation
rather than further type classification. Detailed proofs and exact
certificates are given in Supplemental
Secs.~\ref{sec:supp-clifford-representation}--\ref{sec:supp-certificates}.

\par\noindent\textbf{Discussion.}---The quantitative converse connects resource
compatibility to the support and Schur weights of optimal system marginals.
For the diagonal torus, the exact error with incoherent system marginals and the minimum
system-marginal coherence are two quantitative consequences of this
structure. The former identifies the full restricted-class performance
and its query scaling; the latter is attained by an explicit optimal
protocol. The orthogonal and Clifford applications exhibit \mbox{different
mechanisms}: conjugation symmetry permits real attainment, whereas
phase-space witnesses exclude Wigner-positive optimal marginals in the
specified query window and yield a strict finite-window performance gap.

\begingroup
\clubpenalty=10000
Resource requirements on the system marginal established at zero type-I error need not persist when nonzero type-I error is allowed.
For example, in the two-query qubit torus test, an incoherent probe with a nondiagonal measurement attains the unrestricted optimum at certain positive average type-I-error tolerances, whereas protocols with diagonal joint probes and measurement effects remain suboptimal at the same tolerances.
Propositions~\ref{prop:supp-fully-diagonal-ROC} and~\ref{prop:supp-qubit-measurement-example} in Supplemental Sec.~\ref{sec:supp-diagonal-torus} give the precise comparison and constructions.
\par
\endgroup

\begingroup
\looseness=1
Within the present parallel architecture, the converse applies to every closed
subgroup of \(\mathrm U(d)\), but turning it into a resource statement requires
subgroup-specific information about the maximizing types and the free set.
At zero type-I error and without resource restrictions, parallel
protocols attain the same optimal type-II error as adaptive and
indefinite-causal-order strategies~\cite{hayashi2025predicting}, but
this equality does not imply equal resource requirements. The present locking
conditions concern parallel probes; resource-constrained attainability in
the larger strategy classes is not addressed. In
higher odd-prime Clifford dimensions,
the qutrit analysis suggests a concrete program: determine the type structure,
identify the resulting locked marginals, and construct either phase-space
witnesses or free optimal protocols. The qutrit conclusions do not
automatically extend, because both type selection and the separating witnesses
are dimension-dependent.
\par\endgroup

\begingroup
\looseness=1
Several questions remain. The two torus optima do not determine the
performance at a general coherence budget or the combined resource cost
of the joint probe and measurement. The qubit coherence rate applies to
exact optimizers; a uniform rate for protocols with small relative excess
type-II error requires further analysis. Proposition~\ref{prop:supp-finite-alpha} in
Supplemental Sec.~\ref{sec:supp-locking} extends the marginal converse to
nonzero type-I error, but does not by itself determine the general
resource-restricted error tradeoff.
Other open problems include uniform qutrit phase-space separation beyond
nine queries, sufficient conditions for a locked marginal to be realizable
by an optimal protocol, constraints on optimal measurement effects, and
stability relative to the actual optimal-marginal~set.
\par
\endgroup

\begingroup
\makeatletter
\patchcmd{\bibsection}
  {\addvspace{19\p@}}{\addvspace{15\p@}}
  {}{\PackageError{local-layout}{Unexpected bibliography heading spacing}{}}
\edef\textbalancetolerance{\the\dimexpr\baselineskip\relax}
\let\savedtextbalance\balance@two
\patchcmd{\balance@two}
  {\dimen@ii<.5\p@}{\dimen@ii<\textbalancetolerance\relax}
  {}{\PackageError{local-layout}{Unexpected REVTeX balance routine}{}}
\patchcmd{\balance@two}
  {\dimen@ii>-.5\p@}{\dimen@ii>-\textbalancetolerance\relax}
  {}{\PackageError{local-layout}{Unexpected REVTeX balance routine}{}}
\makeatother

\onecolumngrid
\clearpage
\endgroup

\begin{center}
\textbf{\large Supplemental Material}
\end{center}

\setcounter{section}{0}
\renewcommand{\thesection}{\Roman{section}}
\renewcommand{\thesubsection}{\Alph{subsection}}
\renewcommand{\theHsection}{supp.\arabic{section}}
\renewcommand{\theHsubsection}{supp.\arabic{section}.\arabic{subsection}}

\setcounter{equation}{0}
\renewcommand{\theequation}{S\arabic{equation}}
\renewcommand{\theHequation}{supp.\arabic{equation}}

\setcounter{theorem}{0}
\renewcommand{\thetheorem}{S\arabic{theorem}}
\renewcommand{\theHtheorem}{supp.\arabic{theorem}}

\section{Quantitative support and weight locking}
\label{sec:supp-locking}

We retain the testing framework and notation of the main text. For each
\(\eta\in\mathcal I_m\), let \(\mathcal H_\eta\) be the corresponding
\(G\)-isotypic subspace. Under the fixed Schur--Weyl and branching
decompositions, define the set of active Schur labels by
\[
\Lambda_\eta:=\{\lambda:n_{\eta,\lambda}>0\}.
\]
Let \(\Pi_{\eta,\lambda}\) project onto the joint
\(G\)-type/Schur sector
\(V_\eta\otimes\mathbb C^{\,n_{\eta,\lambda}}\otimes S_\lambda\), and
let \(\Pi_\eta:=\sum_{\lambda\in\Lambda_\eta}\Pi_{\eta,\lambda}\)
project onto \(\mathcal H_\eta\). For a state \(\rho\) on
\(\mathcal H^{\otimes m}\), set
\[
p_\eta:=\Tr(\Pi_\eta\rho),
\qquad
q_{\eta,\lambda}:=\Tr(\Pi_{\eta,\lambda}\rho).
\]
Thus
\(p_\eta=\sum_{\lambda\in\Lambda_\eta}q_{\eta,\lambda}\). Define the
Schur profile (for \(p_\eta>0\)) and the typewise extremal profile by
\[
\boldsymbol{\pi}_\eta
:=
\left(
\frac{q_{\eta,\lambda}}{p_\eta}
\right)_{\lambda\in\Lambda_\eta},
\qquad
\boldsymbol{\pi}_\eta^\star
:=
\left(
\frac{d_\lambda n_{\eta,\lambda}}
     {d_\eta h_m(\eta)}
\right)_{\lambda\in\Lambda_\eta}.
\]
Write \(\pi_{\eta,\lambda}\) and
\(\pi^\star_{\eta,\lambda}\) for the corresponding components.
Equation~\eqref{eq:subgroup-type-score} shows that
\(\boldsymbol{\pi}_\eta^\star\) is a probability vector with full
support on \(\Lambda_\eta\). For probability vectors \(\mathbf a\)
and \(\mathbf b\) on the same finite set, with \(b_\lambda>0\), write
\(\chi^2(\mathbf a\Vert\mathbf b)
:=\sum_\lambda(a_\lambda-b_\lambda)^2/b_\lambda\).
We use the convention that
\(\eta\notin\mathcal I_m\) implies \(\mathcal H_\eta=\{0\}\), and
hence \(p_\eta=0\). We call \(\eta\) occupied when \(p_\eta>0\). Every occupied
type then belongs to \(\mathcal I_m\), and hence \(h_m(\eta)>0\).

\begin{theorem}[Quantitative support and weight locking]
\label{thm:supp-support-weight-locking}
Let \((\sigma,M)\) be an \(m\)-query parallel protocol satisfying
\(\alpha(\sigma,M)=0\), and set \(\rho:=\Tr_R\sigma\). Then
\begin{align}
\beta(\sigma,M)-\beta^\star_{d,m}(0)
\geq{}&
\sum_{\eta:p_\eta>0}
p_\eta
\left(
\frac{1}{h_m(\eta)}
-
\frac{1}{h_{\max}}
\right)
+
\sum_{\eta:p_\eta>0}
\frac{p_\eta}{h_m(\eta)}
\chi^2\!\left(
\boldsymbol{\pi}_\eta
\middle\Vert
\boldsymbol{\pi}_\eta^\star
\right).
\label{eq:supp-global-bound}
\end{align}
If moreover
\(\beta(\sigma,M)=\beta_{d,m}^{\star}(0)=h_{\max}^{-1}\), then
\begin{align}
\operatorname{supp}\rho
&\subseteq
\bigoplus_{\eta\in\mathcal I_m^\star}\mathcal H_\eta ,
\label{eq:supp-support-locking}
\\
\frac{q_{\eta,\lambda}}{p_\eta}
&=
\frac{d_\lambda n_{\eta,\lambda}}{d_\eta h_{\max}} ,
\label{eq:supp-weight-locking}
\end{align}
where the second relation holds for every
\(\eta\in\mathcal I_m^\star\) with \(p_\eta>0\) and every
\(\lambda\in\Lambda_\eta\).
\end{theorem}

The proof uses the following operator bound.

\emph{Positive-operator partial-trace bound.}
If \(X\succeq0\) acts on \(\mathbb C^n\otimes\mathcal K\), let
\(E_i:=|i\rangle\langle i|\otimes\id_{\mathcal K}\),
let \(\zeta:=e^{2\pi\mathrm{i}/n}\), and set
\(W:=\sum_{j=1}^n\zeta^jE_j\).
Each conjugate \(W^kXW^{-k}\) is positive, the \(k=0\) term equals
\(X\), and summing over the phases removes the off-diagonal blocks.
Hence
\begin{equation}
X
\preceq
\sum_{k=0}^{n-1}W^kXW^{-k}
=
n\sum_{i=1}^nE_iXE_i
\preceq
n\,\id_{\mathbb C^n}
\otimes
\Tr_{\mathbb C^n}X.
\label{eq:supp-partial-trace-bound}
\end{equation}
For the last inequality, write
\[
X_{ii}
:=
(\langle i|\otimes\id_{\mathcal K})
X
(|i\rangle\otimes\id_{\mathcal K}).
\]
Then
\(E_iXE_i=|i\rangle\langle i|\otimes X_{ii}\) and
\[
X_{ii}
\preceq
\sum_jX_{jj}
=
\Tr_{\mathbb C^n}X.
\]

\begin{proof}
\emph{Twirling.}
For operators on the query systems and ancilla, define
\begin{equation}
\begin{split}
\mathcal T_G^{(R)}(X)
&:=
\int_G
(g^{\otimes m}\otimes\id_R)
X
(g^{\otimes m}\otimes\id_R)^\dagger
\,\dd\mu_G(g),
\\
\mathcal T_{\mathrm U(d)}^{(R)}(X)
&:=
\int_{\mathrm U(d)}
(U^{\otimes m}\otimes\id_R)
X
(U^{\otimes m}\otimes\id_R)^\dagger
\,\dd\mu_{\mathrm{Haar}}(U).
\end{split}
\label{eq:supp-G-twirl}
\end{equation}
Set
\(\bar\sigma:=\mathcal T_G^{(R)}(\sigma)\) and
\(\bar\rho:=\Tr_R\bar\sigma\). Since the twirl acts trivially on the
ancilla,
\[
\bar\rho
=
\int_G
g^{\otimes m}\rho(g^{\otimes m})^\dagger
\,\dd\mu_G(g).
\]
Invariance of the two measures gives
\[
\mathcal T_G^{(R)}\circ\mathcal T_G^{(R)}
=
\mathcal T_G^{(R)},
\qquad
\mathcal T_{\mathrm U(d)}^{(R)}
\circ
\mathcal T_G^{(R)}
=
\mathcal T_{\mathrm U(d)}^{(R)}.
\]
The first identity follows from invariance of \(\mu_G\), while the second
follows from the right invariance of the ambient Haar measure under
\(U\mapsto Ug\). Hence the type-I and type-II errors are unchanged:
\begin{equation}
\alpha(\bar\sigma,M)=\alpha(\sigma,M),
\qquad
\beta(\bar\sigma,M)=\beta(\sigma,M).
\label{eq:supp-error-preservation}
\end{equation}
The twirl is used only as a proof device; no invariance assumption is
made about the original probe \(\sigma\).

\emph{Block decomposition and profile identification.}
Write
\[
\mathcal M_\eta
:=
\bigoplus_{\lambda\in\Lambda_\eta}
\mathbb C^{\,n_{\eta,\lambda}}\otimes S_\lambda,
\]
so that
\(\mathcal H_\eta\otimes\mathcal H_R
\cong
V_\eta\otimes\mathcal M_\eta\otimes\mathcal H_R\).
Schur's lemma gives
\[
\bar\sigma
=
\bigoplus_{\eta\in\mathcal I_m}
\frac{\id_{V_\eta}}{d_\eta}\otimes\Theta_\eta,
\]
where \(\Theta_\eta\succeq0\) acts on
\(\mathcal M_\eta\otimes\mathcal H_R\).
Since \(\Pi_\eta\) commutes with \(g^{\otimes m}\) (it is the
isotypic projector for the \(G\)-action), the block decomposition gives
\[
\Tr\Theta_\eta
=
\Tr(\Pi_\eta\bar\rho)
=
\Tr(\Pi_\eta\rho)
=
p_\eta.
\]
Normalization of \(\bar\sigma\) gives
\[
\sum_{\eta\in\mathcal I_m}p_\eta=1.
\]
For each occupied type \(\eta\), define
\[
\tau_\eta:=\frac{\Theta_\eta}{p_\eta},
\qquad
\bar\sigma_\eta
:=
\frac{\id_{V_\eta}}{d_\eta}\otimes\tau_\eta.
\]
Then \(\tau_\eta\) and \(\bar\sigma_\eta\) have unit trace and
\(\bar\sigma=\sum_{\eta:p_\eta>0}p_\eta\bar\sigma_\eta\).
For \(\lambda\in\Lambda_\eta\), let \(P_{\eta,\lambda}\) project
\(\mathcal M_\eta\) onto
\(\mathbb C^{\,n_{\eta,\lambda}}\otimes S_\lambda\), and set
\(\widetilde P_{\eta,\lambda}:=P_{\eta,\lambda}\otimes\id_R\).
The joint projector \(\Pi_{\eta,\lambda}\) commutes with
\(g^{\otimes m}\), so
\begin{equation}
\Tr(\Pi_{\eta,\lambda}\bar\rho)
=
\Tr(\Pi_{\eta,\lambda}\rho)
=
q_{\eta,\lambda}.
\label{eq:supp-projector-weight-preservation}
\end{equation}
For \(p_\eta>0\), the block decomposition therefore gives
\[
\Tr(\widetilde P_{\eta,\lambda}\tau_\eta)
=
\frac{q_{\eta,\lambda}}{p_\eta}
=
\pi_{\eta,\lambda}.
\]
This identification uses neither the zero-type-I-error condition nor
optimality.

\emph{Consequence of zero type-I error.}
The zero-type-I-error condition and \(G\)-invariance of \(\bar\sigma\) give
\[
1
=
\Tr(M_0\bar\sigma)
=
\sum_{\eta:p_\eta>0}
p_\eta\Tr(M_0\bar\sigma_\eta).
\]
Because the occupied weights are positive and sum to one and each acceptance
probability lies in \([0,1]\), this equality forces
\(\Tr(M_0\bar\sigma_\eta)=1\) for every occupied \(\eta\).
Let \(Q_\eta\) be the orthogonal projector onto
\(\operatorname{supp}\tau_\eta\subseteq
\mathcal M_\eta\otimes\mathcal H_R\).
Since \(\id-M_0\succeq0\), acceptance with probability one implies that
\(\id-M_0\) annihilates
\(V_\eta\otimes\operatorname{supp}\tau_\eta\). Under the fixed block
identification,
\begin{equation}
(\Pi_\eta\otimes\id_R)
M_0
(\Pi_\eta\otimes\id_R)
\succeq
\id_{V_\eta}\otimes Q_\eta
\qquad
(p_\eta>0).
\label{eq:supp-block-acceptance}
\end{equation}
This is the only point in the argument at which the condition
\(\alpha(\sigma,M)=0\) is used.

\emph{Typewise quantitative lower bound.}
Fix an occupied type \(\eta\). Define, for
\(\lambda\in\Lambda_\eta\),
\begin{equation}
\begin{aligned}
\widetilde\tau_{\eta,\lambda}
&:=
\widetilde P_{\eta,\lambda}
\tau_\eta
\widetilde P_{\eta,\lambda},
&
\widehat\tau_{\eta,\lambda}
&:=
\Tr_{\mathbb C^{\,n_{\eta,\lambda}}}
\widetilde\tau_{\eta,\lambda}.
\end{aligned}
\label{eq:supp-conditional-blocks}
\end{equation}
The ambient Haar twirl annihilates off-diagonal Schur blocks and
depolarizes the \(U_\lambda\) factor in each diagonal block:
\[
\mathcal T_{\mathrm U(d)}^{(R)}(\bar\sigma_\eta)
=
\sum_{\lambda\in\Lambda_\eta}
\frac{\id_{U_\lambda}}{d_\lambda}
\otimes
\widehat\tau_{\eta,\lambda}.
\]
Because
\(\beta(\bar\sigma_\eta,M)
=\Tr[M_0\mathcal T_{\mathrm U(d)}^{(R)}(\bar\sigma_\eta)]\)
and \(\mathcal T_{\mathrm U(d)}^{(R)}(\bar\sigma_\eta)\) commutes with
\(\Pi_\eta\otimes\id_R\), discarding the nonnegative contribution outside
the \(\eta\)-block and applying
Eq.~\eqref{eq:supp-block-acceptance} gives
\begin{equation}
\begin{aligned}
\beta_\eta
&:=
\beta(\bar\sigma_\eta,M)
\\
&\ge
\sum_{\lambda\in\Lambda_\eta}
\frac{d_\eta}{d_\lambda}
\Tr\!\left[
Q_\eta
\left(
\id_{\mathbb C^{\,n_{\eta,\lambda}}}
\otimes
\widehat\tau_{\eta,\lambda}
\right)
\right].
\end{aligned}
\label{eq:supp-typewise-prebound}
\end{equation}

Since \(Q_\eta\tau_\eta^{1/2}=\tau_\eta^{1/2}\), the Cauchy--Schwarz
inequality for the Hilbert--Schmidt inner product gives
\[
\begin{aligned}
\pi_{\eta,\lambda}^{\,2}
&=
\left|
\Tr\!\left[
(Q_\eta\tau_\eta^{1/2})^\dagger
Q_\eta\widetilde P_{\eta,\lambda}
\tau_\eta^{1/2}
\right]
\right|^2
\\
&\le
\Tr(Q_\eta\widetilde\tau_{\eta,\lambda}).
\end{aligned}
\]
Applying the positive-operator partial-trace bound in
Eq.~\eqref{eq:supp-partial-trace-bound} to
\(\widetilde\tau_{\eta,\lambda}\), with
\(\mathcal K=S_\lambda\otimes\mathcal H_R\), yields
\[
\pi_{\eta,\lambda}^{\,2}
\le
n_{\eta,\lambda}
\Tr\!\left[
Q_\eta
\left(
\id_{\mathbb C^{\,n_{\eta,\lambda}}}
\otimes
\widehat\tau_{\eta,\lambda}
\right)
\right].
\]
Substitution into Eq.~\eqref{eq:supp-typewise-prebound} gives
\begin{equation}
\beta_\eta
\ge
d_\eta
\sum_{\lambda\in\Lambda_\eta}
\frac{\pi_{\eta,\lambda}^{\,2}}
{d_\lambda n_{\eta,\lambda}}.
\label{eq:supp-typewise-bound}
\end{equation}
Since
\(d_\lambda n_{\eta,\lambda}
=d_\eta h_m(\eta)\pi^\star_{\eta,\lambda}\) and
\[
\sum_{\lambda\in\Lambda_\eta}
\frac{\pi_{\eta,\lambda}^{\,2}}
{\pi^\star_{\eta,\lambda}}
=
1+
\chi^2\!\left(
\boldsymbol{\pi}_\eta
\middle\Vert
\boldsymbol{\pi}_\eta^\star
\right),
\]
Eq.~\eqref{eq:supp-typewise-bound} yields
\begin{equation}
\beta_\eta
\ge
\frac{1}{h_m(\eta)}
\left[
1+
\chi^2\!\left(
\boldsymbol{\pi}_\eta
\middle\Vert
\boldsymbol{\pi}_\eta^\star
\right)
\right].
\label{eq:supp-typewise-quantitative}
\end{equation}

\emph{Global quantitative bound.}
Linearity of the type-II error,
Eq.~\eqref{eq:supp-error-preservation},
\(\sum_{\eta:p_\eta>0}p_\eta=1\), and
\(\beta_{d,m}^{\star}(0)=h_{\max}^{-1}\) give
\begin{align*}
\beta(\sigma,M)-\beta_{d,m}^{\star}(0)
&=
\sum_{\eta:p_\eta>0}
p_\eta
\left(
\beta_\eta-\frac{1}{h_{\max}}
\right)
\\
&\ge
\sum_{\eta:p_\eta>0}
p_\eta
\left(
\frac{1}{h_m(\eta)}-\frac{1}{h_{\max}}
\right)
+
\sum_{\eta:p_\eta>0}
\frac{p_\eta}{h_m(\eta)}
\chi^2\!\left(
\boldsymbol{\pi}_\eta
\middle\Vert
\boldsymbol{\pi}_\eta^\star
\right),
\end{align*}
which proves Eq.~\eqref{eq:supp-global-bound}. Both sums on the right
are nonnegative. Dropping them recovers the converse bound
\(\beta(\sigma,M)\ge h_{\max}^{-1}\).

\emph{Exact locking as the zero-gap case.}
Suppose
\(\beta(\sigma,M)=\beta_{d,m}^{\star}(0)\). The left-hand side of
Eq.~\eqref{eq:supp-global-bound} is zero, so both nonnegative sums
vanish. For every suboptimal type,
\(h_m(\eta)^{-1}-h_{\max}^{-1}>0\), and hence \(p_\eta=0\).
Since \(\rho\succeq0\),
\[
0
=
\Tr(\Pi_\eta\rho)
=
\|\Pi_\eta\rho^{1/2}\|_2^2,
\]
so \(\operatorname{supp}\rho\) is orthogonal to every suboptimal
\(\mathcal H_\eta\). This proves
Eq.~\eqref{eq:supp-support-locking}.

For every occupied maximizing type, the second sum can vanish only if
\[
\chi^2\!\left(
\boldsymbol{\pi}_\eta
\middle\Vert
\boldsymbol{\pi}_\eta^\star
\right)
=0,
\]
and hence
\(\boldsymbol{\pi}_\eta=\boldsymbol{\pi}_\eta^\star\).
Because \(h_m(\eta)=h_{\max}\), the componentwise equality is exactly
Eq.~\eqref{eq:supp-weight-locking}.
\end{proof}

\emph{Near-optimal consequences.}
For a zero-type-I-error protocol, define
\[
\varepsilon
:=
\beta(\sigma,M)-\beta_{d,m}^{\star}(0).
\]
If \(\mathcal I_m\neq\mathcal I_m^\star\), let
\[
h_{\mathrm{sub}}
:=
\max_{\eta\in\mathcal I_m\setminus\mathcal I_m^\star}h_m(\eta).
\]
Since
\(h_{\mathrm{sub}}^{-1}-h_{\max}^{-1}>0\), the first penalty in
Eq.~\eqref{eq:supp-global-bound} gives
\begin{equation}
\sum_{\eta\in\mathcal I_m\setminus\mathcal I_m^\star}p_\eta
\leq
\frac{\varepsilon}
{h_{\mathrm{sub}}^{-1}-h_{\max}^{-1}}.
\label{eq:supp-support-leakage}
\end{equation}
The second penalty and \(h_m(\eta)\le h_{\max}\) give
\begin{equation}
\sum_{\eta:p_\eta>0}
p_\eta
\chi^2\!\left(
\boldsymbol{\pi}_\eta
\middle\Vert
\boldsymbol{\pi}_\eta^\star
\right)
\leq
h_{\max}\varepsilon.
\label{eq:supp-profile-stability}
\end{equation}
If \(\mathcal I_m=\mathcal I_m^\star\), no suboptimal type occurs,
so the support leakage vanishes identically.
For probability vectors,
\(\|\mathbf a-\mathbf b\|_1\le
\sqrt{\chi^2(\mathbf a\Vert\mathbf b)}\). Therefore,
\[
\sum_{\eta:p_\eta>0}
p_\eta
\left\|
\boldsymbol{\pi}_\eta-\boldsymbol{\pi}_\eta^\star
\right\|_1
\leq
\sqrt{h_{\max}\varepsilon}.
\]
Here \(\|\cdot\|_1\) denotes the \(\ell_1\) norm; the averaged total-variation distance is therefore at most
\(\frac12\sqrt{h_{\max}\varepsilon}\). Thus support leakage is
\(O(\varepsilon)\), whereas the averaged profile deviation is
\(O(\sqrt{\varepsilon})\) in total variation.

\begin{corollary}[Trace-distance stability of the locked marginal set]
\label{cor:supp-trace-stability}
Fix \(G,d,m\), and let \((\sigma,M)\) be an \(m\)-query parallel
protocol satisfying \(\alpha(\sigma,M)=0\).
Set \(\rho=\Tr_R\sigma\) and
\(\varepsilon=\beta(\sigma,M)-\beta^\star_{d,m}(0)\).
For the locked marginal set \(\mathcal E_{\mathrm{lock}}\)
defined by Eqs.~\eqref{eq:supp-support-locking}
and~\eqref{eq:supp-weight-locking},
\begin{equation}
\label{eq:supp-trace-stability}
\min_{\omega\in\mathcal E_{\mathrm{lock}}}
\frac12\|\rho-\omega\|_1
\le \sqrt{C_{d,m}\varepsilon},
\end{equation}
where
\begin{equation}
\label{eq:supp-trace-constant}
C_{d,m}:=
\begin{cases}
\max\!\left\{
h_{\max},
\bigl(h_{\mathrm{sub}}^{-1}-h_{\max}^{-1}\bigr)^{-1}
\right\},
&\mathcal I_m\ne\mathcal I_m^\star,\\[1mm]
h_{\max},
&\mathcal I_m=\mathcal I_m^\star.
\end{cases}
\end{equation}
Here \(h_{\mathrm{sub}}\) is defined as above when suboptimal
types occur, and the dependence of \(C_{d,m}\) on the fixed
subgroup \(G\) is suppressed.
\end{corollary}

\begin{proof}
For density operators, the locking conditions in
Eqs.~\eqref{eq:supp-support-locking}
and~\eqref{eq:supp-weight-locking} are equivalent to
\begin{equation}
\label{eq:supp-locked-linear}
\begin{gathered}
\Tr(\Pi_\eta\omega)=0
\qquad
(\eta\in\mathcal I_m\setminus\mathcal I_m^\star),
\\
\Tr(\Pi_{\eta,\lambda}\omega)
=\pi^\star_{\eta,\lambda}\Tr(\Pi_\eta\omega)
\qquad
(\eta\in\mathcal I_m^\star,\ \lambda\in\Lambda_\eta).
\end{gathered}
\end{equation}
Indeed, positivity turns zero type weight into the support
condition, while \(\Pi_{\eta,\lambda}\preceq\Pi_\eta\) makes the
second equality automatic when \(\Tr(\Pi_\eta\omega)=0\).
These constraints are affine in \(\omega\), so
\(\mathcal E_{\mathrm{lock}}\) is a closed convex subset of the compact
state space.
It is nonempty: choose a maximizing type and mix states in its
active joint sectors with weights \(\pi^\star_{\eta,\lambda}\).

For an arbitrary system state \(\rho\), define
\begin{equation}
\label{eq:supp-locking-deviations}
\delta_{\mathrm{leak}}
:=\sum_{\eta\in\mathcal I_m\setminus\mathcal I_m^\star}p_\eta,
\qquad
\delta_{\mathrm{prof}}
:=\sum_{\substack{\eta\in\mathcal I_m^\star\\p_\eta>0}}
p_\eta
\chi^2(\boldsymbol{\pi}_\eta\Vert\boldsymbol{\pi}_\eta^\star).
\end{equation}
These quantities measure the weight on suboptimal types and the
weighted Pearson deviation within maximizing types, respectively.
We first establish the geometric estimate
\begin{equation}
\label{eq:supp-geometric-trace-bound}
\min_{\omega\in\mathcal E_{\mathrm{lock}}}
\frac12\|\rho-\omega\|_1
\le
\sqrt{\delta_{\mathrm{leak}}+\delta_{\mathrm{prof}}}.
\end{equation}
If \(\delta_{\mathrm{leak}}=1\), then
\(\delta_{\mathrm{prof}}=0\), and the distance to the nonempty
locked set is at most one, proving
Eq.~\eqref{eq:supp-geometric-trace-bound} in this case.
Suppose \(\delta_{\mathrm{leak}}<1\), and choose a purification
\(|\Psi\rangle\in
\mathcal H^{\otimes m}\otimes\mathcal H_{R'}\)
of \(\rho\). For \(\eta\in\mathcal I_m\) and
\(\lambda\in\Lambda_\eta\), set
\[
|v_{\eta,\lambda}\rangle
=\frac{(\Pi_{\eta,\lambda}\otimes\id_{R'})|\Psi\rangle}
{\sqrt{q_{\eta,\lambda}}}
\]
when \(q_{\eta,\lambda}>0\). Otherwise, choose any unit vector
in \(\operatorname{Ran}(\Pi_{\eta,\lambda}\otimes\id_{R'})\),
which is nonzero for an active joint sector.
These vectors are orthonormal, and
\[
|\Psi\rangle
=\sum_{\eta\in\mathcal I_m}\sum_{\lambda\in\Lambda_\eta}
\sqrt{q_{\eta,\lambda}}\,|v_{\eta,\lambda}\rangle.
\]
Define
\begin{equation}
\label{eq:supp-locked-comparison}
|\Phi\rangle
=\sum_{\eta\in\mathcal I_m^\star}
\sum_{\lambda\in\Lambda_\eta}
\sqrt{
\frac{p_\eta\pi^\star_{\eta,\lambda}}
{1-\delta_{\mathrm{leak}}}
}\,|v_{\eta,\lambda}\rangle,
\qquad
\omega=\Tr_{R'}|\Phi\rangle\langle\Phi|.
\end{equation}
Since
\(\sum_{\eta\in\mathcal I_m^\star}p_\eta
=1-\delta_{\mathrm{leak}}\),
the vector \(|\Phi\rangle\) in
Eq.~\eqref{eq:supp-locked-comparison} is normalized.
Its marginal is supported on maximizing types and satisfies
\[
\Tr(\Pi_{\eta,\lambda}\omega)
=\frac{p_\eta\pi^\star_{\eta,\lambda}}
{1-\delta_{\mathrm{leak}}},
\qquad
\Tr(\Pi_\eta\omega)
=\frac{p_\eta}{1-\delta_{\mathrm{leak}}}
\quad
(\eta\in\mathcal I_m^\star,\ \lambda\in\Lambda_\eta).
\]
The resulting weights satisfy Eq.~\eqref{eq:supp-locked-linear},
so \(\omega\in\mathcal E_{\mathrm{lock}}\), also when a
maximizing type has \(p_\eta=0\).

For occupied maximizing types, let
\[
\varphi_\eta
:=\sum_{\lambda\in\Lambda_\eta}
\sqrt{\pi_{\eta,\lambda}\pi^\star_{\eta,\lambda}},
\qquad
a:=\sum_{\substack{\eta\in\mathcal I_m^\star\\p_\eta>0}}
p_\eta\varphi_\eta.
\]
The construction gives
\(\langle\Psi|\Phi\rangle
=a/\sqrt{1-\delta_{\mathrm{leak}}}\).
Normalization of the profiles and
\(\pi^\star_{\eta,\lambda}>0\) imply
\[
\begin{aligned}
2(1-\varphi_\eta)
&=\sum_{\lambda\in\Lambda_\eta}
\left(
\sqrt{\pi_{\eta,\lambda}}-\sqrt{\pi^\star_{\eta,\lambda}}
\right)^2\\
&=\sum_{\lambda\in\Lambda_\eta}
\frac{(\pi_{\eta,\lambda}-\pi^\star_{\eta,\lambda})^2}
{\left(
\sqrt{\pi_{\eta,\lambda}}+\sqrt{\pi^\star_{\eta,\lambda}}
\right)^2}
\le
\chi^2(\boldsymbol{\pi}_\eta\Vert\boldsymbol{\pi}_\eta^\star).
\end{aligned}
\]
Using
\(\bigl(a-(1-\delta_{\mathrm{leak}})\bigr)^2\ge0\), we obtain
\[
\begin{aligned}
1-|\langle\Psi|\Phi\rangle|^2
&=1-\frac{a^2}{1-\delta_{\mathrm{leak}}}\\
&\le
\delta_{\mathrm{leak}}
+2\bigl((1-\delta_{\mathrm{leak}})-a\bigr)\\
&=
\delta_{\mathrm{leak}}
+2\sum_{\substack{\eta\in\mathcal I_m^\star\\p_\eta>0}}
p_\eta(1-\varphi_\eta)\\
&\le
\delta_{\mathrm{leak}}+\delta_{\mathrm{prof}}.
\end{aligned}
\]
The pure-state trace-distance formula and contractivity under
partial trace therefore give
\[
\frac12\|\rho-\omega\|_1
\le
\sqrt{1-|\langle\Psi|\Phi\rangle|^2}
\le
\sqrt{\delta_{\mathrm{leak}}+\delta_{\mathrm{prof}}}.
\]
This proves Eq.~\eqref{eq:supp-geometric-trace-bound}.

Now let \(\rho\) be the marginal of the stated protocol.
The deviations in Eq.~\eqref{eq:supp-locking-deviations}
obey a common error budget.
If \(\mathcal I_m\ne\mathcal I_m^\star\), retaining the
suboptimal-type contribution and the maximizing-type profile
contribution in Eq.~\eqref{eq:supp-global-bound} gives
\[
\varepsilon
\ge
\bigl(h_{\mathrm{sub}}^{-1}-h_{\max}^{-1}\bigr)
\delta_{\mathrm{leak}}
+\frac{\delta_{\mathrm{prof}}}{h_{\max}}.
\]
If \(\mathcal I_m=\mathcal I_m^\star\), then
\(\delta_{\mathrm{leak}}=0\), and Eq.~\eqref{eq:supp-global-bound} gives
\[
\varepsilon
\ge\frac{\delta_{\mathrm{prof}}}{h_{\max}}.
\]
Thus, by Eq.~\eqref{eq:supp-trace-constant}, in either case,
\begin{equation}
\label{eq:supp-deviation-budget}
\delta_{\mathrm{leak}}+\delta_{\mathrm{prof}}
\le C_{d,m}\varepsilon.
\end{equation}
Combining Eqs.~\eqref{eq:supp-geometric-trace-bound}
and~\eqref{eq:supp-deviation-budget}
proves Eq.~\eqref{eq:supp-trace-stability}.
\end{proof}

A nonempty closed set of system states disjoint from
\(\mathcal E_{\mathrm{lock}}\) has a strictly positive minimum
trace distance from it. Every zero-type-I-error protocol whose
system marginal lies in that set therefore has excess type-II
error at least the square of this distance divided by \(C_{d,m}\).

Corollary~\ref{cor:supp-trace-stability} bounds the distance to
\(\mathcal E_{\mathrm{lock}}\), but does not establish proximity to
the set of optimal system marginals.
The locking conditions do not prescribe the normalized states
within occupied joint sectors of a locked marginal or require
coherences between those sectors to vanish.
Near-extremal weights alone need not imply near-optimality of
a given protocol, and neither the full probe nor the POVM is
characterized.
The stability estimates above assume \(\alpha=0\).

The following proposition extends the marginal converse to nonzero type-I error.
For a system state \(\rho\), define
\begin{equation}
L(\rho):=\sum_{\eta:p_\eta>0}\frac{p_\eta}{h_m(\eta)}
\bigl[1+\chi^2(\boldsymbol\pi_\eta\Vert\boldsymbol\pi_\eta^\star)\bigr].
\label{eq:supp-finite-alpha-L}
\end{equation}
For the fixed \(G,d,m\), also set
\begin{equation}
K_{G,d,m}:=\max_{\substack{\eta\in\mathcal I_m\\\lambda\in\Lambda_\eta}}
\frac{d_\eta}{d_\lambda n_{\eta,\lambda}}\le1.
\label{eq:supp-finite-alpha-K}
\end{equation}
\begin{proposition}[Finite-error marginal converse]
\label{prop:supp-finite-alpha}
Let \((\sigma,M)\) be an \(m\)-query parallel protocol with actual average
errors \(\alpha,\beta\), and set \(\rho=\Tr_R\sigma\). Then
\begin{equation}
\beta\ge\bigl[\sqrt{L(\rho)}-\sqrt{K_{G,d,m}\alpha}\bigr]_+^2,
\label{eq:supp-finite-alpha-bound}
\end{equation}
where \(x_+:=\max\{x,0\}\). If the actual error satisfies \(\alpha\le a_0\),
the same bound holds with \(\alpha\) replaced by \(a_0\).
\end{proposition}
\begin{proof}
The twirling and profile-identification steps of
Theorem~\ref{thm:supp-support-weight-locking} do not require
\(\alpha=0\). Retain their normalized states \(\tau_\eta\), projections
\(\widetilde P_{\eta,\lambda}\), and conditional blocks in
Eq.~\eqref{eq:supp-conditional-blocks}. The typewise errors \(\beta_\eta\)
are defined as in Eq.~\eqref{eq:supp-typewise-prebound}. For each occupied
type, take the normalized partial trace of the accepting effect's
\(\eta\)-block:
\[
E_\eta:=\frac1{d_\eta}\Tr_{V_\eta}
\bigl[(\Pi_\eta\otimes\id_R)M_0(\Pi_\eta\otimes\id_R)\bigr],
\qquad 0\preceq E_\eta\preceq\id.
\]
Set
\[
a_\eta:=1-\Tr(E_\eta\tau_\eta),\qquad
x_{\eta,\lambda}:=\Tr(E_\eta\widetilde\tau_{\eta,\lambda}),
\qquad c_{\eta,\lambda}:=\frac{d_\eta}{d_\lambda n_{\eta,\lambda}}.
\]
Thus \(0\le a_\eta\le1\) and
\(\sum_{\eta:p_\eta>0}p_\eta a_\eta=\alpha\).
Discarding the nonnegative Haar-output contribution outside the
\(\eta\)-block, then applying Eq.~\eqref{eq:supp-partial-trace-bound}, gives
\begin{align}
\beta_\eta
&\ge\sum_{\lambda\in\Lambda_\eta}\frac{d_\eta}{d_\lambda}
\Tr\bigl[E_\eta(\id_{\mathbb C^{n_{\eta,\lambda}}}
\otimes\widehat\tau_{\eta,\lambda})\bigr]
\ge\sum_{\lambda\in\Lambda_\eta}c_{\eta,\lambda}x_{\eta,\lambda}.
\label{eq:supp-finite-alpha-block}
\end{align}
This step does not use Eq.~\eqref{eq:supp-block-acceptance},
which requires \(\alpha=0\).

For a fixed occupied \(\eta\), abbreviate \(E=E_\eta\),
\(\tau=\tau_\eta\), \(P=\widetilde P_{\eta,\lambda}\), and
\(\pi=\Tr(P\tau)\). Splitting \(\Tr(\tau P)\) with
\(E+(\id-E)=\id\) and applying Hilbert--Schmidt Cauchy--Schwarz to
each term yields
\[
\begin{aligned}
\pi
&\le |\Tr(\tau EP)|+|\Tr(\tau(\id-E)P)|\\
&\le\sqrt{\Tr(E\tau)\Tr(EP\tau P)}
 +\sqrt{\Tr((\id-E)\tau)\Tr((\id-E)P\tau P)}\\
&\le\sqrt{(1-a_\eta)x_{\eta,\lambda}}
 +\sqrt{a_\eta\pi}
\le\sqrt{x_{\eta,\lambda}}+\sqrt{a_\eta\pi}.
\end{aligned}
\]
The second square root uses \(0\preceq\id-E\preceq\id\).
These estimates hold for singular states and for \(a_\eta=0,1\).

Multiply the last inequality by
\(\sqrt{p_\eta c_{\eta,\lambda}}\). Monotonicity of the Euclidean
norm on nonnegative vectors and its triangle inequality give
\[
\begin{aligned}
\sqrt{L(\rho)}
&=\Bigl(\sum_{\eta:p_\eta>0}\sum_{\lambda\in\Lambda_\eta}
 p_\eta c_{\eta,\lambda}\pi_{\eta,\lambda}^2\Bigr)^{1/2}\\
&\le\Bigl(\sum_{\eta,\lambda}p_\eta c_{\eta,\lambda}
 x_{\eta,\lambda}\Bigr)^{1/2}
 +\Bigl(\sum_{\eta,\lambda}p_\eta c_{\eta,\lambda}
 a_\eta\pi_{\eta,\lambda}\Bigr)^{1/2}\\
&\le\sqrt\beta+\sqrt{K_{G,d,m}\alpha}.
\end{aligned}
\]
The abbreviated sums still run only over occupied types and active sectors.
In the last inequality, Eq.~\eqref{eq:supp-finite-alpha-block} and
\(\beta=\sum_{\eta:p_\eta>0}p_\eta\beta_\eta\) bound the first term
by \(\sqrt\beta\).
Here \(\sum_\lambda\pi_{\eta,\lambda}=1\), and
\(d_\lambda\ge d_\eta n_{\eta,\lambda}\) implies
\(c_{\eta,\lambda}\le n_{\eta,\lambda}^{-2}\le1\).
Since \(\sqrt\beta\ge0\), the preceding inequality gives
\(\sqrt\beta\ge[\sqrt{L(\rho)}-\sqrt{K_{G,d,m}\alpha}]_+\).
Squaring yields Eq.~\eqref{eq:supp-finite-alpha-bound}.
For \(\alpha=0\), Eq.~\eqref{eq:supp-finite-alpha-bound} is equivalent to
Eq.~\eqref{eq:supp-global-bound}; for \(\alpha=1\), it remains valid since
\(L(\rho)\le K_{G,d,m}\).
\end{proof}

\section{Diagonal-torus performance and coherence}
\label{sec:supp-diagonal-torus}

We now specialize the support and weight locking framework to the diagonal
torus \(\mathbb T_d\subset\mathrm U(d)\), with \(d\ge2\) and \(m\ge1\).

\subsection{Occupation types and Schur profiles}
\label{subsec:supp-occupation-profiles}
For an occupation vector
\(\mathbf n=(n_1,\ldots,n_d)\in\mathbb Z_{\ge0}^d\) with
\(|\mathbf n|:=\sum_{j=1}^d n_j=m\), let \(W_{\mathbf n}\) be the span of
the computational-basis vectors \(|x\rangle\) whose string \(x\) contains
each symbol \(j\) exactly \(n_j\) times. The occupation spaces give the decomposition
\[
\mathcal H^{\otimes m}
=
\bigoplus_{|\mathbf n|=m}W_{\mathbf n}.
\]
For \(z=\operatorname{diag}(z_1,\ldots,z_d)\in\mathbb T_d\), the tensor-power
action is scalar on \(W_{\mathbf n}\):
\[
z^{\otimes m}\big|_{W_{\mathbf n}}
=
\chi_{\mathbf n}(z)\id_{W_{\mathbf n}},
\qquad
\chi_{\mathbf n}(z)
:=
\prod_{j=1}^d z_j^{n_j}.
\]
Thus the torus types are indexed by occupation vectors, and
\(d_{\mathbf n}=1\).

For \(\lambda\vdash m\) with \(\ell(\lambda)\le d\), the character of
\(U_\lambda\) restricted to \(\mathbb T_d\) is the Schur polynomial
\[
s_\lambda(z_1,\ldots,z_d)
=
\sum_{|\mathbf n|=m}
K_{\lambda,\mathbf n}
z_1^{n_1}\cdots z_d^{n_d},
\]
where \(K_{\lambda,\mathbf n}\) is the Kostka
number~\cite{Macdonald1995}. The
\(\chi_{\mathbf n}\)-weight space of \(U_\lambda\) therefore has
dimension \(K_{\lambda,\mathbf n}\), so
\(n_{\mathbf n,\lambda}=K_{\lambda,\mathbf n}\) and
\[
h_m(\mathbf n)
=
\sum_{\substack{\lambda\vdash m\\\ell(\lambda)\le d}}
d_\lambda K_{\lambda,\mathbf n}.
\]
Since \(d_\lambda=s_\lambda(1^d)\), this score is the coefficient of
\(\mathbf z^{\mathbf n}:=z_1^{n_1}\cdots z_d^{n_d}\) in
\(\sum_\lambda s_\lambda(1^d)s_\lambda(\mathbf z)\). The Cauchy
identity gives
\[
\sum_\lambda s_\lambda(1^d)s_\lambda(\mathbf z)
=
\prod_{j=1}^d(1-z_j)^{-d}.
\]
Extracting the coefficient of \(\mathbf z^{\mathbf n}\) yields
\begin{equation}
h_m(\mathbf n)
=
\prod_{j=1}^d
\binom{n_j+d-1}{d-1}.
\label{eq:supp-torus-score}
\end{equation}

Write \(m=qd+r\), where \(0\le r<d\), and let
\(\mathcal B_{m,d}\) be the set of coordinate permutations of the vector
with \(r\) entries equal to \(q+1\) and \(d-r\) entries equal to \(q\).
Call these occupation vectors balanced. They are exactly the maximizing torus
types.
Indeed, if \(n_i\ge n_j+2\) and
\(\mathbf n'=\mathbf n-\mathbf e_i+\mathbf e_j\), then
\[
\frac{h_m(\mathbf n')}{h_m(\mathbf n)}
=
\frac{n_i(n_j+d)}
     {(n_i+d-1)(n_j+1)}
>1,
\]
because the numerator exceeds the denominator by
\((d-1)(n_i-n_j-1)\).
Each such smoothing step strictly increases the score and decreases
\(\sum_j n_j^2\), so repeated smoothing terminates at a balanced
occupation. Since Eq.~\eqref{eq:supp-torus-score} is symmetric in the entries
of \(\mathbf n\), all balanced occupations have the same score. Hence
\(\mathcal I_m^\star=\mathcal B_{m,d}\), and support locking restricts
every optimal system marginal to the direct sum of the balanced
occupation spaces.

To compare the weight-locking condition with the incoherence
constraint, let \(\Delta\) denote complete dephasing of all query
systems in the computational product basis. For each occupation \(\mathbf n\), let
\(\Pi_{\mathbf n}\) project onto \(W_{\mathbf n}\), and let
\(\Pi_{\mathbf n,\lambda}\) denote the corresponding joint
occupation/Schur projector.
Let \(\Pi_\lambda\) project onto the Schur sector
\(U_\lambda\otimes S_\lambda\) in
Eq.~\eqref{eq:schur-weyl-branching}.
Write \(f^\lambda:=\dim S_\lambda\), and
set
\[
N_{\mathbf n}:=\dim W_{\mathbf n}
=\frac{m!}{\prod_j n_j!}.
\]
The active Schur labels are
\[
\Lambda_{\mathbf n}
:=
\{\lambda\vdash m:\ell(\lambda)\le d,\,
K_{\lambda,\mathbf n}>0\}.
\]
Writing \(U_\lambda[\mathbf n]\) for the
\(\chi_{\mathbf n}\)-weight space of \(U_\lambda\), of dimension
\(K_{\lambda,\mathbf n}\), we have
\[
W_{\mathbf n}
\cong
\bigoplus_{\lambda\in\Lambda_{\mathbf n}}
U_\lambda[\mathbf n]\otimes S_\lambda.
\]

The projector \(\Pi_{\mathbf n,\lambda}\) commutes with the
\(\mathfrak S_m\)-action, which is transitive on the computational-basis
strings of occupation \(\mathbf n\). Its diagonal is therefore constant
on these strings, and its trace equals
\(f^\lambda K_{\lambda,\mathbf n}\). Hence
\begin{equation}
\Delta(\Pi_{\mathbf n,\lambda})
=
\frac{f^\lambda K_{\lambda,\mathbf n}}{N_{\mathbf n}}
\Pi_{\mathbf n}.
\label{eq:supp-occupation-constant-diagonal}
\end{equation}
Define the incoherent Schur profile by
\[
\pi^{\mathrm{inc}}_{\mathbf n,\lambda}
:=
\frac{f^\lambda K_{\lambda,\mathbf n}}{N_{\mathbf n}},
\qquad
\lambda\in\Lambda_{\mathbf n}.
\]
Since
\(\sum_{\lambda\in\Lambda_{\mathbf n}}
f^\lambda K_{\lambda,\mathbf n}=N_{\mathbf n}\),
this is a probability vector. Moreover, every diagonal operator \(X\)
supported on \(W_{\mathbf n}\) satisfies
\(\Tr(\Pi_{\mathbf n,\lambda}X)
=\pi^{\mathrm{inc}}_{\mathbf n,\lambda}\Tr X\).
Thus an incoherent system marginal has conditional Schur profile
\(\boldsymbol\pi^{\mathrm{inc}}_{\mathbf n}\) in every occupied type.

The typewise extremal profile of
Theorem~\ref{thm:supp-support-weight-locking} specializes to
\[
\pi^\star_{\mathbf n,\lambda}
:=
\frac{d_\lambda K_{\lambda,\mathbf n}}{h_m(\mathbf n)},
\qquad
\lambda\in\Lambda_{\mathbf n}.
\]
It is also a probability vector, with full support on
\(\Lambda_{\mathbf n}\).

Now fix a balanced occupation
\(\mathbf b\in\mathcal B_{m,d}\).
Every balanced occupation is a coordinate permutation of
\(\mathbf b\). By symmetry of the Schur polynomials,
\(K_{\lambda,\mathbf c}=K_{\lambda,\mathbf b}\)
for every balanced \(\mathbf c\), while
\(N_{\mathbf c}=N_{\mathbf b}\).
Hence all balanced occupations have the same active Schur-label set.
Since \(h_m(\mathbf b)=h_{\max}\), the two profiles reduce to
\begin{equation}
\pi^\star_{\mathbf b,\lambda}
=
\frac{d_\lambda K_{\lambda,\mathbf b}}{h_{\max}},
\qquad
\lambda\in\Lambda_{\mathbf b},
\label{eq:supp-locked-schur-profile}
\end{equation}
and
\begin{equation}
\pi^{\mathrm{inc}}_{\mathbf b,\lambda}
=
\frac{f^\lambda K_{\lambda,\mathbf b}}{N_{\mathbf b}},
\qquad
\lambda\in\Lambda_{\mathbf b}.
\label{eq:supp-incoherent-schur-profile}
\end{equation}
We write
\(\boldsymbol\pi^\star_{\mathbf b}\)
and
\(\boldsymbol\pi^{\mathrm{inc}}_{\mathbf b}\)
for the corresponding probability vectors.
Both profiles are independent of the choice of balanced occupation.
At unrestricted optimality, weight locking requires the profile
\(\boldsymbol\pi^\star_{\mathbf b}\), whereas incoherence fixes it to
\(\boldsymbol\pi^{\mathrm{inc}}_{\mathbf b}\).
Both profiles have full support on \(\Lambda_{\mathbf b}\), so
\(D(\boldsymbol\pi^\star_{\mathbf b}
\Vert\boldsymbol\pi^{\mathrm{inc}}_{\mathbf b})\)
is finite.

\subsection{Incoherent-marginal performance}
Equation~\eqref{eq:supp-occupation-constant-diagonal} fixes the Schur
profile in every occupied torus type of an incoherent system marginal.
We now determine the exact zero-type-I-error performance under this
marginal restriction, while allowing arbitrary ancillas and measurements.

For an \(m\)-query parallel protocol \((\sigma,M)\), write
\(\rho:=\Tr_R\sigma\) for its system marginal. Under the
zero-type-I-error constraint, define
\[
\beta^{\mathrm{inc}}_{d,m}
:=
\inf\left\{\beta(\sigma,M):
\alpha(\sigma,M)=0,\;\Delta(\rho)=\rho\right\}.
\]
Only the system marginal is constrained: the ancillary system,
system--ancilla correlations in the joint probe, and the measurement
are otherwise unrestricted.

For an occupation \(\mathbf n\), choose any computational-basis string
\(|x\rangle\in W_{\mathbf n}\) and project back onto it after the query.
The diagonal torus acts on \(|x\rangle\) by a phase, so this return test
has zero type-I error. Its Haar type-II error depends only on
\(\mathbf n\); denote it by
\begin{equation}
B(\mathbf n)
:=
\int_{\mathrm U(d)}
\prod_{j=1}^{d}|U_{jj}|^{2n_j}\,
\dd\mu_{\mathrm{Haar}}(U).
\label{eq:supp-B-Haar}
\end{equation}

\begin{theorem}[Incoherent-marginal optimum]
\label{thm:supp-diagonal-error}
For every \(d\ge2\), \(m\ge1\), and balanced occupation
\(\mathbf b\),
\begin{equation}
\beta^{\mathrm{inc}}_{d,m}
=
\min_{|\mathbf n|=m}B(\mathbf n)
=
B(\mathbf b).
\label{eq:supp-diagonal-exact}
\end{equation}
The infimum is attained without an ancilla by the return test on any
computational-basis string \(|x\rangle\in W_{\mathbf b}\).
Explicitly, take \(\sigma=M_0=|x\rangle\langle x|\) and \(M_1=\id-M_0\).
This protocol is diagonal at both the probe and measurement stages.
The same optimum holds when the joint probe and both POVM effects are
required to be diagonal in the computational product basis tensored
with a fixed ancilla basis.
\end{theorem}
\begin{proof}
\emph{Schur evaluation of the return test.}
Fix a computational-basis string \(|x\rangle\in W_{\mathbf n}\), and let
\(V_\tau\) permute the query systems. Transitivity on strings of
occupation \(\mathbf n\) gives
\[
\frac1{m!}\sum_{\tau\in\mathfrak S_m}
V_\tau|x\rangle\langle x|V_\tau^\dagger
=\frac{\Pi_{\mathbf n}}{N_{\mathbf n}}.
\]
Let \(|x_\lambda\rangle:=\Pi_\lambda|x\rangle\), without renormalizing.
Since the partial trace over \(S_\lambda\) is unchanged by the
permutation action, the average above gives
\begin{equation}
R_\lambda:=\Tr_{S_\lambda}|x_\lambda\rangle\langle x_\lambda|
=\frac{f^\lambda}{N_{\mathbf n}}\id_{U_\lambda[\mathbf n]}.
\label{eq:supp-basis-Schur-reduction}
\end{equation}
Here the identity is understood as zero outside \(U_\lambda[\mathbf n]\).
Set \(T_\lambda:=\Tr_{U_\lambda}|x_\lambda\rangle\langle x_\lambda|\).
The Haar twirl of \(\lvert x\rangle\langle x\rvert\) is
\[
\bigoplus_\lambda\frac{\id_{U_\lambda}}{d_\lambda}\otimes T_\lambda.
\]
Since \(R_\lambda\) and \(T_\lambda\) are the two reductions of the same
subnormalized pure vector, they have the same nonzero spectrum. Therefore
\begin{equation}
\begin{aligned}
B(\mathbf n)
&=\sum_{\lambda\in\Lambda_{\mathbf n}}
\frac{\Tr(T_\lambda^2)}{d_\lambda}
=\sum_{\lambda\in\Lambda_{\mathbf n}}
\frac{\Tr(R_\lambda^2)}{d_\lambda}\\
&=\frac1{N_{\mathbf n}^2}
\sum_{\lambda\in\Lambda_{\mathbf n}}
\frac{(f^\lambda)^2K_{\lambda,\mathbf n}}{d_\lambda}.
\end{aligned}
\label{eq:supp-basis-Haar-error}
\end{equation}
Using the two profiles defined in
Subsec.~\ref{subsec:supp-occupation-profiles}
and \(1+\chi^2(\mathbf p\Vert\mathbf q)
=\sum_\lambda p_\lambda^2/q_\lambda\),
Eq.~\eqref{eq:supp-basis-Haar-error} becomes
\begin{equation}
B(\mathbf n)
=\sum_{\lambda\in\Lambda_{\mathbf n}}
\frac{(\pi^{\mathrm{inc}}_{\mathbf n,\lambda})^2}
{d_\lambda K_{\lambda,\mathbf n}}
=\frac{1+\chi^2(\boldsymbol\pi^{\mathrm{inc}}_{\mathbf n}
\Vert\boldsymbol\pi^\star_{\mathbf n})}{h_m(\mathbf n)}.
\label{eq:supp-diagonal-Pearson}
\end{equation}

\emph{Restricted-class converse.}
For an incoherent system marginal,
Eq.~\eqref{eq:supp-occupation-constant-diagonal} fixes the Schur
profile to \(\boldsymbol\pi^{\mathrm{inc}}_{\mathbf n}\) in every
occupied type.
Applying Theorem~\ref{thm:supp-support-weight-locking} with
\(d_{\mathbf n}=1\), and using Eq.~\eqref{eq:supp-diagonal-Pearson}, gives
\begin{equation}
\beta(\sigma,M)
\ge\sum_{\mathbf n:p_{\mathbf n}>0}p_{\mathbf n}B(\mathbf n)
\ge\min_{|\mathbf n|=m}B(\mathbf n).
\label{eq:supp-diagonal-converse}
\end{equation}
Because the sum includes every occupied type, balanced or not, the
bound applies to the entire incoherent-marginal class, with arbitrary
ancillas and measurements; no support-locking assumption is used.

\emph{Balancing.}
It remains to minimize the Haar integral in
Eq.~\eqref{eq:supp-B-Haar}.
Choose \(i,j\) with
\(n_i=r\ge s+2=n_j+2\), and let
\(\mathbf n':=\mathbf n-\mathbf e_i+\mathbf e_j\).
For \(U\in\mathrm U(d)\), write
\[
u:=|U_{ii}|^2,
\qquad
v:=|U_{jj}|^2,
\qquad
F_0(U):=
\prod_{\ell\ne i,j}|U_{\ell\ell}|^{2n_\ell}.
\]
Simultaneously exchanging rows and columns \(i\) and \(j\) preserves
Haar measure, swaps \(u\) and \(v\), and leaves \(F_0\) unchanged.
Averaging the integrand with its exchanged copy gives
\begin{equation}
B(\mathbf n)-B(\mathbf n')
=\frac12\mathbb E\!\left[
F_0u^sv^s(u-v)\bigl(u^{r-s-1}-v^{r-s-1}\bigr)
\right]\ge0.
\label{eq:supp-diagonal-smoothing}
\end{equation}
This argument uses only the joint Haar symmetry; no independence of
the diagonal entries is assumed.

The same smoothing operation as in
Subsec.~\ref{subsec:supp-occupation-profiles} decreases
\(\sum_j n_j^2\), so iteration reaches a balanced occupation.
Since \(B\) is invariant under coordinate permutations,
\[
\min_{|\mathbf n|=m}B(\mathbf n)=B(\mathbf b).
\]
The balanced return test attains \(B(\mathbf b)\), while
Eq.~\eqref{eq:supp-diagonal-converse} gives the matching converse.
This proves Eq.~\eqref{eq:supp-diagonal-exact}.

Every protocol whose joint probe and POVM effects are diagonal in the
computational product basis of the query systems tensored with a fixed
ancilla basis has an incoherent system marginal.
The balanced return test is diagonal at both the probe and
measurement stages. Hence imposing this stronger joint-diagonality
restriction does not change the zero-type-I-error optimum.
\end{proof}

\emph{Special cases.}
At \(m=1\), \(B(\mathbf n)=1/d\). For \(d=2\), unitarity gives
\(|U_{11}|^2=|U_{22}|^2=t\), with \(t\) uniform on \([0,1]\) under Haar
measure, so every occupation has \(B(\mathbf n)=1/(m+1)\).
Thus balanced occupations need not be the only minimizers. In combination
with Eq.~\eqref{eq:supp-torus-score},
\begin{equation}
\beta^{\mathrm{inc}}_{2,m}=\frac1{m+1},\qquad
\beta^\star_{2,m}(0)=\frac1{\lfloor(m+2)^2/4\rfloor}.
\label{eq:supp-qubit-error-values}
\end{equation}

Theorem~\ref{thm:supp-diagonal-error} and the integral representation
in Eq.~\eqref{eq:supp-B-Haar} also determine the fixed-dimension
separation between the incoherent-marginal and unrestricted optima.

\begin{proposition}[Fixed-dimension error and query scaling]
\label{prop:supp-diagonal-scaling}
Fix \(d\ge2\). As \(m\to\infty\),
\begin{equation}
\beta^{\mathrm{inc}}_{d,m}
=\Theta_d\!\left(m^{-d(d-1)/2}\right),\qquad
\beta^\star_{d,m}(0)
=\Theta_d\!\left(m^{-d(d-1)}\right).
\label{eq:supp-diagonal-scaling}
\end{equation}
For \(0<\delta<1\), define
\[
\begin{aligned}
m_{\mathrm{inc}}(d,\delta)
&:=\min\left\{m\ge1:\beta^{\mathrm{inc}}_{d,m}\le\delta\right\},\\
m_\star(d,\delta)
&:=\min\left\{m\ge1:\beta^\star_{d,m}(0)\le\delta\right\}.
\end{aligned}
\]
Then, as \(\delta\downarrow0\),
\begin{equation}
m_{\mathrm{inc}}(d,\delta)
=\Theta_d\!\left(\delta^{-2/[d(d-1)]}\right),\qquad
m_\star(d,\delta)
=\Theta_d\!\left(\delta^{-1/[d(d-1)]}\right).
\label{eq:supp-query-scaling}
\end{equation}
\end{proposition}
\begin{proof}
\emph{Unrestricted optimum.}
For a balanced occupation \(\mathbf b\),
Eq.~\eqref{eq:supp-torus-score} and
\(\beta^\star_{d,m}(0)=h_{\max}^{-1}\) give
\[
\beta^\star_{d,m}(0)
=\left[\prod_{j=1}^d\binom{b_j+d-1}{d-1}\right]^{-1}.
\]
Since \(\mathbf b\) is balanced,
\(b_j\in\{\lfloor m/d\rfloor,\lceil m/d\rceil\}\), and hence
\(b_j=m/d+O(1)\) uniformly in \(j\).
Each binomial factor is \(\Theta_d(m^{d-1})\), so their product is
\(\Theta_d(m^{d(d-1)})\). Therefore
\(\beta^\star_{d,m}(0)=\Theta_d(m^{-d(d-1)})\).

\emph{Incoherent-marginal optimum.}
Set \(a_j(U):=|U_{jj}|^2\) and
\(F_m(U):=\prod_{j=1}^d a_j(U)^{b_j}\).
Theorem~\ref{thm:supp-diagonal-error} and
Eq.~\eqref{eq:supp-B-Haar} give
\[
\beta^{\mathrm{inc}}_{d,m}
=\int_{\mathrm U(d)}F_m(U)\,\dd\mu_{\mathrm{Haar}}(U).
\]
For \(m\ge d\), every \(b_j\) is positive. Since
\(0\le a_j(U)\le1\), we have \(F_m(U)=1\) if and only if
\(a_j(U)=1\) for every \(j\), equivalently \(U\in\mathbb T_d\).
Thus the integral is concentrated near the torus \(\mathbb T_d\).

At the identity, the normal directions to \(\mathbb T_d\) are represented by
\[
\mathfrak p:=\{H=H^\dagger:\operatorname{diag}H=0\},
\qquad
\|H\|^2:=\sum_{i<j}|H_{ij}|^2.
\]
This real vector space has dimension
\(s_d=d^2-d=d(d-1)\).
The error exponent below is \(s_d/2\).

The differential of \(\Phi(D,H):=D e^{iH}\) is invertible at every
\((D,0)\), since the diagonal and zero-diagonal Hermitian directions
are complementary. By the tubular-neighborhood theorem, for some
\(\epsilon_0>0\), \(\Phi\) is a diffeomorphism from
\(\mathbb T_d\times\{H\in\mathfrak p:\|H\|<\epsilon_0\}\)
onto a neighborhood of \(\mathbb T_d\).

On a smaller fixed tube, Haar measure has the form
\[
\dd\mu_{\mathrm{Haar}}(De^{iH})
=J(D,H)\,\dd\nu(D)\,\dd H,
\]
where \(\nu\) is normalized Haar measure on \(\mathbb T_d\), and \(J\) is
bounded above and below by positive constants depending only on \(d\).
Since left multiplication by \(D\in\mathbb T_d\) changes only row phases,
\(|(De^{iH})_{jj}|=|(e^{iH})_{jj}|\), so \(F_m(De^{iH})\) is independent
of \(D\).

Taylor expansion on this fixed neighborhood gives
\[
a_j(e^{iH})
=1-\sum_{k\ne j}|H_{jk}|^2+O_d(\|H\|^3).
\]
After shrinking the tube, assume \(a_j(e^{iH})\ge1/2\) for every \(j\).
Then
\begin{equation}
\Psi_m(H):=-\sum_jb_j\ln a_j(e^{iH})
=\sum_{i<j}(b_i+b_j)|H_{ij}|^2+R_m(H),\qquad
|R_m(H)|\le C_d m\|H\|^3.
\label{eq:supp-torus-normal-expansion}
\end{equation}
Here \(\ln\) denotes the natural logarithm.
For all sufficiently large \(m\), balancedness gives
\(b_i+b_j=\Theta_d(m)\). After shrinking the fixed tube once more,
the remainder is absorbed into the quadratic term, and hence
\[
c_d m\|H\|^2\le\Psi_m(H)\le C_d m\|H\|^2
\]
for some constants \(c_d,C_d>0\).

The quadratic bounds and the Jacobian estimates give
\[
\int_{\mathrm{tube}}F_m(U)\,\dd\mu_{\mathrm{Haar}}(U)
=\Theta_d\!\left(m^{-s_d/2}\right).
\]
The upper bound follows by extending the Gaussian integral to all of
\(\mathfrak p\). For the lower bound, integrate over
\(\|H\|\le m^{-1/2}\), where the exponent remains bounded and the ball
has volume \(\Theta_d(m^{-s_d/2})\).

Outside the fixed tube, define \(F(U):=\prod_{j=1}^d a_j(U)\).
Since \(F(U)=1\) exactly on \(\mathbb T_d\), compactness gives
\(\vartheta_d<1\) such that \(F(U)\le\vartheta_d\) on the complement
of the tube. For \(m\ge2d\), balancedness gives \(b_j\ge m/(2d)\), and hence
\[
F_m(U)\le F(U)^{m/(2d)}\le\vartheta_d^{\,m/(2d)}.
\]
This contribution is exponentially small compared with \(m^{-s_d/2}\).
Combining the local and nonlocal estimates yields
\[
\beta^{\mathrm{inc}}_{d,m}
=\Theta_d\!\left(m^{-s_d/2}\right)
=\Theta_d\!\left(m^{-d(d-1)/2}\right),
\]
which completes the proof of the error-rate statement.

Finally, if a positive sequence satisfies
\(c m^{-s}\le g(m)\le C m^{-s}\)
for all sufficiently large \(m\), with \(c,C,s>0\), then the smallest
integer with \(g(m)\le\delta\) is
\(\Theta(\delta^{-1/s})\) as \(\delta\downarrow0\).
Applying this observation to the two estimates in
Eq.~\eqref{eq:supp-diagonal-scaling} proves
Eq.~\eqref{eq:supp-query-scaling}.
\end{proof}

\subsection{Minimum system-marginal coherence at the unrestricted optimum}
Throughout this subsection, entropy and relative entropy are measured
in bits unless \(\ln\) is written explicitly.
Define
\(C_{\mathrm{rel}}(\rho):=D(\rho\Vert\Delta(\rho))\).
Let \(C_{\min}(d,m)\) denote the infimum of
\(C_{\mathrm{rel}}(\Tr_R\sigma)\) over all \(m\)-query parallel
protocols satisfying \(\alpha(\sigma,M)=0\) and
\(\beta(\sigma,M)=\beta^\star_{d,m}(0)\).

\begin{theorem}[Minimum system-marginal coherence]
\label{thm:supp-minimum-coherence}
For every \(d\ge2\) and \(m\ge1\), the infimum defining
\(C_{\min}(d,m)\) is attained. For any balanced occupation \(\mathbf b\),
\begin{equation}
C_{\min}(d,m)
=D(\boldsymbol\pi^\star_{\mathbf b}\Vert\boldsymbol\pi^{\mathrm{inc}}_{\mathbf b}).
\label{eq:supp-coherence-minimum}
\end{equation}
The right-hand side is independent of the balanced occupation chosen.
It vanishes exactly at \(m=1\) and is strictly positive for every
\(m\ge2\). At \(m=1\), the zero minimum is attained by the
ancilla-free basis return test of Theorem~\ref{thm:supp-diagonal-error}.
\end{theorem}
\begin{proof}
\emph{Converse.}
Let \((\sigma,M)\) be an optimal \(m\)-query parallel protocol with
\(m\ge2\), and set \(\rho:=\Tr_R\sigma\). We first prove
\begin{equation}
C_{\mathrm{rel}}(\rho)\ge
D(\boldsymbol\pi^\star_{\mathbf b}\Vert\boldsymbol\pi^{\mathrm{inc}}_{\mathbf b}).
\label{eq:supp-coherence-bound}
\end{equation}

Support locking and
\(\mathcal I_m^\star=\mathcal B_{m,d}\) give
\[
\operatorname{supp}\rho
\subseteq
\bigoplus_{\mathbf c\in\mathcal B_{m,d}}W_{\mathbf c},
\qquad
\sum_{\mathbf c\in\mathcal B_{m,d}}p_{\mathbf c}=1.
\]
Complete dephasing \(\Delta\) preserves every occupation space, and
\(\Pi_{\mathbf c}\) is diagonal in the computational basis. Since
\(\Delta\) is self-adjoint with respect to the Hilbert--Schmidt inner
product,
\[
\Tr(\Pi_{\mathbf c}\Delta(\rho))
=
\Tr(\Delta(\Pi_{\mathbf c})\rho)
=
\Tr(\Pi_{\mathbf c}\rho)
=
p_{\mathbf c}.
\]
Thus \(\Delta(\rho)\) is also supported on the balanced direct sum and
has the same occupation-type weights as \(\rho\).

Define
\[
\Pi_\perp
:=
\id_{\mathcal H^{\otimes m}}
-
\sum_{\mathbf c\in\mathcal B_{m,d}}
\sum_{\lambda\in\Lambda_{\mathbf b}}
\Pi_{\mathbf c,\lambda}.
\]
The corresponding classical measurement channel is
\begin{equation}
\mathcal M(X)
:=
\sum_{\mathbf c\in\mathcal B_{m,d}}
\sum_{\lambda\in\Lambda_{\mathbf b}}
\Tr(\Pi_{\mathbf c,\lambda}X)
|\mathbf c,\lambda\rangle\langle\mathbf c,\lambda|
+
\Tr(\Pi_\perp X)|\perp\rangle\langle\perp|.
\label{eq:supp-type-schur-measurement}
\end{equation}
The joint projectors together with \(\Pi_\perp\) form a complete projective
measurement, so \(\mathcal M\) is CPTP.
The outcome \(\perp\) has zero probability for both \(\rho\) and
\(\Delta(\rho)\).

For \(\mathbf c\in\mathcal B_{m,d}\) and
\(\lambda\in\Lambda_{\mathbf b}\), the output probabilities on the joint
outcomes are
\begin{equation}
\begin{aligned}
P(\mathbf c,\lambda)
&:=
\Tr(\Pi_{\mathbf c,\lambda}\rho)
=
p_{\mathbf c}\pi^\star_{\mathbf c,\lambda},
\\
Q(\mathbf c,\lambda)
&:=
\Tr(\Pi_{\mathbf c,\lambda}\Delta(\rho))
=
p_{\mathbf c}\pi^{\mathrm{inc}}_{\mathbf c,\lambda}.
\end{aligned}
\label{eq:supp-classical-profile-distributions}
\end{equation}
The first identity is weight locking when \(p_{\mathbf c}>0\); when
\(p_{\mathbf c}=0\), both sides vanish because
\(\Pi_{\mathbf c,\lambda}\preceq\Pi_{\mathbf c}\). The second identity
follows because the \(\mathbf c\)-block of \(\Delta(\rho)\) is a positive
diagonal operator of trace \(p_{\mathbf c}\), so the constant-diagonal
identity in Eq.~\eqref{eq:supp-occupation-constant-diagonal} applies.

Since the two output distributions have the same occupation-type marginal
and outcomes with \(p_{\mathbf c}=0\) contribute zero,
\[
\begin{aligned}
D(P\|Q)
&=
\sum_{\substack{\mathbf c\in\mathcal B_{m,d}\\p_{\mathbf c}>0}}
p_{\mathbf c}
D\!\left(
\boldsymbol\pi^\star_{\mathbf c}
\middle\|
\boldsymbol\pi^{\mathrm{inc}}_{\mathbf c}
\right)
\\
&=
D\!\left(
\boldsymbol\pi^\star_{\mathbf b}
\middle\|
\boldsymbol\pi^{\mathrm{inc}}_{\mathbf b}
\right).
\end{aligned}
\]
The last equality follows because all balanced occupations have the same
two profiles (Subsec.~\ref{subsec:supp-occupation-profiles}) and the
occupied weights sum to one.
Applying the data-processing inequality to the same channel \(\mathcal M\) gives
\[
\begin{aligned}
C_{\mathrm{rel}}(\rho)
&=
D\!\left(\rho\middle\|\Delta(\rho)\right)
\\
&\ge
D\!\left(
\mathcal M(\rho)
\middle\|
\mathcal M(\Delta(\rho))
\right)
\\
&=
D(P\|Q),
\end{aligned}
\]
which proves Eq.~\eqref{eq:supp-coherence-bound}.

\emph{Attainment.}
Define
\begin{equation}
\omega_{\mathbf b}:=
\sum_{\lambda\in\Lambda_{\mathbf b}}
\frac{\pi^\star_{\mathbf b,\lambda}}{f^\lambda K_{\lambda,\mathbf b}}
\Pi_{\mathbf b,\lambda}
=\frac1{h_{\max}}\sum_{\lambda\in\Lambda_{\mathbf b}}
\frac{d_\lambda}{f^\lambda}\Pi_{\mathbf b,\lambda}.
\label{eq:supp-minimum-coherence-state}
\end{equation}
Its coefficients are positive and the projectors are mutually orthogonal
with ranks \(f^\lambda K_{\lambda,\mathbf b}\). Hence
\(\omega_{\mathbf b}\succeq0\), \(\Tr\omega_{\mathbf b}=1\),
\(\operatorname{supp}\omega_{\mathbf b}=W_{\mathbf b}\), and its sector
weights are \(\pi^\star_{\mathbf b,\lambda}\). The constant-diagonal identity gives
\begin{equation}
\Delta(\omega_{\mathbf b})=\frac{\Pi_{\mathbf b}}{N_{\mathbf b}},\qquad
C_{\mathrm{rel}}(\omega_{\mathbf b})
=\log_2N_{\mathbf b}-S(\omega_{\mathbf b})
=D(\boldsymbol\pi^\star_{\mathbf b}\Vert\boldsymbol\pi^{\mathrm{inc}}_{\mathbf b}).
\label{eq:supp-minimum-coherence-entropy}
\end{equation}
Indeed,
\(S(\omega_{\mathbf b})=H(\boldsymbol\pi^\star_{\mathbf b})
+\sum_\lambda\pi^\star_{\mathbf b,\lambda}
\log_2(f^\lambda K_{\lambda,\mathbf b})\),
where \(H\) denotes the Shannon entropy.
The dephased state is uniform on \(W_{\mathbf b}\), not on the whole query space.

Choose an orthonormal basis \(|\mathbf b;\lambda,a,t\rangle\) of each
\(U_\lambda[\mathbf b]\otimes S_\lambda\). A normalized purification on
an \(N_{\mathbf b}\)-dimensional ancilla is
\[
|\Psi_{\mathbf b}\rangle
=\sum_{\lambda\in\Lambda_{\mathbf b}}
\sqrt{\frac{\pi^\star_{\mathbf b,\lambda}}{K_{\lambda,\mathbf b}f^\lambda}}
\sum_{a=1}^{K_{\lambda,\mathbf b}}\sum_{t=1}^{f^\lambda}
|\mathbf b;\lambda,a,t\rangle_A|\lambda,a,t\rangle_R.
\]
Its system marginal is \(\omega_{\mathbf b}\). More generally, any
purification has support in \(W_{\mathbf b}\otimes\mathcal H_R\), since
its squared norm outside this space is
\(\Tr[(\id-\Pi_{\mathbf b})\omega_{\mathbf b}]=0\).
The torus acts on this space by the character
\(\chi_{\mathbf b}(g)\), so the protocol with joint probe
\(|\Psi_{\mathbf b}\rangle\langle\Psi_{\mathbf b}|\),
accepting effect
\(M_0=|\Psi_{\mathbf b}\rangle\langle\Psi_{\mathbf b}|\),
and \(M_1=\id-M_0\) has \(\alpha=0\).

For the Haar error, the subnormalized Schur reductions are
\begin{equation}
R_\lambda:=\Tr_{S_\lambda}
(\Pi_\lambda\omega_{\mathbf b}\Pi_\lambda)
=\frac{\pi^\star_{\mathbf b,\lambda}}{K_{\lambda,\mathbf b}}
\id_{U_\lambda[\mathbf b]},\qquad
\Tr R_\lambda=\pi^\star_{\mathbf b,\lambda}.
\label{eq:supp-optimal-Schur-reduction}
\end{equation}
The acceptance amplitude is
\(\Tr(\omega_{\mathbf b}U^{\otimes m})
=\sum_\lambda\Tr(R_\lambda U_\lambda(U))\).
Haar orthogonality of inequivalent irreducible matrix elements then gives
\begin{equation}
\beta=\sum_{\lambda\in\Lambda_{\mathbf b}}\frac{\Tr R_\lambda^2}{d_\lambda}
=\sum_{\lambda\in\Lambda_{\mathbf b}}
\frac{(\pi^\star_{\mathbf b,\lambda})^2}{d_\lambda K_{\lambda,\mathbf b}}
=\frac1{h_{\max}}.
\label{eq:supp-minimum-coherence-attainment}
\end{equation}
Thus the lower bound in Eq.~\eqref{eq:supp-coherence-bound} is attained by an
optimal protocol, proving Eq.~\eqref{eq:supp-coherence-minimum}.

\emph{Strict positivity.}
For \(d,m\ge2\),
\begin{equation}
D(\boldsymbol\pi^\star_{\mathbf b}\Vert\boldsymbol\pi^{\mathrm{inc}}_{\mathbf b})>0.
\label{eq:supp-diagonal-coherence-positive}
\end{equation}

Suppose, for contradiction, that the divergence in
Eq.~\eqref{eq:supp-diagonal-coherence-positive} vanishes. The profiles
have the same support, so they must coincide. For every
\(\lambda\in\Lambda_{\mathbf b}\), canceling the positive factor
\(K_{\lambda,\mathbf b}\) gives
\[
\frac{d_\lambda}{f^\lambda}
=
\frac{h_{\max}}{N_{\mathbf b}},
\]
which would make \(d_\lambda/f^\lambda\) constant over all active
partitions.

Let \(H_k\) denote the complete homogeneous symmetric polynomial of
degree \(k\). Since \(s_{(m)}=H_m\), the coefficient of
\(\mathbf z^{\mathbf b}\) gives \(K_{(m),\mathbf b}=1\). The Jacobi--Trudi
identity~\cite{Macdonald1995} also gives
\[
s_{(m-1,1)}
=
H_{m-1}H_1-H_m.
\]
The coefficient of \(\mathbf z^{\mathbf b}\) on the right-hand side is
\(\#\{j:b_j>0\}-1\). A balanced occupation vector has
\(\#\{j:b_j>0\}=\min\{m,d\}\ge2\), so
\(K_{(m-1,1),\mathbf b}=\min\{m,d\}-1>0\).
Thus both \((m)\) and \((m-1,1)\) are active.

The Weyl dimension and hook-length
formulas~\cite{Fulton1997} give
\[
f^{(m)}=1,
\qquad
d_{(m)}=\binom{d+m-1}{m},
\qquad
f^{(m-1,1)}=m-1,
\qquad
d_{(m-1,1)}
=(m-1)\binom{d+m-2}{m}.
\]
Their dimension ratios are therefore
\begin{equation}
\frac{d_{(m)}}{f^{(m)}}
=
\binom{d+m-1}{m},
\qquad
\frac{d_{(m-1,1)}}{f^{(m-1,1)}}
=
\binom{d+m-2}{m}.
\label{eq:supp-torus-dimension-ratios}
\end{equation}
The first value is strictly larger than the second for \(d,m\ge2\);
indeed, their ratio is \((d+m-1)/(d-1)>1\). This contradicts the
required constancy and proves
Eq.~\eqref{eq:supp-diagonal-coherence-positive}.

\emph{The one-query boundary.}
For \(m=1\), take the ancilla-free protocol
\[
\sigma:=|1\rangle\langle1|,
\qquad
M_0:=|1\rangle\langle1|,
\qquad
M_1:=\id-M_0.
\]
Every
\(g=\operatorname{diag}(z_1,\ldots,z_d)\in\mathbb T_d\) satisfies
\(g|1\rangle=z_1|1\rangle\), so the protocol has
\(\alpha(\sigma,M)=0\). Under the Haar alternative, the single-copy Haar
average gives
\begin{equation}
\begin{aligned}
\beta(\sigma,M)
&=
\int_{\mathrm U(d)}
|\langle1|U|1\rangle|^2
\,\dd\mu_{\mathrm{Haar}}(U)
\\
&=
\frac1d
=
\beta_{d,1}^{\star}(0).
\end{aligned}
\label{eq:supp-torus-one-query}
\end{equation}
For the last equality, Eq.~\eqref{eq:supp-torus-score} gives
\(h_{\max}=d\) at \(m=1\).
The probe is diagonal in the computational basis, so
\(C_{\mathrm{rel}}=0\); both POVM effects are diagonal as well.
Thus the one-query optimum is attained without coherence.
Both profiles consist of the single entry one, so the right-hand side
of Eq.~\eqref{eq:supp-coherence-minimum} also vanishes.
Since \(C_{\mathrm{rel}}\ge0\), this gives
\(C_{\min}(d,1)=0\).
\end{proof}

\emph{Two-query qubit example.}
For \(d=m=2\) and \(\mathbf b=(1,1)\), let
\[
|s\rangle=\frac{|01\rangle+|10\rangle}{\sqrt2},
\qquad
|a\rangle=\frac{|01\rangle-|10\rangle}{\sqrt2}.
\]
In the order \(\lambda=(2),(1,1)\), the two profiles are
\(\boldsymbol\pi^\star_{\mathbf b}=(3/4,1/4)\) and
\(\boldsymbol\pi^{\mathrm{inc}}_{\mathbf b}=(1/2,1/2)\).
Equation~\eqref{eq:supp-qubit-error-values} gives
\(\beta^{\mathrm{inc}}_{2,2}=1/3\), whereas
\(\beta^\star_{2,2}(0)=1/4\).

The minimum \(C_{\min}(2,2)=1-h_2(3/4)\) is attained by
\[
\omega_{\mathbf b}
=\frac{3|s\rangle\langle s|+|a\rangle\langle a|}{4},
\]
where \(h_2\) is the binary entropy.
This state is supported on a single occupation space and is invariant
under the collective torus action, yet it has coherence in the
computational product basis.

For qubits, Eq.~\eqref{eq:supp-coherence-minimum} also determines the
large-\(m\) growth of \(C_{\min}(2,m)\).

\begin{proposition}[Qubit minimum-coherence rate]
\label{prop:supp-qubit-coherence-rate}
As \(m\to\infty\),
\begin{equation}
C_{\min}(2,m)
=\left(\frac23-\frac1{6\ln2}\right)m+O(\log m)\quad\mathrm{bits}.
\label{eq:supp-qubit-coherence-rate}
\end{equation}
\end{proposition}
\begin{proof}
Let \(k:=\lfloor m/2\rfloor\) and \(\mathbf b:=(m-k,k)\).
Then \(N_{\mathbf b}=\binom mk\) and
\(h_{\max}=(k+1)(m-k+1)\).
Index the active Schur labels by \(\lambda_j:=(m-j,j)\),
\(0\le j\le k\). For these labels, \(K_{\lambda_j,\mathbf b}=1\).
Set \(d_j:=d_{\lambda_j}=m-2j+1\) and
\(f_j:=f^{\lambda_j}=\binom mj-\binom m{j-1}\),
with \(\binom m{-1}:=0\). The two profiles are
\[
\pi^\star_{\mathbf b,j}=\frac{d_j}{h_{\max}},
\qquad
\pi^{\mathrm{inc}}_{\mathbf b,j}=\frac{f_j}{N_{\mathbf b}}.
\]
These are positive probability vectors on \(j=0,\ldots,k\).
Since \(f_j=\binom mj(m-2j+1)/(m-j+1)\), the common factor
\(m-2j+1\) cancels, giving
\begin{equation}
\frac{\pi^\star_{\mathbf b,j}}{\pi^{\mathrm{inc}}_{\mathbf b,j}}
=\frac{N_{\mathbf b}(m-j+1)}{h_{\max}\binom mj}.
\label{eq:supp-qubit-profile-ratio}
\end{equation}

Let \(\mathsf{h}(x):=-x\ln x-(1-x)\ln(1-x)\), with
\(0\ln0:=0\), and set \(\varphi(x):=\ln2-\mathsf{h}(x)\).
Uniformly,
\[
\frac{e^{m\mathsf{h}(j/m)}}{m+1}\le\binom mj
\le e^{m\mathsf{h}(j/m)},\qquad 0\le j\le k.
\]
For \(0<j<m\), the \(j\)th mass of the binomial distribution with
parameter \(j/m\) is maximal among its \(m+1\) masses, so it lies
between \(1/(m+1)\) and \(1\). The case \(j=0\) is immediate.
Also \(2^m/(m+1)\le N_{\mathbf b}\le2^m\) and
\(1/(m+1)\le(m-j+1)/h_{\max}\le1\). Equation~\eqref{eq:supp-qubit-profile-ratio}
therefore gives
\begin{equation}
\ln\frac{\pi^\star_{\mathbf b,j}}{\pi^{\mathrm{inc}}_{\mathbf b,j}}
=m\varphi(j/m)+E_{m,j},\qquad
|E_{m,j}|\le2\ln(m+1).
\label{eq:supp-qubit-uniform-log-bound}
\end{equation}
Since \(\sum_{j=0}^k\pi^\star_{\mathbf b,j}=1\),
Eqs.~\eqref{eq:supp-qubit-uniform-log-bound} and~\eqref{eq:supp-coherence-minimum} give
\[
(\ln2)C_{\min}(2,m)
=m\sum_{j=0}^k\pi^\star_{\mathbf b,j}\varphi(j/m)+O(\ln m).
\]

Define \(w(x):=4(1-2x)\) on \([0,1/2]\).
The profile \(\boldsymbol\pi^\star_{\mathbf b}\) obeys
\[
\sum_{j=0}^k\left|\pi^\star_{\mathbf b,j}-\frac1m w(j/m)\right|
\le\frac8m.
\]
Indeed, writing \(4h_{\max}=m^2+\varepsilon_m\), where
\(0<\varepsilon_m\le4m+4\), direct subtraction together with
\(\sum_{j=0}^k(m-2j)=h_{\max}-(k+1)\) gives the stated bound.

The function \(w\varphi\) is nonnegative and decreasing on
\([0,1/2]\). After extending it by zero to \([1/2,1]\), its Riemann
sum on the mesh \(1/m\) differs from the integral by \(O(1/m)\).
Together with the profile bound and \(0\le\varphi\le\ln2\), this yields
\[
\sum_{j=0}^k\pi^\star_{\mathbf b,j}\varphi(j/m)
=\int_0^{1/2}w(x)\varphi(x)\,\dd x+O(1/m).
\]
A direct integration gives
\[
\int_0^{1/2}w(x)\varphi(x)\,\dd x=\frac23\ln2-\frac16.
\]
Combining the likelihood-ratio estimate, the profile approximation,
and the integral evaluation gives
\[
(\ln2)C_{\min}(2,m)
=\left(\frac23\ln2-\frac16\right)m+O(\ln m).
\]
Dividing by \(\ln2\) converts the result to bits and proves
Eq.~\eqref{eq:supp-qubit-coherence-rate}.
\end{proof}

\subsection{Probe and measurement restrictions at finite tolerance}
At positive type-I-error tolerance, we compare incoherent system
marginals with fully diagonal protocols, whose joint probe state and
both POVM effects are diagonal in the system--ancilla product basis
formed from the computational product basis on the query systems and
a fixed reference basis on the ancilla.

\begin{proposition}[Fully diagonal optimum at finite tolerance]
\label{prop:supp-fully-diagonal-ROC}
Let \(d\ge2\), \(m\ge1\), \(0\le a_0\le1\), and let \(\mathbf b\) be any
balanced occupation. Every fully diagonal \(m\)-query parallel protocol satisfying
\(\alpha(\sigma,M)\le a_0\) obeys
\begin{equation}
\beta(\sigma,M)\ge(1-a_0)B(\mathbf b).
\label{eq:supp-fully-diagonal-ROC}
\end{equation}
The bound is attained by an ancilla-free fully diagonal protocol.
\end{proposition}
\begin{proof}
Write
\[
\sigma=\sum_{x,r}t_{x,r}|x,r\rangle\langle x,r|,
\qquad
e_{y,r}:=\langle y,r|M_0|y,r\rangle,
\]
where \(x,y\) are query-basis strings and \(r\) labels the ancillary
basis. Then \(t_{x,r}\ge0\), \(\sum_{x,r}t_{x,r}=1\), and
\(0\le e_{y,r}\le1\).
Since the torus fixes every joint basis projector,
\[
1-\alpha(\sigma,M)=\sum_{x,r}t_{x,r}e_{x,r}.
\]
Let \(\mathbf n(x)\) denote the occupation of \(x\).
Dropping the nonnegative contributions with \(y\ne x\) and using
Eq.~\eqref{eq:supp-B-Haar} together with
Theorem~\ref{thm:supp-diagonal-error} gives
\[
\begin{aligned}
\beta(\sigma,M)
&\ge\sum_{x,r}t_{x,r}e_{x,r}B(\mathbf n(x))\\
&\ge B(\mathbf b)\,[1-\alpha(\sigma,M)]
\ge(1-a_0)B(\mathbf b).
\end{aligned}
\]
Equality is attained by any computational-basis string
\(|x\rangle\in W_{\mathbf b}\), with
\(\sigma=|x\rangle\langle x|\),
\(M_0=(1-a_0)|x\rangle\langle x|\), and \(M_1=\id-M_0\).
This ancilla-free fully diagonal protocol has \(\alpha=a_0\) and
saturates Eq.~\eqref{eq:supp-fully-diagonal-ROC}.
\end{proof}

Proposition~\ref{prop:supp-qubit-measurement-example} gives a two-query
qubit example in which the fully diagonal and incoherent-marginal
restrictions have different optimal type-II errors.

\begin{proposition}[Probe--measurement separation at finite tolerance]
\label{prop:supp-qubit-measurement-example}
For \(\mathbb T_2\)-versus-\(\mathrm U(2)\) testing with \(m=2\) and
type-I-error tolerance \(a_0\in[1/5,1]\), the following hold.
\par\noindent\textup{(i)} An ancilla-free protocol with an incoherent system
marginal attains the unrestricted optimum.
\par\noindent\textup{(ii)} The fully diagonal optimum is \((1-a_0)/3\);
for \(a_0<1\), no fully diagonal protocol attains the unrestricted optimum.
\end{proposition}
\begin{proof}
Take
\[
\sigma=|01\rangle\langle01|,\qquad
|\phi\rangle=\frac{2|01\rangle+|10\rangle}{\sqrt5},\qquad
M_0=|\phi\rangle\langle\phi|,\quad M_1=\id-M_0.
\]
The probe is incoherent in the computational basis, whereas \(M_0\) is
nondiagonal.
The torus twirl fixes \(\sigma\), while
\[
\mathcal T_{\mathrm U(2)}(\sigma)=\frac{\Pi_{\mathrm{sym}}}{6}
+\frac{\Pi_{\mathrm{asym}}}{2}.
\]
Since \(|\langle\phi|01\rangle|^2=4/5\), and the symmetric and
antisymmetric weights of \(|\phi\rangle\) are \(9/10\) and \(1/10\),
respectively, this protocol has \(\alpha=\beta=1/5\).

For \(a_0\in[1/5,1]\), replace the accepting effect by
\[
M_0^{(a_0)}:=\frac{5(1-a_0)}4|\phi\rangle\langle\phi|,
\qquad
M_1^{(a_0)}:=\id-M_0^{(a_0)}.
\]
The prefactor lies in \([0,1]\), and the resulting ancilla-free protocol
with the same incoherent probe \(\sigma\) has \(\alpha=a_0\) and
\(\beta=(1-a_0)/4\), the unrestricted optimum at tolerance \(a_0\)
established in Ref.~\cite{hayashi2025predicting}.
Proposition~\ref{prop:supp-fully-diagonal-ROC} and \(B(1,1)=1/3\) from
Eq.~\eqref{eq:supp-qubit-error-values} give the fully diagonal optimum
\((1-a_0)/3\), which is strictly larger than \((1-a_0)/4\) whenever
\(a_0<1\). This proves the proposition.
\end{proof}

The construction establishes incoherent-marginal attainment for
\(a_0\ge1/5\), but does not determine the smallest such tolerance.

\section{Real attainability}
\label{sec:supp-orthogonal-real}

This section proves that every closed subgroup
\(G\subseteq\mathrm U(d)\) invariant under entrywise complex conjugation
admits a real optimal parallel protocol at every query number and every
average type-I-error tolerance.

Let \(A\) and \(B\) denote the input and output copies of
\((\mathbb C^d)^{\otimes m}\), equipped with their computational product
bases. All transposes below are taken in these fixed bases. Fix also a
reference basis on the ancillary system. A protocol is real if its probe state
and both POVM effects have real matrix entries in the resulting product basis.

For \(X:A\to B\), define the unnormalized output--input vectorization
\begin{equation}
|X\rangle\rangle_{BA}
:=
\sum_{j,k=1}^{d^m}
\langle j|X|k\rangle
|j\rangle_B|k\rangle_A.
\label{eq:supp-vectorization-convention}
\end{equation}
On \(B\otimes A\), define the performance operators
\begin{equation}
\begin{aligned}
\Omega_G
&:=
\int_G
|g^{\otimes m}\rangle\rangle
\langle\langle g^{\otimes m}|
\,\dd\mu_G(g),
\\
\Omega_{\mathrm U}
&:=
\int_{\mathrm U(d)}
|U^{\otimes m}\rangle\rangle
\langle\langle U^{\otimes m}|
\,\dd\mu_{\mathrm{Haar}}(U).
\end{aligned}
\label{eq:supp-orthogonal-performance-operators}
\end{equation}

A parallel acceptance tester is a pair \((T,\tau)\), where \(\tau\) is a
density operator on \(A\) and \(T\) is a positive operator on
\(B\otimes A\) satisfying \(T\preceq\id_B\otimes\tau\). The
complementary tester effect is \(\id_B\otimes\tau-T\). For every
\(V\in\mathrm U(d)\), the tester's acceptance probability is
\(\langle\!\langle V^{\otimes m}|T|V^{\otimes m}\rangle\!\rangle\); the two
tester effects give probabilities summing to one because
\[
\langle\!\langle V^{\otimes m}|
(\id_B\otimes\tau)
|V^{\otimes m}\rangle\!\rangle
=
\Tr\tau
=
1.
\]

Ancilla-assisted channel measurements are described by 1-testers, and every
1-tester admits a physical realization \cite{Chiribella2009Networks}. In the
single-channel setting, Ziman gives an equivalent correspondence in process-POVM
language \cite{Ziman2008PPOVM}. Chiribella et al.~\cite{Chiribella2009Networks} use
\(p_i(C)=\Tr(P_i^{\mathsf T}C)\) for the Choi operator \(C\) of the
measured channel, with \(\sum_iP_i=\id_B\otimes\varsigma\). Setting
\(T_i:=P_i^{\mathsf T}\) and \(\tau:=\varsigma^{\mathsf T}\) yields
\(p_i(C)=\Tr(T_iC)\) and
\(\sum_iT_i=\id_B\otimes\tau\) in our convention. Consequently, if
\(V^{\otimes m}\) is viewed as a single channel from \(A\) to \(B\), the
minimum type-II error over \(m\)-query parallel protocols satisfying
\(\alpha\le a_0\) is the optimal value of the following tester SDP,
where \(T\) is the effect corresponding to acceptance of the subgroup
hypothesis:
\begin{equation}
\begin{array}{ll}
\underset{T,\tau}{\textnormal{minimize}}
&
\Tr(T\Omega_{\mathrm U})
\\[1mm]
\textnormal{subject to}
&
\Tr(T\Omega_G)\ge1-a_0,
\\
&
0\preceq T\preceq\id_B\otimes\tau,
\\
&
\tau\succeq0,
\qquad
\Tr\tau=1.
\end{array}
\label{eq:supp-orthogonal-tester-sdp}
\end{equation}
The acceptance constraint enforces \(\alpha\le a_0\), while the objective is
the type-II error. The feasible set contains the always-accept tester and is closed and
bounded in a finite-dimensional space, hence compact; the continuous
objective therefore attains its minimum.

\begin{proposition}[Real attainability under conjugation closure]
\label{prop:supp-conjugation-real}
Let \(d\ge2\), and let \(G\subseteq\mathrm U(d)\) be a closed subgroup
satisfying \(\{\overline g:g\in G\}=G\), where the bar denotes
entrywise complex conjugation in the fixed basis. For every \(m\ge1\) and \(a_0\in[0,1]\), there exists a real
\(m\)-query parallel protocol \((\sigma,M)\) for
\(G\)-versus-\(\mathrm U(d)\) testing, with \(M=\{M_0,M_1\}\), such that
\(\alpha(\sigma,M)\le a_0\) and \(\beta(\sigma,M)\) is the minimum
type-II error among all \(m\)-query parallel protocols satisfying
this constraint.
\end{proposition}

\begin{proof}
\emph{Real optimal tester.}
Let \((T,\tau)\) attain the minimum in
Eq.~\eqref{eq:supp-orthogonal-tester-sdp}, and write \(\overline X\) for the
entrywise complex conjugate of \(X\) in the fixed product basis. Since
\(g\mapsto\overline g\) is a continuous automorphism of \(G\),
normalized Haar measure is invariant under this map. Vectorization
therefore gives \(\overline{\Omega_G}=\Omega_G\); the same argument gives
\(\overline{\Omega}_{\mathrm U}=\Omega_{\mathrm U}\).

Entrywise conjugation preserves positivity, trace normalization, and the
operator inequality defining a tester. Thus
\((\overline T,\overline\tau)\) is feasible and has the same objective value
as \((T,\tau)\). Averaging gives
\begin{equation}
\widetilde T
:=
\frac{T+\overline T}{2},
\qquad
\widetilde\tau
:=
\frac{\tau+\overline\tau}{2}.
\label{eq:supp-real-tester-average}
\end{equation}
The averaged pair is feasible, optimal, and entrywise real.
Relabeling this pair as \((T,\tau)\), we obtain a real minimizer of
Eq.~\eqref{eq:supp-orthogonal-tester-sdp}.

\emph{Real protocol realization.}
Since
\(0\preceq T\preceq\id_B\otimes\tau\),
the operator \(T\) is supported on
\(B\otimes\operatorname{supp}\tau\).
Let \(P_\tau\) denote the orthogonal projector onto
\(\operatorname{supp}\tau\), set \(S:=\tau^{1/2}\), and let \(S^+\) be the
Moore--Penrose inverse of \(S\), equal to \(S^{-1}\) on
\(\operatorname{supp}\tau\) and zero on \(\ker\tau\).
Because \(\tau\) is real symmetric, \(S\), \(S^+\), and \(P_\tau\) are
real symmetric. They satisfy
\[
S^+S
=
SS^+
=
P_\tau,
\qquad
T
=
(\id_B\otimes P_\tau)
T
(\id_B\otimes P_\tau).
\]

Fix a basis-preserving identification \(R\simeq A\), where \(R\) carries a
copy of the computational product basis, and regard \(S\) as a map
\(R\to A\). Define
\begin{equation}
|\psi_\tau\rangle_{AR}
:=
|S\rangle\rangle_{AR},
\qquad
\sigma
:=
|\psi_\tau\rangle
\langle\psi_\tau|.
\label{eq:supp-real-probe}
\end{equation}
The probe is normalized since
\(
\langle\psi_\tau|\psi_\tau\rangle
=
\Tr(S^\dagger S)
=
\Tr\tau
=
1
\),
and its system marginal is
\(\Tr_R\sigma=SS^\dagger=\tau\). Since \(S\) is real,
\(|\psi_\tau\rangle\) and \(\sigma\) are real in the fixed basis.

Under this identification, regard \(T\), \(S^+\), and \(P_\tau\) as acting
on the corresponding \(B\otimes R\) factors. Define
\begin{equation}
M_0
:=
(\id_B\otimes S^+)
T
(\id_B\otimes S^+),
\qquad
M_1
:=
\id_{B\otimes R}-M_0.
\label{eq:supp-real-povm}
\end{equation}
The operator \(M_0\) is positive, and the constraint
\(T\preceq\id_B\otimes\tau\) gives
\[
M_0
\preceq
\id_B\otimes S^+\tau S^+
=
\id_B\otimes P_\tau
\preceq
\id_{B\otimes R}.
\]
Hence \(M_1\succeq0\), and \(M=\{M_0,M_1\}\) is a binary POVM.
Both effects are real.

For \(V\in\mathrm U(d)\), vectorization gives
\[
(V^{\otimes m}\otimes\id_R)|\psi_\tau\rangle
=
|V^{\otimes m}S\rangle\rangle
=
(\id_B\otimes S^{\mathsf T})
|V^{\otimes m}\rangle\rangle
=
(\id_B\otimes S)
|V^{\otimes m}\rangle\rangle,
\]
where the last equality uses the real symmetry of \(S\). The support
relation and \(S^+S=SS^+=P_\tau\) yield
\[
(\id_B\otimes S)
M_0
(\id_B\otimes S)
=
(\id_B\otimes P_\tau)
T
(\id_B\otimes P_\tau)
=
T.
\]
Therefore, for every \(V\in\mathrm U(d)\),
\begin{equation}
\Tr\!\left[
M_0
(V^{\otimes m}\otimes\id_R)
\sigma
(V^{\otimes m}\otimes\id_R)^\dagger
\right]
=
\langle\!\langle V^{\otimes m}|
T
|V^{\otimes m}\rangle\!\rangle .
\label{eq:supp-tester-protocol-statistics}
\end{equation}
Thus the physical protocol reproduces the tester's acceptance
probability pointwise. Averaging over \(G\) and \(\mathrm U(d)\), and using the
optimality of \((T,\tau)\), shows that the protocol satisfies
\(\alpha(\sigma,M)\le a_0\) and attains the minimum type-II error under
this constraint.
The construction supplies one real optimal protocol, which uses an
ancilla; it does not assert that every optimal protocol is real or
that the optimum can be attained by a real ancilla-free protocol.
\end{proof}

Taking \(G=\mathrm O(d)\) and \(a_0=0\) proves
Theorem~\ref{thm:orthogonal-real-attainability}.

\section{Qutrit-Clifford representation theory and type selection}
\label{sec:supp-clifford-representation}

This section develops the representation-theoretic input for
qutrit-Clifford testing. We first classify the possible effective
Clifford types, then determine their actual occurrence through branching
multiplicities, and finally identify the type maximizing the
representation-theoretic score.

\subsection{Effective Clifford types}

Fix the computational basis
\(\{|j\rangle:j\in\mathbb F_3\}\), let
\(\omega:=e^{2\pi i/3}\), and interpret all basis labels and exponents
modulo \(3\). We use the qutrit Clifford
generators~\cite{HostensDehaeneDeMoor2005}
\[
X|j\rangle=|j+1\rangle,
\qquad
Z|j\rangle=\omega^j|j\rangle,
\qquad
P=\operatorname{diag}(1,1,\omega),
\qquad
F|j\rangle
=
\frac{1}{\sqrt3}
\sum_{k\in\mathbb F_3}
\omega^{jk}|k\rangle .
\]
Thus
\[
\mathcal C_3
=
\langle X,Z,P,F\rangle
\le\mathrm U(3),
\qquad
\overline{\mathcal C}_3
:=
\mathcal C_3/
\bigl(\mathcal C_3\cap\mathrm U(1)\id_3\bigr),
\]
where
\(
\mathcal C_3\cap\mathrm U(1)\id_3
=
\langle e^{i\pi/6}\id_3\rangle
\).

For \(m\ge1\), set
\[
X_m:=X^{\otimes m},
\qquad
Z_m:=Z^{\otimes m},
\qquad
P_m:=P^{\otimes m},
\qquad
F_m:=F^{\otimes m}.
\]
The tensor-image group, its scalar subgroup, and its effective quotient are
\begin{equation}
\begin{aligned}
\widetilde G_m
&:=
\{C^{\otimes m}:C\in\mathcal C_3\}
=
\langle X_m,Z_m,P_m,F_m\rangle,
\\
\mathcal Z_m
&:=
\widetilde G_m\cap\mathrm U(1)\id,
\qquad
\bar G_m
:=
\widetilde G_m/\mathcal Z_m .
\end{aligned}
\label{eq:supp-clifford-effective-quotient}
\end{equation}
We distinguish the \(m\)-dependent groups
\(\widetilde G_m,\mathcal Z_m,\bar G_m\) from the corresponding single-copy
groups \(\mathcal C_3,\overline{\mathcal C}_3\). The
physical tensor representation has the tautological central character
\[
\zeta_m(c\id):=c,
\qquad
c\id\in\mathcal Z_m.
\]

The possible effective types fall into two families, according to the central
character carried by the tensor-power displacement subgroup. When \(3\mid m\), the tensor-power
displacement generators commute. When \(3\nmid m\), they instead generate
a finite Heisenberg subgroup with a fixed nontrivial central character.
We write
\begin{equation}
\begin{aligned}
\mathcal F_0
&:=
\{
\chi_0,\chi_1,\chi_2,
\tau_{01},\tau_{02},\tau_{12},
\nu_3,
\Gamma_0,\Gamma_1,\Gamma_2
\},
\\
\mathcal F_1
&:=
\{
\eta_3^{(0)},\eta_3^{(1)},\eta_3^{(2)},
\eta_6^{(01)},\eta_6^{(02)},\eta_6^{(12)},
\eta_9
\}.
\end{aligned}
\label{eq:supp-clifford-type-families}
\end{equation}
For the residual \(\mathrm{SL}(2,\mathbb F_3)\)-types, we take
\(\tau_{12}\) as the faithful two-dimensional representation with
trivial determinant character and use
\[
\tau_{01}
=
\tau_{12}\otimes\chi_2,
\qquad
\tau_{02}
=
\tau_{12}\otimes\chi_1.
\]
The representations \(\Gamma_t\) are the three eight-dimensional types
induced from the nonzero displacement-character orbit. For
\(3\nmid m\), the subscripts on
\(\eta_3^{(i)},\eta_6^{(ij)},\eta_9\) record their dimensions, whereas the
superscripts record the residual
\(\mathrm{SL}(2,\mathbb F_3)\)-type. Their dependence on
\(\kappa=m\bmod3\) is suppressed.

\begin{proposition}[Possible effective Clifford types]
\label{prop:supp-clifford-type-families}
For every \(m\ge1\), the set \(\mathcal I_m\) of effective Clifford
types occurring at \(m\) queries satisfies
\[
\mathcal I_m
\subseteq
\begin{cases}
\mathcal F_0,
&
3\mid m,
\\[1mm]
\mathcal F_1,
&
3\nmid m.
\end{cases}
\]
\end{proposition}

\begin{proof}
Write
\(\bar x,\bar z,\bar p,\bar f\) for the images of
\(X_m,Z_m,P_m,F_m\) in \(\bar G_m\).

\emph{The case \(3\mid m\).}
The Weyl relation gives
\[
Z_mX_m
=
\omega^mX_mZ_m
=
X_mZ_m.
\]
Thus \(\bar x\) and \(\bar z\) commute and have order \(3\). Let
\[
N:=\langle\bar x,\bar z\rangle.
\]
If \(\bar x^a\bar z^b=1\), then \(X_m^aZ_m^b\) is scalar. For
\(a\ne0\), this operator moves every computational-basis vector and
therefore cannot be scalar. Hence \(a=0\). Comparing the eigenvalues of
\(Z_m^b\) on \(|0\cdots0\rangle\) and \(|10\cdots0\rangle\) then gives
\(b=0\). Consequently,
\[
N\simeq\mathbb F_3^2.
\]

The conjugation relations
\[
PXP^{-1}=XZ,
\qquad
PZP^{-1}=Z,
\qquad
FXF^{-1}=Z,
\qquad
FZF^{-1}=X^{-1}
\]
show that \(N\) is normal. With exponent vectors written as row vectors,
the actions induced by \(\bar p\) and \(\bar f\) are
\begin{equation}
u
=
\begin{pmatrix}
1&1\\
0&1
\end{pmatrix},
\qquad
s
=
\begin{pmatrix}
0&1\\
-1&0
\end{pmatrix}.
\label{eq:supp-clifford-sl2-action}
\end{equation}
These matrices generate \(\mathrm{SL}(2,\mathbb F_3)\), so conjugation
induces a surjection
\(
\bar G_m/N
\twoheadrightarrow
\mathrm{SL}(2,\mathbb F_3)
\).

Conversely, the standard structure theorem for projective Clifford
groups in odd prime dimension identifies the projective qutrit Clifford
group as the semidirect product of its displacement and symplectic
subgroups \cite{Appleby2005Clifford}. Thus it fits into the split exact
sequence
\[
1
\longrightarrow
\mathbb F_3^2
\longrightarrow
\overline{\mathcal C}_3
\longrightarrow
\mathrm{SL}(2,\mathbb F_3)
\longrightarrow
1.
\]
The tensor-power map
\([C]\mapsto[C^{\otimes m}]\) gives a surjection
\(\overline{\mathcal C}_3\to\bar G_m\), whose restriction to the
displacement subgroup is an isomorphism onto \(N\). It therefore induces
a surjection in the opposite direction,
\(
\mathrm{SL}(2,\mathbb F_3)
\twoheadrightarrow
\bar G_m/N
\).
Since both groups are finite, these two surjections are isomorphisms. The
image of a symplectic section supplies a complement \(L\) to \(N\), and
hence
\[
\bar G_m
=
N\rtimes L
\simeq
\mathbb F_3^2
\rtimes
\mathrm{SL}(2,\mathbb F_3).
\]

Consider the action of \(L\) on the character group \(\widehat N\). The
determinant pairing identifies both \(N\) and \(\widehat N\) with
\(\mathbb F_3^2\), and the induced \(L\)-action is the
standard action of \(\mathrm{SL}(2,\mathbb F_3)\). This action has exactly
two orbits: the zero character and the eight nonzero characters. Indeed,
any two nonzero vectors can be completed to bases with equal determinant,
and the corresponding change of basis belongs to
\(\mathrm{SL}(2,\mathbb F_3)\).

The zero orbit gives the irreducible \(L\)-modules inflated to
\(N\rtimes L\), with \(N\) acting trivially. For
\(i\in\{0,1,2\}\), let
\[
\chi_i(u)=\omega^i,
\qquad
\chi_i(s)=1.
\]
Together with the three two-dimensional types
\(\tau_{01},\tau_{02},\tau_{12}\) and the unique
three-dimensional type \(\nu_3\), these characters give all irreducible
\(L\simeq\mathrm{SL}(2,\mathbb F_3)\)-modules. Their squared dimensions
sum to \(3\cdot1^2+3\cdot2^2+3^2=24=|L|\).

For the nonzero orbit, fix the character
\(
\xi(\bar x^a\bar z^b)
:=
\omega^a
\).
Its stabilizer \(L_\xi=\langle u\rangle\) is cyclic of order \(3\). For \(t=0,1,2\), let
\(\psi_t(u):=\omega^t\) and define
\[
\Gamma_t
:=
\operatorname{Ind}_{N\rtimes L_\xi}^{N\rtimes L}
(\xi\otimes\psi_t).
\]
Then
\[
\dim\Gamma_t
=
[L:L_\xi]
=
8,
\qquad
\left.\Gamma_t\right|_N
\simeq
\bigoplus_{\theta\in L\cdot\xi}\theta.
\]
An endomorphism commuting with \(N\) is diagonal on these eight distinct
weight spaces. Commutation with the transitive \(L\)-action forces all
diagonal entries to agree. Hence
\(
\operatorname{End}_{N\rtimes L}(\Gamma_t)
=
\mathbb C\id
\), so \(\Gamma_t\) is irreducible.

If \(t\ne t'\), an intertwiner
\(\Gamma_t\to\Gamma_{t'}\) restricts on the line carrying \(\xi\) to an
intertwiner between \(\psi_t\) and \(\psi_{t'}\), and hence vanishes on
that line. Since the induced module is generated by its \(\xi\)-weight
line, the intertwiner vanishes identically. Thus
\(\Gamma_0,\Gamma_1,\Gamma_2\) are pairwise inequivalent.

The zero and nonzero orbits therefore give the ten irreducible
\(\bar G_m\)-modules underlying the labels in \(\mathcal F_0\). Their
squared dimensions sum to
\(3\cdot1^2+3\cdot2^2+3^2+3\cdot8^2=216=|\bar G_m|\), so the list is
complete.

It remains to restore the central character of the physical tensor
representation. Write \(m=3k\) and define
\[
\delta_m(C^{\otimes m})
:=
\det(C)^k.
\]
The character \(\delta_m\) is well defined. If
\(C^{\otimes m}=D^{\otimes m}\), then
\(D^{-1}C=\alpha\id_3\) for some \(\alpha\) with \(\alpha^m=1\), and
\(\det(D^{-1}C)^k=\alpha^{3k}=1\). Moreover,
\(\delta_m|_{\mathcal Z_m}=\zeta_m\).

Let
\[
q_m:
\widetilde G_m
\twoheadrightarrow
\bar G_m
\twoheadrightarrow
\bar G_m/N
\simeq
L
\]
be the composite quotient map. With indices on \(\chi_i\) understood
modulo \(3\), define
\begin{equation}
\varepsilon_m
:=
\delta_m\cdot
\bigl(\chi_{-k}\circ q_m\bigr).
\label{eq:supp-clifford-central-untwisting}
\end{equation}
A direct evaluation on the tensor-power generators gives
\(\varepsilon_m(X_m)=\varepsilon_m(Z_m)=\varepsilon_m(P_m)=1\) and
\(\varepsilon_m(F_m)=i^{-k}\).
Since \(q_m\) is trivial on \(\mathcal Z_m\),
\(\varepsilon_m|_{\mathcal Z_m}=\zeta_m\). Twisting a physical
\(\widetilde G_m\)-module by \(\varepsilon_m^{-1}\) therefore removes
its scalar action and produces a \(\bar G_m\)-module. Conversely,
inflation followed by tensoring with \(\varepsilon_m\) restores the
physical central character. Under the label convention fixed above, this
correspondence identifies every occurring type with an element of
\(\mathcal F_0\).

\emph{The case \(3\nmid m\).}
Let
\(\kappa=m\bmod3\in\{1,2\}\). The displacement generators now satisfy
\begin{equation}
Z_mX_m
=
\omega^\kappa X_mZ_m.
\label{eq:supp-clifford-heisenberg-relation}
\end{equation}
Thus
\[
H_m
:=
\langle\omega\id,X_m,Z_m\rangle
\triangleleft
\widetilde G_m
\]
is a nonabelian Heisenberg group of order \(27\), with center
\(\langle\omega\id\rangle\). The physical tensor representation carries
the nontrivial central character
\(\omega\id\mapsto\omega\). By the finite Stone--von Neumann theorem,
the unique irreducible \(H_m\)-module with this central character is
\(V^{(\kappa)}\), with \(\dim V^{(\kappa)}=3\).
Let \(W\) be an irreducible \(\widetilde G_m\)-constituent of
\((\mathbb C^3)^{\otimes m}\). By Clifford's theorem on restrictions of
irreducible representations to normal subgroups, the \(H_m\)-module
\(W|_{H_m}\) is a direct sum of \(\widetilde G_m\)-conjugates of a single irreducible
\(H_m\)-module. The scalar \(\omega\id\) is central in
\(\widetilde G_m\), so these conjugates have the same central character
and are all isomorphic to \(V^{(\kappa)}\).

To make the residual action explicit, fix a Schur sector \(U_\lambda\)
and put
\[
E_r^\lambda
:=
\ker(Z_m-\omega^r\id)\cap U_\lambda,
\qquad
r\in\mathbb F_3.
\]
Equation~\eqref{eq:supp-clifford-heisenberg-relation} gives
\(
X_mE_r^\lambda
=
E_{r+\kappa}^\lambda
\), so the three eigenspaces \(E_r^\lambda\) are mutually orthogonal and
have equal dimension. Let \(M_\lambda:=E_0^\lambda\) and define
\[
\iota_\lambda:
\mathbb C^3\otimes M_\lambda
\longrightarrow
U_\lambda,
\qquad
\iota_\lambda(|a\rangle\otimes v)
:=
X_m^av.
\]
This map is unitary and satisfies
\[
\iota_\lambda^\dagger X_m\iota_\lambda
=
X\otimes\id_{M_\lambda},
\qquad
\iota_\lambda^\dagger Z_m\iota_\lambda
=
Z^\kappa\otimes\id_{M_\lambda}.
\]

Let \(J|j\rangle:=|-j\rangle\) and choose
\[
P_\kappa
:=
\begin{cases}
P,&\kappa=1,\\
P^2,&\kappa=2,
\end{cases}
\qquad
F_\kappa
:=
\begin{cases}
F,&\kappa=1,\\
JF,&\kappa=2.
\end{cases}
\]
Define
\[
\vartheta_{m,\kappa}
:=
\begin{cases}
i^{-(m-1)/3},
&
\kappa=1,
\\
i^{-(m+1)/3},
&
\kappa=2,
\end{cases}
\qquad
\widetilde F_{m,\kappa}
:=
\vartheta_{m,\kappa}F_\kappa.
\]
The exponents are integers because \(3\mid(m-1)\) when \(\kappa=1\),
whereas \(3\mid(m+1)\) when \(\kappa=2\). These first-factor operators
implement the same automorphisms of the Heisenberg generators as
\(P_m,F_m\):
\[
\begin{aligned}
P_\kappa XP_\kappa^{-1}
&=
XZ^\kappa,
&
P_\kappa Z^\kappa P_\kappa^{-1}
&=
Z^\kappa,
\\
\widetilde F_{m,\kappa}X
\widetilde F_{m,\kappa}^{-1}
&=
Z^\kappa,
&
\widetilde F_{m,\kappa}Z^\kappa
\widetilde F_{m,\kappa}^{-1}
&=
X^{-1}.
\end{aligned}
\]
They also obey
\[
P_\kappa^3=\id,
\qquad
\widetilde F_{m,\kappa}^4=\id,
\qquad
\widetilde F_{m,\kappa}^2P_\kappa
\widetilde F_{m,\kappa}^{-2}P_\kappa^{-1}
=
Z^\kappa,
\]
and the scalar correction makes the two cubes carry the same phase:
\[
(\widetilde F_{m,\kappa}P_\kappa)^3
=
e^{im\pi/6}\id_{\mathbb C^3},
\qquad
(F_mP_m)^3
=
e^{im\pi/6}\id_{(\mathbb C^3)^{\otimes m}}.
\]

Transport \(P_m,F_m\) through \(\iota_\lambda\):
\[
\widehat P_\lambda
:=
\iota_\lambda^\dagger P_m\iota_\lambda,
\qquad
\widehat F_\lambda
:=
\iota_\lambda^\dagger F_m\iota_\lambda.
\]
Left-multiplying \(\widehat P_\lambda\) and
\(\widehat F_\lambda\) by the inverses of their respective first-factor
operators yields operators that commute with both
\(X\otimes\id\) and \(Z^\kappa\otimes\id\). Since \(X\) and
\(Z^\kappa\) generate \(M_3(\mathbb C)\), their joint commutant is
\(\id_{\mathbb C^3}\otimes\operatorname{End}(M_\lambda)\). Hence there
are unique \(A_\lambda,B_\lambda\in\mathrm U(M_\lambda)\) such that
\begin{equation}
\widehat P_\lambda
=
P_\kappa\otimes A_\lambda,
\qquad
\widehat F_\lambda
=
\widetilde F_{m,\kappa}\otimes B_\lambda.
\label{eq:supp-clifford-sector-factorization}
\end{equation}

Comparing the tensor Clifford relations with the first-factor relations
gives
\[
A_\lambda^3=\id,
\qquad
B_\lambda^4=\id,
\qquad
B_\lambda^2A_\lambda=A_\lambda B_\lambda^2,
\qquad
(B_\lambda A_\lambda)^3=\id.
\]
These are defining relations for the generators \(u,s\) in
Eq.~\eqref{eq:supp-clifford-sl2-action}: \(s^2\) is central, and
quotienting by \(\langle s^2\rangle\) gives the standard presentation
\(\langle u,s\mid u^3=s^2=(su)^3=1\rangle\simeq A_4\), so the presented
group has order at most \(24\). It surjects onto
\(\mathrm{SL}(2,\mathbb F_3)\), which also has order \(24\), and hence the
surjection is an isomorphism. Thus
\(u\mapsto A_\lambda,\ s\mapsto B_\lambda\)
defines a unitary representation of
\(\mathrm{SL}(2,\mathbb F_3)\) on \(M_\lambda\).

The Heisenberg subgroup acts irreducibly on the
\(\mathbb C^3\) factor and trivially on \(M_\lambda\). Consequently, an
\(H_m\)-invariant subspace has the form
\(\mathbb C^3\otimes K\), and it is invariant under the full tensor-image
group precisely when \(K\) is invariant under \(A_\lambda,B_\lambda\).
Since every occurring constituent belongs to some Schur sector
\(U_\lambda\), the classification reduces to the irreducible
constituents of this residual
\(\mathrm{SL}(2,\mathbb F_3)\)-module.

Tensoring the residual irreducibles with \(V^{(\kappa)}\) gives
\[
\eta_3^{(i)}
:=
V^{(\kappa)}\otimes\chi_i,
\qquad
\eta_6^{(ij)}
:=
V^{(\kappa)}\otimes\tau_{ij},
\qquad
\eta_9
:=
V^{(\kappa)}\otimes\nu_3,
\]
where \(i\in\{0,1,2\}\), while the pair label \(ij\) ranges over
\(01,02,12\). These modules have dimensions \(3,6,9\), respectively,
and exhaust the possible effective types in this central sector.
Therefore every occurring type belongs to \(\mathcal F_1\).
\end{proof}

\subsection{Branching multiplicities}

To determine which members of the families in
Proposition~\ref{prop:supp-clifford-type-families} actually occur, we
restrict each Schur sector \(U_\lambda\) in
Eq.~\eqref{eq:schur-weyl-branching} to the tensor-image group
\(\widetilde G_m\) and denote the resulting branching multiplicities by
\(n_{\eta,\lambda}\).

For each candidate effective type \(\eta\) in the family appropriate to
\(m\), the preceding subsection fixes a compatible linear lift
\(\widehat V_\eta\) to \(\widetilde G_m\) in the physical central
sector, regardless of whether \(\eta\) occurs. By construction, every
\(z\in\mathcal Z_m\) acts on \(\widehat V_\eta\) as
\(\zeta_m(z)\id\). Write
\(\widehat\chi_\eta:\widetilde G_m\to\mathbb C\) for its character.
The branching multiplicity of \(\widehat V_\eta\) in \(U_\lambda\) is
\(n_{\eta,\lambda}
=\dim\operatorname{Hom}_{\widetilde G_m}
(\widehat V_\eta,U_\lambda)\).
The ordinary finite-group character inner product then gives
\[
n_{\eta,\lambda}
=
\frac{1}{|\widetilde G_m|}
\sum_{g\in\widetilde G_m}
\chi_{U_\lambda}(g)\,
\overline{\widehat\chi_\eta(g)}.
\]

For \(C\in\mathcal C_3\), let
\(x_1(C),x_2(C),x_3(C)\) be the eigenvalues of \(C\) in the defining
qutrit representation. For every partition
\(\lambda\vdash m\) with \(\ell(\lambda)\le3\), Schur--Weyl duality
identifies the character of \(U_\lambda\) with the Schur polynomial:
\[
\chi_{U_\lambda}(C^{\otimes m})
=
s_\lambda\!\left(x_1(C),x_2(C),x_3(C)\right).
\]

To pull the normalized character average back to the single-copy
Clifford group, let
\(\pi_m:\mathcal C_3\to\widetilde G_m\) be the surjective homomorphism
\(\pi_m(C)=C^{\otimes m}\), with
\(\ker\pi_m=\{C\in\mathcal C_3:C^{\otimes m}=\id\}\).
Since every fiber is a coset of \(\ker\pi_m\), all fibers have the same
cardinality. Consequently,
\(|\mathcal C_3|=|\ker\pi_m|\,|\widetilde G_m|\), and for every
function \(f\) on \(\widetilde G_m\),
\[
\frac{1}{|\widetilde G_m|}
\sum_{g\in\widetilde G_m}f(g)
=
\frac{1}{|\mathcal C_3|}
\sum_{C\in\mathcal C_3}f(C^{\otimes m}).
\]
Applying this identity to the summand of the character inner product and using
the Schur--Weyl character formula gives
\begin{equation}
n_{\eta,\lambda}
=
\frac{1}{|\mathcal C_3|}
\sum_{C\in\mathcal C_3}
s_\lambda\!\left(
x_1(C),x_2(C),x_3(C)
\right)
\overline{\widehat\chi_\eta(C^{\otimes m})}.
\label{eq:supp-clifford-branching}
\end{equation}

Equation~\eqref{eq:supp-clifford-branching} makes the occurrence
criterion explicit:
\(\eta\in\mathcal I_m\) if and only if
\(n_{\eta,\lambda}>0\) for at least one partition
\(\lambda\vdash m\) with \(\ell(\lambda)\le3\).
For an occurring type \(\eta\),
Eq.~\eqref{eq:subgroup-type-score} gives
\[
h_m(\eta)
=
\frac{1}{d_\eta}
\sum_{\substack{\lambda\vdash m\\ \ell(\lambda)\le3}}
d_\lambda n_{\eta,\lambda}.
\]

\subsection{Global type selection}

The branching formula in
Eq.~\eqref{eq:supp-clifford-branching} reduces both the occurrence test
and the score evaluation for each candidate type to finite character
sums. Comparing these sums over the relevant type families determines
the maximizing type.

\begin{theorem}[Global type selection for qutrit-Clifford testing]
\label{thm:supp-clifford-type-selection}
For every \(m\ge5\), there is a unique effective Clifford type
\(\eta_m^\star\in\mathcal I_m\) such that
\(h_m(\eta_m^\star)=h_{\max}\). It is determined by \(m\bmod6\) as
follows:
\[
\begin{array}{c|cccccc}
m\bmod6
&0&1&2&3&4&5
\\
\hline
\eta_m^\star
&\chi_1
&\eta_3^{(0)}
&\eta_3^{(1)}
&\tau_{01}
&\eta_6^{(01)}
&\eta_3^{(1)}
\end{array}
\]
\end{theorem}

\begin{proof}
\emph{Candidate-family reduction.}
For this comparison, we extend the definition of \(h_m(\eta)\) in
Eq.~\eqref{eq:subgroup-type-score} to every candidate type in the family
appropriate to \(m\). By the preceding subsection, every candidate has a
compatible linear lift \(\widehat V_\eta\) to \(\widetilde G_m\) in the physical
central sector, with the character-inner-product formula in
Eq.~\eqref{eq:supp-clifford-branching}. If a candidate is absent,
this formula gives \(n_{\eta,\lambda}=0\) in every Schur sector, and hence
\(h_m(\eta)=0\). Since \((\mathbb C^3)^{\otimes m}\ne\{0\}\), at least one
candidate type occurs, and every occurring type has strictly positive
score. Consequently, any strict maximizer over the full candidate family
belongs to \(\mathcal I_m\).

\emph{Score generating function.}
For \(C\in\mathcal C_3\), define
\[
A_m(C)
:=
\sum_{\substack{\lambda\vdash m\\ \ell(\lambda)\le3}}
d_\lambda
s_\lambda\!\left(x_1(C),x_2(C),x_3(C)\right).
\]
Substituting Eq.~\eqref{eq:supp-clifford-branching} into
Eq.~\eqref{eq:subgroup-type-score} gives
\[
h_m(\eta)
=
\frac{1}{|\mathcal C_3|d_\eta}
\sum_{C\in\mathcal C_3}
A_m(C)\,
\overline{\widehat\chi_\eta(C^{\otimes m})}.
\]
The summand is independent of the scalar representative. If \(C\) is
replaced by \(e^{i\theta}C\), where
\(e^{i\theta}\id_3\in\mathcal C_3\), homogeneity gives
\(A_m(e^{i\theta}C)=e^{im\theta}A_m(C)\), and the tautological
central character gives
\[
\widehat\chi_\eta\!\left((e^{i\theta}C)^{\otimes m}\right)
=
e^{im\theta}\widehat\chi_\eta(C^{\otimes m}).
\]
Complex conjugation cancels the two scalar factors. Since the scalar
subgroup \(\mathcal C_3\cap\mathrm U(1)\id_3\) has order \(12\), averaging
over \(\mathcal C_3\) reduces to averaging over its projective quotient:
\begin{equation}
h_m(\eta)
=
\frac{1}{216d_\eta}
\sum_{\bar C\in\overline{\mathcal C}_3}
A_m(C)\,
\overline{\widehat\chi_\eta(C^{\otimes m})},
\label{eq:supp-clifford-score-character-sum}
\end{equation}
where \(C\) is any representative of \(\bar C\).

Because \(d_\lambda=s_\lambda(1,1,1)\), the Cauchy identity gives the
following identity of formal power series:
\begin{equation}
\sum_{m\ge0}A_m(C)t^m
=
\prod_{j=1}^{3}(1-x_j(C)t)^{-3}
=
\det(\id_3-tC)^{-3}.
\label{eq:supp-clifford-score-generating-function}
\end{equation}
Thus the asymptotic comparison reduces to the pole structure determined
by the spectrum of \(C\).

\emph{Spectral classes and asymptotic expansion.}
The \(216\) projective Clifford elements split into the following
spectral classes:
\[
\begin{array}{c|cc}
\text{class}&\text{size}&\text{representative spectrum}
\\
\hline
\text{identity}&1&(1,1,1)
\\
\mathscr P&33&(\alpha,\alpha,\beta),\ \alpha\ne\beta
\\
\mathscr P'&182&(\alpha,\beta,\gamma),\ \text{pairwise distinct}
\end{array}
\]
A unitary representative with three equal eigenvalues is scalar and
therefore belongs to the projective identity class.

For the identity representative,
\[
\sum_{m\ge0}A_m(\id_3)t^m=(1-t)^{-9},
\qquad
A_m(\id_3)=\binom{m+8}{8}.
\]
Since \(\widehat\chi_\eta(\id_3^{\otimes m})=d_\eta\), this class
contributes exactly \(\frac{1}{216}\binom{m+8}{8}\) to every score and
hence cancels in score differences.

For \(\bar C\in\mathscr P\), choose a representative with spectrum
\((\alpha_C,\alpha_C,\beta_C)\) and set
\(q_C:=\beta_C/\alpha_C\). Equation~\eqref{eq:supp-clifford-score-generating-function}
becomes
\[
\sum_{m\ge0}A_m(C)t^m
=
(1-\alpha_Ct)^{-6}(1-\beta_Ct)^{-3}.
\]
The sixth-order pole at \(t=\alpha_C^{-1}\) supplies the \(m^5\) and
\(m^4\) terms, while the third-order pole at \(t=\beta_C^{-1}\)
contributes only \(O(m^2)\). Extracting the first two orders gives
\[
A_m(C)
=
\alpha_C^m
\left[
\frac{m^5}{5!}(1-q_C)^{-3}
+
\frac{m^4}{8}(1-2q_C)(1-q_C)^{-4}
+
O(m^3)
\right]
+
O(m^2).
\]

Fix \(s\in\mathbb Z/12\mathbb Z\) and restrict to
\(m\equiv s\pmod{12}\).
For the fixed representatives with
\(\overline C\in\mathscr P\),
Table~\ref{tab:supp-clifford-lift-orders} gives
\(\alpha_C^m=\alpha_C^s\).
The lifted character value also depends on \(m\) only through
\(s\), by construction. Denote the
latter by
\(\widehat\chi_\eta^{(s)}(C):=\widehat\chi_\eta(C^{\otimes m})\).
The type-dependent coefficients are therefore
\[
\widetilde B_\eta(s)
:=
\frac{1}{216d_\eta\,5!}
\sum_{\bar C\in\mathscr P}
\alpha_C^s(1-q_C)^{-3}\,
\overline{\widehat\chi_\eta^{(s)}(C)}
\]
and
\[
C_\eta(s)
:=
\frac{1}{216d_\eta}
\sum_{\bar C\in\mathscr P}
\alpha_C^s
\frac{(1-2q_C)(1-q_C)^{-4}}{8}\,
\overline{\widehat\chi_\eta^{(s)}(C)}.
\]
These expressions are independent of the representative of \(\bar C\).
Indeed, replacing \(C\) by \(e^{i\theta}C\), where
\(e^{i\theta}\id_3\in\mathcal C_3\), multiplies
\(\alpha_C^s\overline{\widehat\chi_\eta^{(s)}(C)}\) by
\(e^{i(s-m)\theta}=1\).

Exact evaluation of the finite sums over \(\mathscr P\) gives
\(\widetilde B_\eta(s+6)=\widetilde B_\eta(s)\). We therefore write
\(B_\eta(r)\), where \(r\) is the residue of \(m\) modulo \(6\), while
\(C_\eta(s)\) retains its period-twelve argument.

For \(\bar C\in\mathscr P'\), the three eigenvalues are distinct, and
\[
\sum_{m\ge0}A_m(C)t^m
=
\prod_{\xi\in\{\alpha_C,\beta_C,\gamma_C\}}
(1-\xi t)^{-3}.
\]
Every pole has order at most three, so this class contributes
\(O(m^2)\). It does not affect the \(m^5\) or \(m^4\) comparison, but it
cannot be discarded in the exact finite-\(m\) argument. Because the spectral
classes, candidate families, and residue classes are finite, the remainders in
Eq.~\eqref{eq:supp-clifford-asymptotic-expansion} are uniform over types and
residue classes. Combining the three classes yields
\begin{equation}
h_m(\eta)
=
\frac{1}{216}\binom{m+8}{8}
+
B_\eta(r)m^5
+
C_\eta(s)m^4
+
O(m^3).
\label{eq:supp-clifford-asymptotic-expansion}
\end{equation}

\emph{Eventual period-six selection.}
For fixed \(s\) and its reduction \(r\) modulo \(6\),
Eq.~\eqref{eq:supp-clifford-asymptotic-expansion} gives
\[
h_m(\eta)-h_m(\eta')
=
\bigl(B_\eta(r)-B_{\eta'}(r)\bigr)m^5
+
\bigl(C_\eta(s)-C_{\eta'}(s)\bigr)m^4
+
O(m^3).
\]
Thus a strict difference at order \(m^5\) determines the eventual sign;
only a tie in \(B_\eta(r)\) requires the subleading coefficient.
The exact leading comparisons are summarized below. The scale column
gives the factor multiplying \(B_\eta(r)\), and the final column gives
the next smaller distinct value on that scale:
\[
{\renewcommand{\arraystretch}{1.15}
\begin{array}{c|cccc}
r
&\text{scale}
&\text{maximum}
&\text{maximizer(s)}
&\text{next distinct}
\\
\hline
0&207360&41&\chi_1&9
\\
1&622080&41&\eta_3^{(0)},\eta_3^{(1)}&23
\\
2&622080&55&\eta_3^{(1)}&25
\\
3&207360&25&\tau_{01}&23
\\
4&622080&41&\eta_6^{(01)}&23
\\
5&622080&73&\eta_3^{(1)}&9
\end{array}}
\]
The comparison is strict for \(r=0,2,3,4,5\). For \(r=1\), the only
leading tie is between \(\eta_3^{(0)}\) and \(\eta_3^{(1)}\).

For the two corresponding residue classes \(s=1,7\) modulo \(12\),
the exact subleading values are
\[
82944\,C_{\eta_3^{(0)}}(s)=155,
\qquad
82944\,C_{\eta_3^{(1)}}(s)=91,
\]
and therefore
\[
C_{\eta_3^{(0)}}(s)-C_{\eta_3^{(1)}}(s)
=
\frac{1}{1296}.
\]
Consequently,
\[
h_m\!\left(\eta_3^{(0)}\right)
-
h_m\!\left(\eta_3^{(1)}\right)
=
\frac{1}{1296}m^4+O(m^3)
\]
for \(m\equiv1\pmod6\). Together with the strict \(B\)-coefficient
comparisons, this tie resolution shows that the type listed in
Theorem~\ref{thm:supp-clifford-type-selection} strictly maximizes the score
over the full candidate family for all sufficiently large \(m\) in each
residue class modulo \(6\). For each \(r\), a common threshold is
obtained by taking the maximum over the two residues
\(s\equiv r\pmod6\) and the finitely many competing types. By the
candidate-family reduction, this strict maximizer belongs to
\(\mathcal I_m\) and is the unique occurring type attaining the maximum.

These asymptotic comparisons do not yet prove the uniform threshold
\(m\ge5\), because lower-order terms, including the simple-spectrum
contribution, can affect finitely many query numbers.

\emph{Exact threshold certificate.}
The lifted character factors depend only on \(m\bmod12\).
By Table~\ref{tab:supp-clifford-lift-orders}, every
\(C\in\mathcal C_3\) satisfies \(C^{12}=\id_3\) when
\(\overline C\in\mathscr P\), and \(C^{36}=\id_3\) when
\(\overline C\in\mathscr P'\).
Consequently, fixing \(m\bmod36\) fixes every phase in the
exact score formula.

Fix \(c\in\{0,\ldots,35\}\) and write
\begin{equation}
m=c+36n,
\qquad
n\in\mathbb Z_{\ge0}.
\label{eq:supp-clifford-mod36-parametrization}
\end{equation}
For fixed \(c\), call \(n\in\mathbb Z_{\ge0}\) admissible if
\(c+36n\ge1\), and call \(m=c+36n\) the corresponding admissible
query number.
At admissible query numbers along this progression, the relevant candidate
family is constant because \(3\mid36\), and the candidate listed in the
theorem is constant because \(6\mid36\). Denote the latter candidate by
\(\eta_c^\star\).

The exact partial-fraction forms underlying
Eq.~\eqref{eq:supp-clifford-score-generating-function} show that, for
every candidate \(\eta\) in the relevant family, there is a polynomial
\(P_{\eta,c}\) of degree at most \(5\) such that
\begin{equation}
h_{c+36n}(\eta)
=
\frac{1}{216}\binom{c+36n+8}{8}
+
P_{\eta,c}(n),
\label{eq:supp-clifford-residue-polynomial}
\end{equation}
for every \(n\ge0\) for which \(c+36n\ge1\). Indeed, the identity class
gives the displayed binomial term. Once the phases are fixed, the
contribution from \(\mathscr P\) is a polynomial in \(n\) of degree at most
\(5\), and the simple-spectrum contribution has degree at most \(2\).
This is an exact identity; no asymptotic remainder is used.

These classwise formulas initially produce \(P_{\eta,c}\) over a cyclotomic
field. At every admissible query number, however,
\[
h_m(\eta)
=
\frac{1}{d_\eta}
\sum_{\substack{\lambda\vdash m\\ \ell(\lambda)\le3}}
d_\lambda n_{\eta,\lambda}
\in\mathbb Q.
\]
The binomial term in Eq.~\eqref{eq:supp-clifford-residue-polynomial} is
also rational. Since \(\deg P_{\eta,c}\le5\), its values at six distinct
admissible integers, starting at \(n=1\) when \(c=0\), determine its
coefficients through a rational Vandermonde system. Hence
\(P_{\eta,c}\in\mathbb Q[x]\).

For every candidate \(\nu\ne\eta_c^\star\) in the relevant family, define
\begin{equation}
\Delta_{\nu,c}
:=
P_{\eta_c^\star,c}-P_{\nu,c}
\in\mathbb Q[x],
\qquad
\deg\Delta_{\nu,c}\le5.
\label{eq:supp-clifford-difference-polynomial}
\end{equation}
At every admissible integer \(n\), the identity-class contribution, which is
common to all candidate types, cancels, and therefore
\[
\Delta_{\nu,c}(n)
=
h_{c+36n}(\eta_c^\star)
-
h_{c+36n}(\nu).
\]

For each \(c\), the first index \(n\) such that \(c+36n\ge5\) is
\begin{equation}
n_c
:=
\min\{n\in\mathbb Z_{\ge0}:c+36n\ge5\}
=
\begin{cases}
1,&0\le c\le4,\\
0,&5\le c\le35 .
\end{cases}
\label{eq:supp-clifford-first-admissible-index}
\end{equation}
The progressions in
Eq.~\eqref{eq:supp-clifford-mod36-parametrization}, restricted to
\(n\ge n_c\), partition the integers \(m\ge5\).

After translating each difference polynomial to the beginning of its
range, write
\begin{equation}
\Delta_{\nu,c}(n_c+y)
=
\sum_{j=0}^{5}a_j^{(\nu,c)}y^j,
\qquad
a_j^{(\nu,c)}\in\mathbb Q.
\label{eq:supp-clifford-shifted-difference}
\end{equation}
For every comparison, the exact rational data recorded in
Subsec.~\ref{app:supp-clifford-type-certificates} give
\begin{equation}
a_0^{(\nu,c)}>0,
\qquad
a_j^{(\nu,c)}\ge0
\quad(1\le j\le5).
\label{eq:supp-clifford-certificate-signs}
\end{equation}
Hence, for every \(y\ge0\),
\[
\Delta_{\nu,c}(n_c+y)
\ge
a_0^{(\nu,c)}
>
0,
\]
and therefore \(\Delta_{\nu,c}(n)>0\) for every integer \(n\ge n_c\).

The certificate is exhaustive. Exactly \(12\) residue classes satisfy
\(3\mid c\); their family \(\mathcal F_0\) contains \(10\) candidates
and hence \(9\) competitors. The remaining \(24\) classes fall under the
seven-candidate family \(\mathcal F_1\) and hence have \(6\) competitors.
The certificate therefore covers all \(252\) winner--competitor
comparisons. The smallest shifted constant is \(5/2\), attained only at
\(c=10\), by the competitor \(\eta_3^{(0)}\).
The coefficient of \(y^5\) vanishes only when
\(c\equiv1\pmod6\) and \(\nu=\eta_3^{(1)}\), exactly the residue branch
of the leading asymptotic tie resolved above.

\emph{Completion and sharpness.}
Let \(m\ge5\), choose its unique residue
\(c\in\{0,\ldots,35\}\), and write \(m=c+36n\). Then \(n\ge n_c\) and
\(\eta_m^\star=\eta_c^\star\). Hence, for every candidate
\(\nu\ne\eta_m^\star\),
\[
h_m(\eta_m^\star)-h_m(\nu)
=
\Delta_{\nu,c}(n)
>
0.
\]
Thus \(\eta_m^\star\) strictly maximizes the score over the full candidate
family; by the candidate-family reduction,
\(\eta_m^\star\in\mathcal I_m\), and its strict dominance gives
\[
h_m(\eta_m^\star)=h_{\max}.
\]
Consequently, \(\eta_m^\star\) is the unique type in \(\mathcal I_m\)
attaining this value for every \(m\ge5\).

The threshold is sharp. The exact four-query scores satisfy
\[
h_4\!\left(\eta_3^{(0)}\right)
=
h_4\!\left(\eta_6^{(01)}\right)
=
5.
\]
Thus \(h_{\max}=5\) at \(m=4\), and both types occur and attain the maximum. The
maximizer is not unique, so the uniform uniqueness threshold cannot be
lowered below \(5\).
\end{proof}

The theorem determines only the maximizing effective Clifford type; it
does not assert uniqueness of an optimal probe state, system marginal,
POVM, or protocol. Nor does type selection alone imply a Wigner-positive
obstruction.

\section{Wigner-positive qutrit-Clifford testing}
\label{sec:supp-clifford-wigner}

\subsection{Low-query attainability}
\label{sec:supp-clifford-wigner-low-query}

Section~\ref{sec:supp-clifford-representation} identifies the relevant Clifford
types and carries out the score comparisons underlying the unrestricted optimum.
It remains to decide whether this optimum can be
attained with Wigner-positive probes and measurements. For
\(u=(p,q)\in\mathbb F_3^2\), we use the standard
finite-field qutrit phase-space
convention~\cite{GibbonsHoffmanWootters2004,gross2006hudson}:
\begin{equation}
\begin{aligned}
D_{(p,q)}
&:=
\omega^{2pq}X^pZ^q,
&
A_0
&:=
\sum_{j\in\mathbb F_3}
\lvert j\rangle\langle-j\rvert,
&
A_u
&:=
D_uA_0D_u^\dagger.
\end{aligned}
\label{eq:supp-qutrit-phase-point-convention}
\end{equation}
All labels and exponents are understood modulo \(3\). The tensor-product phase-point
operators and the state and effect Wigner functions are defined in
Eq.~\eqref{eq:qutrit-wigner-functions}.

Exact evaluation of
Eqs.~\eqref{eq:supp-clifford-branching}
and~\eqref{eq:subgroup-type-score} at \(m=3,4\) shows that the maximum
Clifford score is \(5\) in both cases; the complete scores are listed in
Subsec.~\ref{app:supp-clifford-type-certificates}.
Equation~\eqref{eq:optimal-type-II-error} therefore gives
\begin{equation}
\beta_{3,3}^{\star}(0)
=
\beta_{3,4}^{\star}(0)
=
\frac15.
\label{eq:supp-clifford-low-query-benchmark}
\end{equation}

\begin{proposition}[Low-query Wigner-positive attainability]
\label{prop:supp-clifford-low-query-wigner}
For each \(m\in\{3,4\}\), there exists an ancilla-free Wigner-positive
\(m\)-query parallel protocol
\((\sigma^{(m)},M^{(m)})\), with
\(M^{(m)}=\{M_0^{(m)},M_1^{(m)}\}\), such that
\[
\alpha(\sigma^{(m)},M^{(m)})=0,
\qquad
\beta(\sigma^{(m)},M^{(m)})
=
\frac15
=
\beta_{3,m}^{\star}(0).
\]
\end{proposition}

\begin{proof}
\emph{Three queries.}
In \(\operatorname{Sym}^3(\mathbb C^3)\), define the orthonormal
vectors
\begin{equation}
\begin{aligned}
\lvert G\rangle
&:=
\frac1{\sqrt3}
\bigl(
\lvert000\rangle
+
\lvert111\rangle
+
\lvert222\rangle
\bigr),
&
\lvert S\rangle
&:=
\frac1{\sqrt6}
\sum_{\pi\in\mathfrak S_3}
\lvert\pi(012)\rangle.
\end{aligned}
\label{eq:supp-clifford-low-query-symmetric-vectors}
\end{equation}
Set
\(
\mathcal K:=\operatorname{span}\{\lvert G\rangle,\lvert S\rangle\}
\)
and
\(
\mathcal Q:=\operatorname{Sym}^3(\mathbb C^3)\ominus\mathcal K
\).
Let \(\Pi_{\mathcal K}\) and \(\Pi_{\mathcal Q}\) denote the corresponding
orthogonal projectors. In the ordered basis
\((\lvert G\rangle,\lvert S\rangle)\), the tensor-power Clifford generators
restrict to
\[
\left.X_3\right|_{\mathcal K}
=
\left.Z_3\right|_{\mathcal K}
=
\id_2,
\qquad
\left.P_3\right|_{\mathcal K}
=
\begin{pmatrix}
1&0\\
0&\omega
\end{pmatrix},
\qquad
\left.F_3\right|_{\mathcal K}
=
\frac1{\sqrt3}
\begin{pmatrix}
1&\sqrt2\\
\sqrt2&-1
\end{pmatrix}.
\]
Thus \(\mathcal K\) is invariant under the tensor-power action of
\(\mathcal C_3\).

Choose the probe and binary POVM
\begin{equation}
\sigma^{(3)}
:=
\lvert G\rangle\langle G\rvert,
\qquad
M_1^{(3)}
:=
\Pi_{\mathcal Q},
\qquad
M_0^{(3)}
:=
\id_3^{\otimes3}-M_1^{(3)}.
\label{eq:supp-clifford-three-query-protocol}
\end{equation}
The two effects are complementary orthogonal projectors, so they form a
binary POVM in which outcome \(0\) accepts the Clifford hypothesis. For every
\(C\in\mathcal C_3\), invariance gives
\(C^{\otimes3}\lvert G\rangle\in\mathcal K\). Since
\(M_1^{(3)}=\Pi_{\mathcal Q}\), the rejection probability vanishes for every
Clifford unitary, and hence
\(\alpha(\sigma^{(3)},M^{(3)})=0\).

Let \(\Pi_{\mathrm{sym}}\) denote the projector onto
\(\operatorname{Sym}^3(\mathbb C^3)\). This space is an irreducible
\(\mathrm U(3)\)-module of dimension \(10\), so the Haar twirl is
\[
\mathcal T_{\mathrm U(3)}(\sigma^{(3)})
=
\frac{\Pi_{\mathrm{sym}}}{10}.
\]
Within the symmetric sector,
\(M_0^{(3)}\) restricts to \(\Pi_{\mathcal K}\), which has rank \(2\).
Therefore
\[
\beta(\sigma^{(3)},M^{(3)})
=
\Tr\!\left(
M_0^{(3)}\frac{\Pi_{\mathrm{sym}}}{10}
\right)
=
\frac{\dim\mathcal K}{10}
=
\frac15.
\]

Since \(\lvert G\rangle\) is a stabilizer state of three qutrits,
\(\sigma^{(3)}\) is Wigner-positive. Diagonal affine-symplectic transformations
and tensor-factor permutations partition \((\mathbb F_3^2)^3\) into four orbit
types: triples with all entries equal, triples with exactly two distinct
entries, triples of three distinct collinear points, and triples of three
distinct noncollinear points. Since \(M_1^{(3)}\) is invariant under the
diagonal Clifford action and tensor-factor permutations, Clifford covariance
makes its Wigner function constant on each orbit. The values of
\(W(M_1^{(3)}\mid\mathbf u)\) on these four orbit types are
\(0,1,0,0\), respectively; the exact evaluations for one representative
of each orbit type are recorded in
Subsec.~\ref{app:supp-wigner-witness-certificates}. Equivalently,
\begin{equation}
W\!\left(M_1^{(3)}\mid\mathbf u\right)
=
\mathbf 1_{\{\lvert\{u_1,u_2,u_3\}\rvert=2\}},
\qquad
\mathbf u=(u_1,u_2,u_3)\in(\mathbb F_3^2)^3.
\label{eq:supp-clifford-three-query-wigner-indicator}
\end{equation}
Thus \(M_1^{(3)}\) is Wigner-positive. Because
\(W(\id_3^{\otimes3}\mid\mathbf u)=1\), the complementary effect satisfies
\[
W\!\left(M_0^{(3)}\mid\mathbf u\right)
=
1-W\!\left(M_1^{(3)}\mid\mathbf u\right)
\in
\{0,1\}.
\]
Hence the probe and both POVM effects are Wigner-positive.

\emph{Four queries.}
Append a maximally mixed fourth query register and let the measurement act
trivially on it:
\begin{equation}
\sigma^{(4)}
:=
\sigma^{(3)}\otimes\frac{\id_3}{3},
\qquad
M_j^{(4)}
:=
M_j^{(3)}\otimes\id_3,
\qquad
j=0,1.
\label{eq:supp-clifford-spectator-lift}
\end{equation}
The fourth qutrit is one of the four query systems, not an ancilla. Since the
measurement acts trivially on the fourth qutrit, its contribution to either
outcome probability is
\[
\Tr\!\left[
\id_3 U(\id_3/3)U^\dagger
\right]
=
1
\qquad
(U\in\mathrm U(3)).
\]
Hence both error probabilities are unchanged:
\[
\alpha(\sigma^{(4)},M^{(4)})=0,
\qquad
\beta(\sigma^{(4)},M^{(4)})=\frac15.
\]
State Wigner functions factor under tensor products, with
\(W_{\id_3/3}(u)=1/9\). For the effects,
\(
W(E\otimes\id_3\mid\mathbf u,u_4)=W(E\mid\mathbf u)
\).
The four-query probe and both POVM effects are therefore Wigner-positive.
Because both protocols have \(\alpha=0\)
and \(\beta=1/5\), Eq.~\eqref{eq:supp-clifford-low-query-benchmark}
shows that each attains the unrestricted optimum.
\end{proof}

\subsection{Finite-window obstruction}
\label{sec:supp-clifford-wigner-obstruction}

For \(m=5,\ldots,9\), phase-space witnesses separate the set of
Wigner-positive system marginals from the locked marginal set by a
positive trace distance, yielding a strict excess type-II-error gap
independently of the final measurement.

\begin{theorem}[Finite-window Wigner-positive gap]
\label{thm:supp-clifford-finite-window-wigner}
For every \(m=5,\ldots,9\), there exists a constant
\(\delta_m>0\) such that every \(m\)-query parallel protocol
satisfying \(\alpha(\sigma,M)=0\) and having a Wigner-positive system
marginal obeys
\[
\beta(\sigma,M)
\ge
\beta_{3,m}^{\star}(0)+\delta_m.
\]
Consequently, no Wigner-positive
\(m\)-query parallel protocol attains \(\beta_{3,m}^{\star}(0)\).
\end{theorem}

\begin{proof}
\emph{Locked marginal sets and phase-space witnesses.}
Fix \(m\in\{5,\ldots,9\}\), and let
\(\eta_m^\star\) be the unique maximizing type from
Theorem~\ref{thm:supp-clifford-type-selection}. Call a Schur label
\(\lambda\) active if \(n_{\eta_m^\star,\lambda}>0\), and define its locked
weight by
\[
w_{m,\lambda}
:=
\frac{d_\lambda n_{\eta_m^\star,\lambda}}
     {d_{\eta_m^\star}h_{\max}}.
\]
These weights are the components of the typewise extremal profile
\(\boldsymbol{\pi}_{\eta_m^\star}^{\star}\) defined in
Sec.~\ref{sec:supp-locking}.
The exact-locking clause of
Theorem~\ref{thm:supp-support-weight-locking}, together with
Theorem~\ref{thm:supp-clifford-type-selection}, implies that every optimal system
marginal belongs to the locked marginal set \(\mathcal E_m\), consisting of the
\(m\)-qutrit density operators \(\rho\) satisfying
\[
\operatorname{supp}\rho
\subseteq
\mathcal H_{\eta_m^\star},
\qquad
\Tr\!\left(
\Pi_{\eta_m^\star,\lambda}\rho
\right)
=
w_{m,\lambda}
\quad
\forall\,\lambda:\,
n_{\eta_m^\star,\lambda}>0.
\]
Let \(\mathcal P_m\) be the free marginal set for Wigner positivity, namely the
set of \(m\)-qutrit density operators \(\rho\) satisfying
\[
W_\rho(\mathbf u)\ge0
\quad
\forall\,\mathbf u\in(\mathbb F_3^2)^m.
\]
Since the maximizing type is unique, \(\mathcal E_m\) is precisely
the locked marginal set of Corollary~\ref{cor:supp-trace-stability}
for the qutrit Clifford subgroup. We establish
\begin{equation}
\inf_{\substack{\rho\in\mathcal P_m\\ \omega\in\mathcal E_m}}
\frac12\|\rho-\omega\|_1
>
0,
\qquad
m=5,\ldots,9.
\label{eq:supp-clifford-locked-wigner-distance}
\end{equation}

The separating witnesses are built from normalized phase-space orbit
averages. In the orbit labels below, commas are omitted from multiplicity
patterns, powers indicate repeated multiplicities, and a thin space separates
adjacent powered groups, as in \(2^2\,1^5\). For a line pattern
\(\mu=(\mu_1,\mu_2,\mu_3)\), let \(\mathcal L_\mu\) be the set of
ordered tuples whose entries lie on a common affine line and whose
multiplicities at the three points of that line are
\(\mu_1,\mu_2,\mu_3\) in some order. For an affine-spanning pattern
\(\mu\), let \(\mathcal A_\mu\)
be the set of ordered tuples whose distinct entries affinely span
\(\mathbb F_3^2\) and whose nonzero multiplicities are given by \(\mu\).
Finally, let \(\mathcal D_{330+1}\) be the set of ordered seven-tuples in which
six entries lie on an affine line with multiplicities \(3,3,0\) at its three
points, while the remaining entry lies off that line. For any orbit set
\(\mathcal O\) defined above, let
\begin{equation}
\begin{aligned}
H_{\mathcal O}
&:=
\frac{1}{|\mathcal O|}
\sum_{\mathbf u\in\mathcal O}A_{\mathbf u},
\\
H_\mu^{\mathrm L}
&:=
H_{\mathcal L_\mu},
\qquad
H_\mu^{\mathrm A}
:=
H_{\mathcal A_\mu},
\qquad
H_{330+1}^{\mathrm{LD}}
:=
H_{\mathcal D_{330+1}}.
\end{aligned}
\label{eq:supp-clifford-phase-space-orbit-averages}
\end{equation}
For every such orbit set \(\mathcal O\),
Eq.~\eqref{eq:qutrit-wigner-functions} gives
\begin{equation}
\Tr(H_{\mathcal O}\rho)
=
\frac{3^m}{|\mathcal O|}
\sum_{\mathbf u\in\mathcal O}W_\rho(\mathbf u)
\ge0
\qquad
(\rho\in\mathcal P_m).
\label{eq:supp-clifford-orbit-wigner-nonnegativity}
\end{equation}
Thus every nonnegative linear combination of these orbit averages is
nonnegative on \(\mathcal P_m\), although the resulting witness need not
be positive semidefinite.

\emph{Joint-commutant reduction.}
Each orbit set is invariant under permutations of the \(m\) tensor
factors and under the diagonal qutrit Clifford action. Every Hermitian
operator \(B\) in the real span of the corresponding orbit averages
therefore belongs to the joint commutant of these actions. For an active
Schur sector \(\lambda\), the corresponding joint \(G\)-type/Schur sector is
\[
\operatorname{Ran}\Pi_{\eta_m^\star,\lambda}
\cong
V_{\eta_m^\star}
\otimes
\mathbb C^{\,n_{\eta_m^\star,\lambda}}
\otimes
S_\lambda,
\qquad
f^\lambda=\dim S_\lambda.
\]
Schur's lemma gives a unique Hermitian operator \(K_{B,\lambda}\) on
the multiplicity space such that
\begin{equation}
\Pi_{\eta_m^\star,\lambda}
B
\Pi_{\eta_m^\star,\lambda}
=
\id_{V_{\eta_m^\star}}
\otimes
K_{B,\lambda}
\otimes
\id_{S_\lambda}.
\label{eq:supp-clifford-witness-reduced-block}
\end{equation}

For \(\rho\in\mathcal E_m\), define the normalized multiplicity-space
state
\[
\rho_\lambda
:=
\frac1{w_{m,\lambda}}
\Tr_{V_{\eta_m^\star}\otimes S_\lambda}
\!\left(
\Pi_{\eta_m^\star,\lambda}
\rho
\Pi_{\eta_m^\star,\lambda}
\right).
\]
Then \(\rho_\lambda\) is a density operator on the multiplicity space, and
the \(\lambda\)-sector contribution to \(\Tr(B\rho)\) depends on \(\rho\) only
through \(\rho_\lambda\). Cross-sector coherences do not contribute
because \(B\) preserves every joint \(G\)-type/Schur sector and acts within
each such sector as the identity on the \(V_{\eta_m^\star}\) and \(S_\lambda\) factors.
Hence
\begin{equation}
\Tr(B\rho)
=
\sum_{\lambda:\,n_{\eta_m^\star,\lambda}>0}
w_{m,\lambda}
\Tr\!\left(K_{B,\lambda}\rho_\lambda\right).
\label{eq:supp-clifford-witness-reduced-expectation}
\end{equation}
Thus each exact certificate reduces to evaluating the corresponding
multiplicity-space blocks \(K_{B,\lambda}\).
Subsec.~\ref{app:supp-wigner-witness-certificates} records the
active-sector, locked-weight, and reduced-block data.

\emph{Exact certificates.}
For each witness \(B_m\) below, write
\(K_{m,\lambda}:=K_{B_m,\lambda}\).
The query-dependent witnesses and their constant values on
\(\mathcal E_m\) are summarized below:
\[
{\renewcommand{\arraystretch}{1.35}
\begin{array}{@{}c@{\qquad}c@{\qquad}c@{}}
\toprule
m
&
B_m
&
\Tr(B_m\rho)
\\
\midrule
5
&
H_{221}^{\mathrm L}
&
-\frac{1}{15}
\\
6
&
H_{222}^{\mathrm L}
&
-\frac{1}{5}
\\
7
&
\frac{1}{12}H_{430}^{\mathrm L}
+
\frac{1}{4}H_{331}^{\mathrm L}
+
\frac{2}{3}H_{330+1}^{\mathrm{LD}}
&
-\frac{55}{8064}
\\
8
&
H_{800}^{\mathrm L}
&
-\frac{1}{3}
\\
9
&
39H_{1^9}^{\mathrm A}
+
1134H_{2^2\,1^5}^{\mathrm A}
+
2H_{333}^{\mathrm L}
&
-\frac{16}{805}
\\
\bottomrule
\end{array}}
\]

For \(m=5\), the maximizing type is
\(\eta_5^\star=\eta_3^{(1)}\). The active Schur sectors are
\((5)\) and \((4,1)\), with locked weights
\(w_{5,(5)}=\tfrac{7}{15}\) and
\(w_{5,(4,1)}=\tfrac{8}{15}\).
Take \(B_5:=H_{221}^{\mathrm L}\).
The two active multiplicity spaces
are one-dimensional, so the reduced blocks are scalars. Exact evaluation
of the weighted combination gives
\begin{equation}
\frac7{15}K_{5,(5)}
+
\frac8{15}K_{5,(4,1)}
=
-\frac1{15}.
\label{eq:supp-clifford-five-query-certificate}
\end{equation}
Equation~\eqref{eq:supp-clifford-witness-reduced-expectation} therefore yields
\[
\Tr(B_5\rho)
=
-\frac1{15}
<
0
\qquad
(\rho\in\mathcal E_5).
\]

For \(m=6\), the maximizing type is
\(\eta_6^\star=\chi_1\), and the only active Schur sector is
\((5,1)\). Take
\(B_6:=H_{222}^{\mathrm L}\). Exact evaluation gives
\[
K_{6,(5,1)}
=
-\frac15.
\]
Equation~\eqref{eq:supp-clifford-witness-reduced-expectation} therefore
yields
\[
\Tr(B_6\rho)
=
-\frac15
<
0
\qquad
(\rho\in\mathcal E_6).
\]

For \(m=7\), the maximizing type is
\(\eta_7^\star=\eta_3^{(0)}\). The active sectors are
\((7)\), \((6,1)\), and \((4,3)\), with locked weights
\[
\begin{aligned}
w_{7,(7)}&=\frac12,
&
w_{7,(6,1)}&=\frac13,
&
w_{7,(4,3)}&=\frac16.
\end{aligned}
\]
Take the witness
\[
B_7
:=
\frac1{12}H_{430}^{\mathrm L}
+
\frac14H_{331}^{\mathrm L}
+
\frac23H_{330+1}^{\mathrm{LD}}.
\]
The exact reduced blocks are
\[
K_{7,(7)}
=
-\frac1{96}\id_2,
\qquad
K_{7,(6,1)}
=
-\frac1{192},
\qquad
K_{7,(4,3)}
=
\frac1{1344}.
\]
Here \(\id_2\) acts on the two-dimensional multiplicity
space in the \((7)\)-sector, while the one-dimensional blocks
are identified with their scalar values. Although the
\((4,3)\)-block is positive, Eq.~\eqref{eq:supp-clifford-witness-reduced-expectation} and the locked weights
above give
\[
\Tr(B_7\rho)
=
-\frac{55}{8064}
<
0
\qquad
(\rho\in\mathcal E_7).
\]

For \(m=8\), the maximizing type is
\(\eta_8^\star=\eta_3^{(1)}\), and the active sectors are
\((8)\), \((7,1)\), \((6,1,1)\), \((5,3)\), and \((4,4)\).
Take \(B_8:=H_{800}^{\mathrm L}\). Exact evaluation gives
\[
K_{8,\lambda}
=
-\frac13\id
\]
on the multiplicity space of every active sector \(\lambda\). Since the
locked weights sum to one,
\[
\Tr(B_8\rho)
=
-\frac13
<
0
\qquad
(\rho\in\mathcal E_8).
\]

For \(m=9\), the maximizing type is
\(\eta_9^\star=\tau_{01}\), and the active sectors are
\((9)\), \((8,1)\), \((7,2)\), \((6,3)\), and \((4,4,1)\). Take
\[
B_9
:=
39H_{1^9}^{\mathrm A}
+
1134H_{2^2\,1^5}^{\mathrm A}
+
2H_{333}^{\mathrm L}.
\]
Only the \((9)\)-sector has a two-dimensional multiplicity space;
all other active sectors are multiplicity-free. Fix an orthonormal basis of
this space. For each orbit-average summand \(H\) in \(B_9\),
Eq.~\eqref{eq:supp-clifford-witness-reduced-expectation} gives
\[
\Tr(H\rho)
=
\Tr\!\Bigl[
\Bigl(
w_{9,(9)}K_{H,(9)}
+
\sum_{\substack{\lambda\ne(9)\\
n_{\eta_9^\star,\lambda}>0}}
w_{9,\lambda}K_{H,\lambda}\id_2
\Bigr)
\rho_{(9)}
\Bigr],
\]
where the reduced blocks in the sum are identified with their scalar
values. For the three orbit-average summands in \(B_9\), denote the
resulting effective Hermitian matrices on this two-dimensional space by
\[
K_{1^9}^{\mathrm A},
\qquad
K_{2^2\,1^5}^{\mathrm A},
\qquad
K_{333}^{\mathrm L}.
\]
Exact evaluation gives the matrix identity
\begin{equation}
39K_{1^9}^{\mathrm A}
+
1134K_{2^2\,1^5}^{\mathrm A}
+
2K_{333}^{\mathrm L}
=
-\frac{16}{805}\id_2.
\label{eq:supp-clifford-nine-query-certificate}
\end{equation}
Consequently,
\[
\Tr(B_9\rho)
=
-\frac{16}{805}
<
0
\qquad
(\rho\in\mathcal E_9).
\]

The five calculations above give
\begin{equation}
\Tr(B_m\omega)=-a_m
\quad
(\omega\in\mathcal E_m),
\qquad
(a_5,a_6,a_7,a_8,a_9)
=
\left(
\frac1{15},
\frac15,
\frac{55}{8064},
\frac13,
\frac{16}{805}
\right).
\label{eq:supp-clifford-witness-values}
\end{equation}
Set \(L_m:=\|B_m\|_\infty>0\).
Each \(B_m\) is a nonnegative linear combination of the normalized
orbit averages in Eq.~\eqref{eq:supp-clifford-phase-space-orbit-averages}, so
Eq.~\eqref{eq:supp-clifford-orbit-wigner-nonnegativity} gives
\(\Tr(B_m\rho)\ge0\) for every \(\rho\in\mathcal P_m\).
Hence, for \(\rho\in\mathcal P_m\) and \(\omega\in\mathcal E_m\),
\[
a_m
\le
\left|\Tr[B_m(\rho-\omega)]\right|
\le
L_m\|\rho-\omega\|_1.
\]
Therefore,
\begin{equation}
\inf_{\omega\in\mathcal E_m}
\frac12\|\rho-\omega\|_1
\ge
\frac{a_m}{2L_m}
\qquad
(\rho\in\mathcal P_m),
\label{eq:supp-clifford-witness-trace-distance}
\end{equation}
which proves Eq.~\eqref{eq:supp-clifford-locked-wigner-distance}.

Let \(C_{3,m}\) denote the constant in
Corollary~\ref{cor:supp-trace-stability} for the qutrit Clifford
subgroup. Combining Eq.~\eqref{eq:supp-clifford-witness-trace-distance}
with that corollary shows that every zero-type-I-error protocol with
system marginal in \(\mathcal P_m\) satisfies
\begin{equation}
\beta(\sigma,M)-\beta_{3,m}^{\star}(0)
\ge
\frac{a_m^2}{4L_m^2C_{3,m}}
=:
\delta_m
>
0,
\qquad
m=5,\ldots,9.
\label{eq:supp-clifford-witness-error-gap}
\end{equation}

It remains to pass from system marginals to protocols. Let
\(\sigma_{AR}\) be a Wigner-positive probe, where \(R\) is a possibly
trivial register of \(r\) qutrits, and set
\(\rho_A:=\Tr_R\sigma_{AR}\). The phase-point identity
\(\sum_{\mathbf v}A_{\mathbf v}=3^r\id_R\), together with
Eq.~\eqref{eq:qutrit-wigner-functions}, gives
\[
W_{\rho_A}(\mathbf u)
=
\sum_{\mathbf v\in(\mathbb F_3^2)^r}
W_{\sigma_{AR}}(\mathbf u,\mathbf v)
\ge
0.
\]
Thus a Wigner-positive probe has a Wigner-positive system marginal,
so Eq.~\eqref{eq:supp-clifford-witness-error-gap} applies to every such
zero-type-I-error protocol, even when the final measurement is
unrestricted. In particular, no such protocol attains
\(\beta_{3,m}^{\star}(0)\) for \(m=5,\ldots,9\).

Although Theorem~\ref{thm:supp-clifford-type-selection} determines the
maximizing Clifford type for every \(m\ge5\), the phase-space
certificates above are available only for \(m=5,\ldots,9\).
Wigner-positive attainability for \(m\ge10\) remains open.
\end{proof}

\section{Exact certificates}
\label{sec:supp-certificates}
\renewcommand{\thetable}{S\arabic{table}}

\subsection{Clifford type selection}
\label{app:supp-clifford-type-certificates}

This subsection records the exact finite data used in the proof of
Theorem~\ref{thm:supp-clifford-type-selection}.  All entries are exact
rational or cyclotomic quantities; no numerical rounding is used.  The
representation labels, candidate families, and lift conventions are those
fixed in Sec.~\ref{sec:supp-clifford-representation}.

\emph{Lift orders and phase periodicity.}
Exact enumeration of \(\mathcal C_3\) over
\(\mathbb Q(e^{i\pi/6})\) gives the counts in
Table~\ref{tab:supp-clifford-lift-orders}.
Spectral multiplicities are determined from characteristic
polynomials, and the listed power identities are verified
for all corresponding lifts in \(\mathcal C_3\).

\begin{center}
\begin{minipage}{\textwidth}
\centering
\setlength{\tabcolsep}{12pt}
\renewcommand{\arraystretch}{1.15}
\begingroup\small
\begin{tabular}{@{}lcc@{}}
\toprule
Property
& \(\mathscr P\)
& \(\mathscr P'\)
\\
\midrule
Projective elements
& \(33\)
& \(182\)
\\
Lifts in \(\mathcal C_3\)
& \(396\)
& \(2184\)
\\
Power identity
& \(C^{12}=\id_3\)
& \(C^{36}=\id_3\)
\\
\bottomrule
\end{tabular}
\par\endgroup
\captionsetup{hypcap=false,justification=centering}
\captionof{table}{Spectral phase-order bounds for the finite Clifford
group \(\mathcal C_3\). Each identity holds for every
\(C\in\mathcal C_3\) whose projective class belongs to the
indicated subset.}
\label{tab:supp-clifford-lift-orders}
\end{minipage}
\end{center}

These identities give eigenvalue orders dividing \(12\) and \(36\),
respectively, independently of the choice of lift within
\(\mathcal C_3\).

\emph{Leading and subleading coefficients.}
The asymptotic comparison in
Eq.~\eqref{eq:supp-clifford-asymptotic-expansion} is controlled by
\(B_\eta(r)\), except on the tied branch \(r=1\), where the
subleading coefficient \(C_\eta(s)\) is required.
Tables~\ref{tab:supp-clifford-b-f0} and~\ref{tab:supp-clifford-b-f1}
give the complete leading data.

\begin{center}
\begin{minipage}{\textwidth}
\centering
\renewcommand{\arraystretch}{1.08}
\setlength{\tabcolsep}{2.75pt}
\begin{tabular}{@{}c*{10}{r}@{}}
\toprule
\(r\)
&\(\chi_0\)&\(\chi_1\)&\(\chi_2\)
&\(\tau_{01}\)&\(\tau_{02}\)&\(\tau_{12}\)
&\(\nu_3\)&\(\Gamma_0\)&\(\Gamma_1\)&\(\Gamma_2\)
\\
\midrule
\(0\)&\(9\)&\(41\)&\(-23\)&\(7\)&\(-25\)&\(-9\)&\(9\)&\(0\)&\(8\)&\(-8\)
\\
\(3\)&\(-9\)&\(23\)&\(-41\)&\(25\)&\(-7\)&\(9\)&\(-9\)&\(0\)&\(8\)&\(-8\)
\\
\bottomrule
\end{tabular}
\captionsetup{hypcap=false}
\captionof{table}{Exact scaled coefficients \(207360\,B_\eta(r)\) for
\(\eta\in\mathcal F_0\) and \(r\in\{0,3\}\).}
\label{tab:supp-clifford-b-f0}
\end{minipage}
\end{center}

\begin{center}
\begin{minipage}{\textwidth}
\centering
\renewcommand{\arraystretch}{1.08}
\setlength{\tabcolsep}{3.25pt}
\begin{tabular}{@{}c*{7}{r}@{}}
\toprule
\(r\)
&\(\eta_3^{(0)}\)&\(\eta_3^{(1)}\)&\(\eta_3^{(2)}\)
&\(\eta_6^{(01)}\)&\(\eta_6^{(02)}\)&\(\eta_6^{(12)}\)&\(\eta_9\)
\\
\midrule
\(1\)&\(41\)&\(41\)&\(-55\)&\(23\)&\(-25\)&\(-25\)&\(9\)
\\
\(2\)&\(-41\)&\(55\)&\(-41\)&\(25\)&\(-23\)&\(25\)&\(-9\)
\\
\(4\)&\(23\)&\(23\)&\(-73\)&\(41\)&\(-7\)&\(-7\)&\(-9\)
\\
\(5\)&\(-23\)&\(73\)&\(-23\)&\(7\)&\(-41\)&\(7\)&\(9\)
\\
\bottomrule
\end{tabular}
\captionsetup{hypcap=false}
\captionof{table}{Exact scaled coefficients \(622080\,B_\eta(r)\) for
\(\eta\in\mathcal F_1\) and \(r\in\{1,2,4,5\}\).}
\label{tab:supp-clifford-b-f1}
\end{minipage}
\end{center}

\begin{samepage}
For \(r=1\), the subleading comparison uses the two residue classes
\(s=1,7\pmod{12}\).  Table~\ref{tab:supp-clifford-c-coefficients} records the coefficient
values, which are identical in the two residue classes.

\begin{center}
\begin{minipage}{\textwidth}
\centering
\renewcommand{\arraystretch}{1.08}
\setlength{\tabcolsep}{3.25pt}
\begin{tabular}{@{}c*{7}{r}@{}}
\toprule
\(s\)
&\(\eta_3^{(0)}\)&\(\eta_3^{(1)}\)&\(\eta_3^{(2)}\)
&\(\eta_6^{(01)}\)&\(\eta_6^{(02)}\)&\(\eta_6^{(12)}\)&\(\eta_9\)
\\
\midrule
\(1\)&\(155\)&\(91\)&\(-165\)&\(69\)&\(-59\)&\(-91\)&\(27\)
\\
\(7\)&\(155\)&\(91\)&\(-165\)&\(69\)&\(-59\)&\(-91\)&\(27\)
\\
\bottomrule
\end{tabular}
\captionsetup{hypcap=false}
\captionof{table}{Exact scaled coefficients \(82944\,C_\eta(s)\) for
\(\eta\in\mathcal F_1\) and \(s\in\{1,7\}\).}
\label{tab:supp-clifford-c-coefficients}
\end{minipage}
\end{center}

Thus for \(s\in\{1,7\}\),
\[
C_{\eta_3^{(0)}}(s)-C_{\eta_3^{(1)}}(s)
=
\frac1{1296}.
\]
\end{samepage}

\emph{The exact mod-\(36\) certificate.}
For the shifted difference-polynomial coefficients defined in
Eqs.~\eqref{eq:supp-clifford-difference-polynomial}--\eqref{eq:supp-clifford-certificate-signs},
define the residue-wise minima
\[
\begin{aligned}
a_{c,0}
&:=
\min_{\nu\ne\eta_c^\star}a_0^{(\nu,c)},
&
a_{c,+}
&:=
\min_{\nu\ne\eta_c^\star}
\min_{1\le j\le5}a_j^{(\nu,c)}.
\end{aligned}
\]
Let \(\nu_{c,0}\) be one competitor attaining \(a_{c,0}\), and let
\((\nu_{c,+},j_{c,+})\) be one competitor--degree pair attaining
\(a_{c,+}\).  Table~\ref{tab:supp-clifford-threshold-minima} records these
data for every residue class.

\begin{center}
\begin{minipage}{\textwidth}
\centering
\small
\renewcommand{\arraystretch}{1.28}
\setlength{\tabcolsep}{2.5pt}
\begin{tabular}{@{}c r c r l @ {\hspace{1.2em}}
                        c r c r l @ {\hspace{1.2em}}
                        c r c r l@{}}
\toprule
\multicolumn{5}{c}{\(0\le c\le11\)}
&\multicolumn{5}{c}{\(12\le c\le23\)}
&\multicolumn{5}{c}{\(24\le c\le35\)}
\\
\cmidrule(lr){1-5}\cmidrule(lr){6-10}\cmidrule(lr){11-15}
\(c\)&\(a_{c,0}\)&\(\nu_{c,0}\)&\(a_{c,+}\)&\((\nu_{c,+},j_{c,+})\)
&\(c\)&\(a_{c,0}\)&\(\nu_{c,0}\)&\(a_{c,+}\)&\((\nu_{c,+},j_{c,+})\)
&\(c\)&\(a_{c,0}\)&\(\nu_{c,0}\)&\(a_{c,+}\)&\((\nu_{c,+},j_{c,+})\)
\\
\midrule
\(0\) & \(13559\) & \(\chi_0\) & \(\frac{46656}{5}\) & \((\chi_0,5)\) & \(12\) & \(109\) & \(\chi_0\) & \(\frac{6574}{5}\) & \((\chi_0,1)\) & \(24\) & \(2114\) & \(\chi_0\) & \(\frac{46656}{5}\) & \((\chi_0,5)\) \\
\(1\) & \(2268\) & \(\eta_3^{(1)}\) & \(0\) & \((\eta_3^{(1)},5)\) & \(13\) & \(70\) & \(\eta_3^{(1)}\) & \(0\) & \((\eta_3^{(1)},5)\) & \(25\) & \(575\) & \(\eta_3^{(1)}\) & \(0\) & \((\eta_3^{(1)},5)\) \\
\(2\) & \(5962\) & \(\eta_6^{(01)}\) & \(2916\) & \((\eta_6^{(12)},5)\) & \(14\) & \(75\) & \(\eta_6^{(01)}\) & \(\frac{1595}{2}\) & \((\eta_6^{(01)},1)\) & \(26\) & \(1072\) & \(\eta_6^{(01)}\) & \(2916\) & \((\eta_6^{(12)},5)\) \\
\(3\) & \(3663\) & \(\chi_1\) & \(\frac{2916}{5}\) & \((\chi_1,5)\) & \(15\) & \(\frac{255}{2}\) & \(\chi_1\) & \(\frac{2916}{5}\) & \((\chi_1,5)\) & \(27\) & \(927\) & \(\chi_1\) & \(\frac{2916}{5}\) & \((\chi_1,5)\) \\
\(4\) & \(3533\) & \(\eta_3^{(0)}\) & \(\frac{8748}{5}\) & \((\eta_3^{(0)},5)\) & \(16\) & \(\frac{79}{2}\) & \(\eta_3^{(0)}\) & \(\frac{8923}{20}\) & \((\eta_3^{(0)},1)\) & \(28\) & \(\frac{1243}{2}\) & \(\eta_3^{(0)}\) & \(\frac{8748}{5}\) & \((\eta_3^{(0)},5)\) \\
\(5\) & \(\frac{15}{2}\) & \(\eta_6^{(01)}\) & \(\frac{1333}{10}\) & \((\eta_6^{(01)},1)\) & \(17\) & \(440\) & \(\eta_6^{(01)}\) & \(\frac{18749}{5}\) & \((\eta_6^{(01)},1)\) & \(29\) & \(\frac{8381}{2}\) & \(\eta_6^{(01)}\) & \(\frac{31104}{5}\) & \((\eta_9,5)\) \\
\(6\) & \(7\) & \(\chi_0\) & \(\frac{769}{5}\) & \((\chi_0,1)\) & \(18\) & \(578\) & \(\chi_0\) & \(\frac{25599}{5}\) & \((\chi_0,1)\) & \(30\) & \(5867\) & \(\chi_0\) & \(\frac{46656}{5}\) & \((\chi_0,5)\) \\
\(7\) & \(\frac{25}{2}\) & \(\eta_6^{(01)}\) & \(0\) & \((\eta_3^{(1)},5)\) & \(19\) & \(230\) & \(\eta_3^{(1)}\) & \(0\) & \((\eta_3^{(1)},5)\) & \(31\) & \(1211\) & \(\eta_3^{(1)}\) & \(0\) & \((\eta_3^{(1)},5)\) \\
\(8\) & \(\frac{15}{2}\) & \(\eta_6^{(01)}\) & \(\frac{541}{4}\) & \((\eta_6^{(01)},1)\) & \(20\) & \(\frac{673}{2}\) & \(\eta_6^{(01)}\) & \(\frac{10525}{4}\) & \((\eta_6^{(01)},1)\) & \(32\) & \(2716\) & \(\eta_6^{(01)}\) & \(2916\) & \((\eta_6^{(12)},5)\) \\
\(9\) & \(\frac{59}{2}\) & \(\chi_1\) & \(\frac{5931}{20}\) & \((\chi_1,1)\) & \(21\) & \(\frac{723}{2}\) & \(\chi_1\) & \(\frac{2916}{5}\) & \((\chi_1,5)\) & \(33\) & \(1910\) & \(\chi_1\) & \(\frac{2916}{5}\) & \((\chi_1,5)\) \\
\(10\) & \(\frac{5}{2}\) & \(\eta_3^{(0)}\) & \(\frac{629}{10}\) & \((\eta_3^{(0)},1)\) & \(22\) & \(188\) & \(\eta_3^{(0)}\) & \(\frac{7627}{5}\) & \((\eta_3^{(0)},1)\) & \(34\) & \(\frac{3187}{2}\) & \(\eta_3^{(0)}\) & \(\frac{8748}{5}\) & \((\eta_3^{(0)},5)\) \\
\(11\) & \(82\) & \(\eta_6^{(01)}\) & \(\frac{19481}{20}\) & \((\eta_6^{(01)},1)\) & \(23\) & \(\frac{3069}{2}\) & \(\eta_6^{(01)}\) & \(\frac{31104}{5}\) & \((\eta_9,5)\) & \(35\) & \(\frac{19289}{2}\) & \(\eta_6^{(01)}\) & \(\frac{31104}{5}\) & \((\eta_9,5)\) \\
\bottomrule
\end{tabular}
\captionsetup{hypcap=false,justification=centering}
\captionof{table}{Residue-wise minima of the shifted difference-polynomial
coefficients.  The three column blocks align \(c\), \(c+12\),
and \(c+24\), which share the same residue modulo \(12\).  For each
residue, the table displays one competitor attaining \(a_{c,0}\) and
one competitor--degree pair attaining \(a_{c,+}\).}
\label{tab:supp-clifford-threshold-minima}
\end{minipage}
\end{center}

The minima satisfy \(a_{c,0}>0\) and \(a_{c,+}\ge0\) for every \(c\).
The equality \(a_{c,+}=0\) occurs only when \(c\equiv1\pmod6\).  The global
minimum of the shifted constant coefficients is \(5/2\), attained only at
\(c=10\), by the competitor \(\eta_3^{(0)}\).  Since each minimum is taken over all competitors, these \(36\)
residue-class entries certify all \(252\) winner--competitor comparisons.

\noindent\begin{minipage}{\textwidth}
\emph{Low-query score checks.}
Table~\ref{tab:supp-clifford-low-query-scores} lists the complete scores at
three and four queries, with the candidates ordered as in
Eq.~\eqref{eq:supp-clifford-type-families}.

\begin{center}
\renewcommand{\arraystretch}{1.12}
\begin{tabular}{@{}c@{\hspace{1.5em}}c@{}}
\toprule
\(m=3,\ \mathcal F_0\)
&
\(m=4,\ \mathcal F_1\)
\\
\cmidrule(lr){1-1}\cmidrule(lr){2-2}
\begin{tabular}{@{}c*{10}{r}@{}}
\(\eta\)
&\(\chi_0\)&\(\chi_1\)&\(\chi_2\)
&\(\tau_{01}\)&\(\tau_{02}\)&\(\tau_{12}\)
&\(\nu_3\)&\(\Gamma_0\)&\(\Gamma_1\)&\(\Gamma_2\)
\\
\(h_3(\eta)\)&\(0\)&\(1\)&\(0\)&\(5\)&\(0\)&\(0\)&\(0\)&\(\frac54\)&\(1\)&\(0\)
\end{tabular}
&
\begin{tabular}{@{}c*{7}{r}@{}}
\(\eta\)
&\(\eta_3^{(0)}\)&\(\eta_3^{(1)}\)&\(\eta_3^{(2)}\)
&\(\eta_6^{(01)}\)&\(\eta_6^{(02)}\)&\(\eta_6^{(12)}\)&\(\eta_9\)
\\
\(h_4(\eta)\)&\(5\)&\(1\)&\(0\)&\(5\)&\(\frac52\)&\(1\)&\(\frac53\)
\end{tabular}
\\
\bottomrule
\end{tabular}
\captionsetup{hypcap=false}
\captionof{table}{Exact Clifford-type scores at three and four queries.}
\label{tab:supp-clifford-low-query-scores}
\end{center}

The maximum score is \(5\) at both query numbers.  At \(m=4\), both
\(\eta_3^{(0)}\) and \(\eta_6^{(01)}\) attain the maximum.  These scores provide the benchmark in
Proposition~\ref{prop:supp-clifford-low-query-wigner} and establish the
sharpness of the uniqueness threshold in
Theorem~\ref{thm:supp-clifford-type-selection}.
\end{minipage}

\subsection{Wigner functions and phase-space witnesses}
\label{app:supp-wigner-witness-certificates}

This subsection records the exact phase-space and reduced-block data used in
Proposition~\ref{prop:supp-clifford-low-query-wigner} and
Theorem~\ref{thm:supp-clifford-finite-window-wigner}.  All displayed values
are exact, and the notation and normalization conventions are those of
Subsecs.~\ref{sec:supp-clifford-wigner-low-query}
and~\ref{sec:supp-clifford-wigner-obstruction}.

\emph{Three-query phase-space orbits.}
For the three-query measurement in
Eq.~\eqref{eq:supp-clifford-three-query-protocol},
\[
M_1^{(3)}
=
\Pi_{\mathcal Q}
=
\Pi_{\mathrm{sym}}
-
\lvert G\rangle\langle G\rvert
-
\lvert S\rangle\langle S\rvert.
\]
Consequently,
\[
W\!\left(M_1^{(3)}\mid\mathbf u\right)
=
\Tr\!\left(\Pi_{\mathrm{sym}}A_{\mathbf u}\right)
-
\langle G\rvert A_{\mathbf u}\lvert G\rangle
-
\langle S\rvert A_{\mathbf u}\lvert S\rangle.
\]
Diagonal affine-symplectic transformations and tensor-factor permutations
partition the \(729\) ordered triples into the four orbit types listed in
Table~\ref{tab:supp-wigner-three-query-orbits}.

\begin{table}[!htbp]
\centering
{%
\normalsize
\renewcommand{\arraystretch}{1.18}
\begin{tabular}{@{}lcccc@{}}
\toprule
Orbit type
& Representative
& \(\lvert\mathcal O\rvert\)
& \(W(M_1^{(3)}\mid\mathbf u)\)
& \(W(M_0^{(3)}\mid\mathbf u)\)
\\
\midrule
Identical
& \((0,0,0)\)
& \(9\)
& \(0\)
& \(1\)
\\
Two distinct
& \((0,0,e_1)\)
& \(216\)
& \(1\)
& \(0\)
\\
Three collinear
& \((0,e_1,2e_1)\)
& \(72\)
& \(0\)
& \(1\)
\\
Three noncollinear
& \((0,e_1,e_2)\)
& \(432\)
& \(0\)
& \(1\)
\\
\bottomrule
\end{tabular}
}
\caption{Orbits of ordered triples in \((\mathbb F_3^2)^3\) under diagonal
affine-symplectic transformations and tensor-factor permutations.
Representatives are chosen with $e_1=(1,0)$ and $e_2=(0,1)$
in the fixed phase-space coordinates.}
\label{tab:supp-wigner-three-query-orbits}
\end{table}

The orbit sizes sum to \(729=9^3\).  Define the ordered quadruple
\[
\mathcal V(\mathbf u)
:=
\left(
\Tr\!\left(\Pi_{\mathrm{sym}}A_{\mathbf u}\right),
\langle G\rvert A_{\mathbf u}\lvert G\rangle,
\langle S\rvert A_{\mathbf u}\lvert S\rangle,
W\!\left(M_1^{(3)}\mid\mathbf u\right)
\right).
\]
Exact evaluation at the four representatives gives
\[
\begin{aligned}
\mathcal V(0,0,0)
&=(2,1,1,0),
&
\mathcal V(0,0,e_1)
&=(1,0,0,1),
\\
\mathcal V(0,e_1,2e_1)
&=\left(\frac12,0,\frac12,0\right),
&
\mathcal V(0,e_1,e_2)
&=(0,0,0,0).
\end{aligned}
\]
Since $M_1^{(3)}$ is invariant under the diagonal Clifford
action and tensor-factor permutations, Clifford covariance
makes its Wigner function constant on each orbit.

\begin{samepage}
\emph{Reduced-block evaluation.}
Taking the partial trace of the factorization in
Eq.~\eqref{eq:supp-clifford-witness-reduced-block} over the
\(V_{\eta_m^\star}\) and \(S_\lambda\) factors gives
\begin{equation}
K_{B,\lambda}
=
\frac{1}{d_{\eta_m^\star}f^\lambda}
\Tr_{V_{\eta_m^\star}\otimes S_\lambda}
\!\left(
\Pi_{\eta_m^\star,\lambda}
B
\Pi_{\eta_m^\star,\lambda}
\right).
\label{eq:app-supp-wigner-reduced-block-evaluation}
\end{equation}
\end{samepage}%
The reduced block \(K_{B,\lambda}\) acts on
\(\mathbb C^{n_{\eta_m^\star,\lambda}}\).  Scalar blocks are identified
with their scalar values, while matrix blocks are represented in a fixed
orthonormal multiplicity basis.  For the witnesses \(B_m\) fixed in
Subsec.~\ref{sec:supp-clifford-wigner-obstruction}, write
\(K_{m,\lambda}:=K_{B_m,\lambda}\).
Orbit labels follow the convention of
Subsec.~\ref{sec:supp-clifford-wigner-obstruction}.

\emph{Locked sectors and reduced certificates.}
Table~\ref{tab:supp-wigner-locked-data} records every active Schur sector, its
locked weight, and the multiplicity of the maximizing Clifford type in that
sector, for \(m=5,\ldots,9\).

\begin{center}
\begin{minipage}{\textwidth}
\centering
\renewcommand{\arraystretch}{1.12}
\setlength{\tabcolsep}{8pt}
\begin{tabular}{@{}ccccc@{}}
\toprule
\(m\)
& \(\eta_m^\star\)
& \(\lambda\)
& \(w_{m,\lambda}\)
& \(n_{\eta_m^\star,\lambda}\)
\\
\midrule
\multirow{2}{*}{\(5\)}
& \multirow{2}{*}{\(\eta_3^{(1)}\)}
& \((5)\) & \(7/15\) & \(1\)
\\
& & \((4,1)\) & \(8/15\) & \(1\)
\\
\addlinespace[0.5em]
\(6\)
& \(\chi_1\)
& \((5,1)\) & \(1\) & \(1\)
\\
\addlinespace[0.5em]
\multirow{3}{*}{\(7\)}
& \multirow{3}{*}{\(\eta_3^{(0)}\)}
& \((7)\) & \(1/2\) & \(2\)
\\
& & \((6,1)\) & \(1/3\) & \(1\)
\\
& & \((4,3)\) & \(1/6\) & \(1\)
\\
\addlinespace[0.5em]
\multirow{5}{*}{\(8\)}
& \multirow{5}{*}{\(\eta_3^{(1)}\)}
& \((8)\) & \(15/83\) & \(1\)
\\
& & \((7,1)\) & \(42/83\) & \(2\)
\\
& & \((6,1,1)\) & \(7/83\) & \(1\)
\\
& & \((5,3)\) & \(14/83\) & \(1\)
\\
& & \((4,4)\) & \(5/83\) & \(1\)
\\
\addlinespace[0.5em]
\multirow{5}{*}{\(9\)}
& \multirow{5}{*}{\(\tau_{01}\)}
& \((9)\) & \(22/69\) & \(2\)
\\
& & \((8,1)\) & \(16/69\) & \(1\)
\\
& & \((7,2)\) & \(27/115\) & \(1\)
\\
& & \((6,3)\) & \(64/345\) & \(1\)
\\
& & \((4,4,1)\) & \(2/69\) & \(1\)
\\
\bottomrule
\end{tabular}
\captionsetup{hypcap=false}
\captionof{table}{Active Schur sectors, locked weights, and Clifford multiplicities
for the maximizing types at \(m=5,\ldots,9\).}
\label{tab:supp-wigner-locked-data}
\end{minipage}
\end{center}

For \(m=5\), the witness value on \(\mathcal E_5\) depends only on
the following weighted combination of the two scalar blocks:
\[
\frac7{15}K_{5,(5)}
+
\frac8{15}K_{5,(4,1)}
=
-\frac1{15}.
\]

For \(m=6\), the only active block is
\[
K_{6,(5,1)}
=
-\frac15.
\]

For \(m=7\), the three reduced blocks are
\[
K_{7,(7)}
=
-\frac1{96}\id_2,
\qquad
K_{7,(6,1)}
=
-\frac1{192},
\qquad
K_{7,(4,3)}
=
\frac1{1344}.
\]
For \(m=8\), every active sector satisfies
\[
K_{8,\lambda}
=
-\frac13\id,
\]
where the identity acts on the multiplicity space whose dimension is listed in
Table~\ref{tab:supp-wigner-locked-data}.

\emph{The nine-query matrix certificate.}
As noted in Subsec.~\ref{sec:supp-clifford-wigner-obstruction}, only the
\((9)\)-sector has a two-dimensional multiplicity space; the remaining active
sectors are multiplicity-free.  For each of the three
orbit averages in \(B_9\), combining the locked scalar contributions from these
sectors with the \((9)\)-sector contribution gives an effective Hermitian
matrix.

In a fixed orthonormal basis of this two-dimensional multiplicity space, the
matrices are
\begin{equation}
\begin{aligned}
K_{1^9}^{\mathrm A}
&=
\begin{pmatrix}
-\dfrac3{161} & 0
\\[4pt]
0 & \dfrac1{115}
\end{pmatrix},
\\[4pt]
K_{2^2\,1^5}^{\mathrm A}
&=
\begin{pmatrix}
\dfrac{1787}{3245760}
&
-\dfrac{11\sqrt{35}}{1081920}
\\[4pt]
-\dfrac{11\sqrt{35}}{1081920}
&
-\dfrac{107}{360640}
\end{pmatrix},
\\[4pt]
K_{333}^{\mathrm L}
&=
\begin{pmatrix}
\dfrac{425}{10304}
&
\dfrac{297\sqrt{35}}{51520}
\\[4pt]
\dfrac{297\sqrt{35}}{51520}
&
-\dfrac{83}{7360}
\end{pmatrix}.
\end{aligned}
\label{eq:app-supp-wigner-nine-query-matrices}
\end{equation}
These entries give the cancellation used in
Eq.~\eqref{eq:supp-clifford-nine-query-certificate}:
\[
39K_{1^9}^{\mathrm A}
+
1134K_{2^2\,1^5}^{\mathrm A}
+
2K_{333}^{\mathrm L}
=
-\frac{16}{805}\id_2.
\]
The combined effective operator is therefore proportional to the identity,
so the certificate value is independent of the multiplicity-space state.

\section*{Use of Artificial Intelligence}

The initial motivation for this work arose from the authors'
observation that the optimal type-II error for unitary subgroup
testing established by Hayashi et al.~\cite{hayashi2025predicting}
does not by itself determine the resources required to attain it.
This led to the question of whether the zero-type-I-error
optimum can be attained using only free probes and free
measurements.

In pursuing this question, the authors used large language
models to assist with literature searches and to analyze the
connection between representation-theoretic optimality
conditions and quantum resource constraints, helping to develop
and organize the approach presented here.
These tools also assisted with code development for exact
computations and certificate verification, and with the
refinement of the manuscript's mathematical exposition and
presentation.
The mathematical arguments, computational results,
interpretations, and references were checked by the authors,
who approved the final manuscript and take full responsibility
for its scientific content.

\end{document}